\documentclass{article}

\PassOptionsToPackage{}{natbib}

\usepackage[preprint]{main}
\usepackage[utf8]{inputenc}
\usepackage[T1]{fontenc}
\usepackage{hyperref}
\usepackage{url}
\usepackage{booktabs}
\usepackage{amsfonts}
\usepackage{nicefrac}
\usepackage{microtype}
\usepackage{xcolor}
\usepackage{graphicx}
\usepackage{amsmath}
\usepackage{amssymb}
\usepackage{amsthm}
\usepackage{cleveref}
\usepackage{empheq}
\usepackage{algorithmic}
\usepackage{algorithm}
\usepackage{authblk}
\usepackage{enumitem}
\usepackage{subcaption}

\theoremstyle{plain}
\newtheorem{theorem}{Theorem}[section]
\newtheorem{proposition}[theorem]{Proposition}
\newtheorem{lemma}[theorem]{Lemma}
\newtheorem{corollary}[theorem]{Corollary}
\theoremstyle{definition}

\newtheorem{assumption}[theorem]{Assumption}
\theoremstyle{remark}
\newtheorem{remark}[theorem]{Remark}

\crefname{definition}{Definition}{Definitions}
\crefname{theorem}{Theorem}{Theorems}
\crefname{assumption}{Assumption}{Assumptions}
\crefname{lemma}{Lemma}{Lemmas}
\crefname{proposition}{Proposition}{Propositions}
\crefname{remark}{Remark}{Remarks}

\crefname{section}{Section}{Sections}
\crefname{equation}{Eq.}{Eqs.}

\newcommand{\argmax}{\mathop{\rm arg~max}\limits}

\newcommand{\dd}{\mathop{}\!\mathrm{d}}

\title{
Error-Controlled Borrowing from External Data
\\
Using Wasserstein Ambiguity Sets
}

\author[$\dagger$,1]{Yui Kimura}
\author[$\dagger$,2,3,$\ast$]{Shu Tamano}
\affil[$\dagger$]{Equal contribution.}
\affil[1]{
Novartis Pharma K.K.,
1-23-1 Toranomon,
Minato-ku,
Tokyo 105-6333,
Japan
}
\affil[2]{
Department of Multidisciplinary Sciences,
Graduate School of Arts and Sciences,
The University of Tokyo,
3-8-1 Komaba,
Meguro-ku,
Tokyo 153-8902,
Japan
}
\affil[3]{
Department of Epidemiology,
National Institute of Infectious Diseases,
Japan Institute for Health Security,
1-23-1 Toyama,
Shinjuku-ku,
Tokyo 162-0052,
Japan}
\affil[$\ast$]{
Email: tamano-shu212@g.ecc.u-tokyo.ac.jp
}

\begin{document}

\maketitle

\begin{abstract}
    Incorporating external data can improve the efficiency of clinical trials, but distributional mismatches between current and external populations threaten the validity of inference.
    While numerous dynamic borrowing methods exist, the calibration of their borrowing parameters relies mainly on ad hoc, simulation-based tuning.
    To overcome this, we propose BOND (Borrowing under Optimal Nonparametric Distributional robustness), a framework that formalizes data noncommensurability through Wasserstein ambiguity sets centered at the current-trial distribution.
    By deriving sharp, closed-form bounds on the worst-case mean drift for both continuous and binary outcomes, we construct a distributionally robust, bias-corrected Wald statistic that ensures asymptotic type I error control uniformly over the ambiguity set.
    Importantly, BOND determines the optimal borrowing strength by maximizing a worst-case power proxy, converting heuristic parameter tuning into a transparent, analytically tractable optimization problem.
    Furthermore, we demonstrate that many prominent borrowing methods can be reparameterized via an effective borrowing weight, rendering our calibration framework broadly applicable.
    Simulation studies and a real-world clinical trial application confirm that BOND preserves the nominal size under unmeasured heterogeneity while achieving efficiency gains over standard borrowing methods.
\end{abstract}
\noindent
\textbf{Keywords}:
Clinical trials;
Distributionally robust optimization;
Dynamic borrowing;
Hybrid control;
Information borrowing;
Wasserstein ball;

\section{Introduction}
\label{sec:introduction}

Randomized controlled trials remain the gold standard for evaluating treatment effects, principally because randomization facilitates valid inference under minimal assumptions \citep{ich2000e10}.
However, in many contemporary development settings, recruiting a concurrently controlled arm of adequate size is practically challenging or ethically contentious.
Constraints such as small patient populations in rare diseases, slow accrual rates, and the ethical dilemma of withholding effective therapies have fueled substantial interest in leveraging historical information or external comparators to augment current trials \citep{schmidli2020beyond}.
Regulatory agencies have explicitly discussed the design and conduct of externally controlled trials \citep{fda2023externally_controlled,ema2025external_controls_concept}, and empirical surveys document how external controls have been used in regulatory submissions and decisions \citep{goring2019characteristics,jahanshahi2021use,liu2025design}.
Consequently, there is a pressing need for methodological frameworks that improve trial efficiency via borrowing while rigorously safeguarding operating characteristics.

The fundamental challenge in borrowing from an external control arm is the potential noncommensurability (or lack of exchangeability) between current and historical populations.
Discrepancies in baseline characteristics, eligibility criteria, endpoint definitions, and secular trends in standard of care can introduce distributional shifts that induce bias and inflate type I error rates.
Even modest drifts in the outcome distribution can shift the borrowing-based test statistic, leading to false rejections under the null hypothesis \citep{viele2014use,van2018including}.
While causal inference techniques, such as propensity score weighting, can adjust for observed covariate imbalance, they cannot account for residual biases arising from unmeasured confounding or unobserved outcome drifts \citep{rippin2022review}.
Therefore, a critical statistical problem is how to quantify and control for these distributional mismatches in a way that is transparent and satisfies regulatory rigor regarding error control.

A vast literature addresses adaptive borrowing, broadly categorized into three streams.
First, frequentist approaches employ test-then-pool (TTP) procedures or dynamic pooling strategies based on similarity tests \citep{viele2014use,li2020revisit,okada2024decoupling}.
Second, Bayesian dynamic borrowing methods introduce explicit parameters controlling the extent of borrowing, including power priors and their variants \citep{ibrahim2000power,ibrahim2003optimality,ibrahim2015power,pawel2023normalized}, commensurate priors \citep{hobbs2011hierarchical,hobbs2012commensurate}, and mixture-based constructions such as robust meta-analytic-predictive (MAP) priors and related exchangeability models \citep{neuenschwander2010summarizing,schmidli2014robust,neuenschwander2016robust,kaizer2018bayesian,yang2023sam,alt2024leap}.
Third, hybrid strategies combine borrowing with explicit covariate adjustment or selective borrowing mechanisms, including propensity-score-integrated priors \citep{liu2021propensity,lu2022propensity,wang2022propensity}, case-weighted power priors \citep{kwiatkowski2024case}, and adaptive hybrid control designs \citep{guo2024adaptive}.
Despite their differences, these methods share a common reliance on tunable quantities such as discount factors, precision parameters, and mixture weights that dictate the extent of borrowing.

A persistent limitation in current practice is that the calibration of these borrowing parameters remains largely heuristic.
Typically, practitioners select discount factors based on extensive simulations under a specific set of designer-specified heterogeneity scenarios \citep{pan2017calibrated,psioda2019bayesian,eggleston2021bayesctdesign,ling2021calibrated,khanal2025commensurate,demartino2025eliciting}.
While this approach assesses performance under assumed conditions, it lacks explicit guarantees against deviations outside those scenarios.
If strict type I error control is demanded against arbitrary incompatibility, the borrowing weight must trivially collapse to zero \citep{psioda2019bayesian,kopp2020power,bennett2021novel,lee2024using,gao2025control}.
To overcome this all-or-nothing tension, we require a framework that
(i) explicitly defines a tolerance set of admissible heterogeneity and
(ii) optimizes borrowing performance within that set.

In this paper, we propose BOND (Borrowing under Optimal Nonparametric Distributional robustness), a framework for integrating historical data under distributional uncertainty.
Leveraging tools from optimal transport \citep{villani2009optimal} and distributionally robust optimization (DRO) \citep{mohajerin2018data,blanchet2019quantifying,gao2023distributionally,kuhn2025distributionally}, BOND formalizes noncommensurability by defining the admissible discrepancy as a Wasserstein ball of radius $\rho$ centered at the current-arm distribution.
Under this geometric formulation, we parameterize borrowing strength via an effective borrowing weight (EBW), which unifies a broad class of estimators (see Appendix~\ref{app:sec:incl-ebw} for details), and derive tight, closed-form bounds on the worst-case mean drift for both continuous and binary outcomes.
We then construct a distributionally robust test that explicitly subtracts this worst-case bias, guaranteeing asymptotic type I error control uniformly over the ambiguity set.
Finally, BOND identifies the optimal borrowing strength by solving a minimax problem:
maximizing a worst-case power proxy subject to the robust size constraint.
Through simulations and a real-data application, we demonstrate how the radius $\rho$ replaces ad hoc tuning by acting as a transparent, analytically tractable sensitivity parameter for regulatory decision-making.

The remainder of the paper is organized as follows.
Section~\ref{sec:preliminaries} introduces the notation, the borrowing framework, and the concept of Wasserstein ambiguity sets.
Section~\ref{sec:proposed-method} details the proposed distributionally robust bias correction and the optimal weight calibration.
Section~\ref{sec:numerical-experiments} presents simulation results evaluating operating characteristics.
Section~\ref{sec:real-world-experiments} applies the method to real-world clinical data.
Section~\ref{sec:discussion} concludes with a discussion of implications, limitations, and future directions.

\section{Preliminaries}
\label{sec:preliminaries}

\subsection{Problem Setup}
\label{subsec:problem-setup}

Let $(\mathcal{X}, \mathcal{B}_{\mathcal{X}})$ be a measurable space for baseline covariates, and let $\mathcal{Y}$ denote the outcome space.
We consider two outcome types:
(i) binary outcomes with $\mathcal{Y}=\{0,1\}$, and
(ii) continuous outcomes with $\mathcal{Y}= \mathbb{R}$, assuming $\mathbb{E}[Y^2]< \infty$.

Data are observed from two sources:
a current randomized clinical trial ($j=C$) and a historical trial ($j=H$).
For subject $i$ in trial $j$, we observe the tuple
\begin{equation*}
    Z_{j,i}
    \coloneq
    (A_{j,i}, X_{j,i},Y_{j,i})
    \in
    \{0,1\}\times \mathcal{X}\times \mathcal{Y}
    ,
\end{equation*}
where $A_{j,i} \in \{0,1\}$ denotes the treatment assignment ($1$ for experimental, $0$ for control).
For each arm $a\in\{0,1\}$ and trial $j\in\{C,H\}$, let $P_{j}^a$ denote the conditional probability measure of the covariate-outcome pair $(X,Y)$:
\begin{equation*}
    P_{j}^a
    \coloneq
    \mathcal{L}\bigl(
        (X,Y)
        \bigm| A=a,j
    \bigr)
    ,
\end{equation*}
defined on the product space $\mathcal{Z} \coloneq \mathcal{X} \times \mathcal{Y}$ equipped with $\mathcal{B}_{\mathcal{X}}\otimes \mathcal{B}_{\mathcal{Y}}$.
The marginal mean outcome in arm $a$ of trial $j$ is defined as
\begin{equation*}
    \mu_j^a
    \coloneq
    \mathbb{E}_{P_j^a}[Y]
    .
\end{equation*}
For binary outcomes, this simplifies to the response probability $\mu_j^a=\mathbb{P}(Y=1\mid A=a,j)$.

The average treatment effect in the current trial on the mean-difference scale is
\begin{equation*}
    \theta_C
    \coloneq
    \mu_C^1 - \mu_C^0
    .
\end{equation*}
Similarly, the historical mean difference is $\theta_H \coloneq \mu_H^1-\mu_H^0$.

To formalize between-trial heterogeneity without relying on specific parametric assumptions, we introduce a parameter $\gamma \in \Gamma$ indexing a family of candidate historical laws $\{P_{H}^{a}(\gamma)\}_{\gamma\in\Gamma}$.
Let $\mu_H^a(\gamma) \coloneq \mathbb{E}_{P_H^a(\gamma)}[Y]$ be the mean outcome under a specific heterogeneity level $\gamma$.
The treatment effect in the historical population is denoted by
\begin{equation*}
    \theta_H(\gamma)
    \coloneq
    \mu_H^1(\gamma)-\mu_H^0(\gamma)
    .
\end{equation*}
We define the discrepancy function $\delta: \Gamma \to \mathbb{R}$ such that
\begin{equation*}
    \theta_H(\gamma)=\theta_C+\delta(\gamma)
    .
\end{equation*}
This discrepancy can be decomposed into arm-specific mean shifts.
Define the drift of the historical arm $a$ relative to the current arm as
\begin{equation*}
    \Delta_a(\gamma)
    \coloneq
    \mu_H^a(\gamma) - \mu_C^a
    ,
    \quad
    a\in\{0,1\}
    .
\end{equation*}
It follows that $\delta(\gamma) = \Delta_1(\gamma) - \Delta_0(\gamma)$.
Throughout this paper, the pair $(\Delta_0(\gamma), \Delta_1(\gamma))$ serves as the sufficient statistic for the bias induced by external information.

\begin{remark}[Extensions beyond two arms and a single historical source]
\label{rem:extensions-multiarm-multisource}
    For clarity of exposition, we focus on a binary treatment $A\in\{0,1\}$ and a single historical source.
    However, the proposed framework extends naturally to
    (i) multi-arm trials with finitely many treatment levels and
    (ii) borrowing from multiple historical datasets.
    These extensions are achieved by indexing arms and sources and replacing scalar borrowing weights with vectors.
    A formal generalization to the multi-arm/multi-source setting, including the robust bias correction and robust noncentrality parameter for arbitrary linear contrasts, is provided in Appendix~\ref{app:sec:multiarm-multisource}.
\end{remark}

\subsection{Effective Borrowing Estimators}
\label{subsec:ebw}

Let $n_j$ be the total sample size of trial $j$, with realized arm sizes
\begin{equation*}
    n_{j,a}
    \coloneq
    \sum_{i=1}^{n_j} \boldsymbol{1}\{A_{j,i}=a\}
    ,
    \quad
    a\in\{0,1\},\;
    j\in\{C,H\}
    .
\end{equation*}
We treat the sample sizes $(n_{j,a})$ as fixed or condition on them throughout.
Define the arm-specific sample mean in trial $j$ by
\begin{equation*}
    \bar{Y}_{j,a}
    \coloneq
    \frac{1}{n_{j,a}}\sum_{i:A_{j,i}=a}Y_{j,i}
    ,
    \quad
    (n_{j,a}\ge 1)
    .
\end{equation*}

We introduce a borrowing parameter
\begin{equation*}
    \lambda=
    (\lambda_0,\lambda_1)
    \in
    \Lambda
    \coloneq
    [0,\Lambda_0]\times [0,\Lambda_1]
    ,
\end{equation*}
where $\Lambda_a$ represents a maximal borrowing cap (e.g., $\Lambda_a=1$ corresponds to the weight of the full historical sample).
If only historical controls are available, one may set $n_{H,1}=0$ and fix $\lambda_1 = 0$.

The effective borrowing estimator for the mean of arm $a$ is defined as:
\begin{equation}
    \label{eq:mu-hat}
    \hat{\mu}_a(\lambda_a)
    \coloneq
    \begin{cases}
        \dfrac{n_{C,a}\bar{Y}_{C,a}+\lambda_a\, n_{H,a}\bar{Y}_{H,a}}{n_{C,a}+\lambda_a n_{H,a}}
        ,
        &
        n_{H,a}\ge 1
        ,
        \\
        \bar{Y}_{C,a}
        ,
        &
        n_{H,a} = 0
        .
    \end{cases}
\end{equation}
This formulation encompasses a wide range of Bayesian borrowing methods (see Appendix~\ref{app:sec:incl-ebw} for the details).
We define the EBW as:
\begin{equation*}
    w_a(\lambda_a)
    \coloneq
    \frac{\lambda_a n_{H,a}}{n_{C,a}+\lambda_a n_{H,a}}
    \in[0,1)
    ,
\end{equation*}
which yields the convex combination $\hat{\mu}_a(\lambda_a) = (1-w_a(\lambda_a))\bar{Y}_{C,a} + w_{a}(\lambda_a)\bar{Y}_{H,a}$.
The resulting estimator for the treatment effect is
\begin{equation*}
    \hat{\theta}(\lambda)
    \coloneq
    \hat{\mu}_1(\lambda_1) - \hat{\mu}_0(\lambda_0)
    .
\end{equation*}
Its expectation satisfies:
\begin{equation}
    \label{eq:theta-mean}
    \mathbb{E}\bigl[
        \hat{\theta}(\lambda)
    \bigr]
    =
    \theta_C + w_1(\lambda_1)\Delta_1(\gamma) - w_0(\lambda_0)\Delta_0(\gamma)
    .
\end{equation}
\eqref{eq:theta-mean} highlights that if the historical data are not perfectly commensurate (i.e., $\Delta_a \neq 0$), the borrowing induces a bias of $w_1\Delta_1(\gamma) - w_0\Delta_0(\gamma)$.
Without correction, this shift can inflate the type I error rate.

\subsection{Wasserstein Ambiguity Sets}
\label{subsec:wasserstein}

To quantify distributional differences between current and historical arms without committing to a parametric model, we employ Wasserstein ambiguity sets \citep{villani2009optimal,mohajerin2018data,blanchet2019quantifying,gao2023distributionally}.
Let $(\mathcal{X}, d_{\mathcal{X}})$ be a metric space for covariates, and equip $\mathcal{Y}$ with the Euclidean distance.
On the product space $\mathcal{Z} = \mathcal{X} \times \mathcal{Y}$, we define the ground metric:
\begin{equation}
    \label{eq:ground-metric}
    d\bigl(
        (x,y),(x',y')
    \bigr)
    \coloneq
    d_{\mathcal{X}}(x,x')
    +
    |y-y'|
    .
\end{equation}

\begin{remark}
\label{rem:metric-choice}
    The additive structure of the metric in \eqref{eq:ground-metric} is deliberate.
    First, it ensures that the outcome mapping $(x,y) \mapsto y$ is $1$-Lipschitz, which is essential for deriving sharp bounds on the mean shift.
    Second, it interprets the cost of transport as a sum of covariate mismatch and outcome drift.
    While other metrics (e.g., $L_p$-combinations) are possible, the additive form offers a clear interpretation where a unit shift in outcome $Y$ contributes directly to the transport cost.
\end{remark}

Let $\mathcal{P}_1(\mathcal{Z})$ denote the set of Borel probability measures on $\mathcal{Z}$ with finite first moment with respect to $d$.
The $1$-Wasserstein distance between two measures $P, Q \in \mathcal{P}_1(\mathcal{Z})$ is defined as:
\begin{equation*}
    W_1(P,Q)
    \coloneq
    \inf_{\pi\in\Pi(P,Q)}
    \int_{\mathcal{Z}\times \mathcal{Z}} d(z,z')
    \pi(\dd z, \dd z')
    ,
\end{equation*}
where $\Pi(P,Q)$ denotes the set of couplings with marginals $P$ and $Q$.
For a specified radius $\rho_a\ge 0$, we define the arm-specific Wasserstein ambiguity set (or ball) centered at $P_C^a$ by
\begin{equation}
    \label{eq:Ua}
    \mathcal{U}_a(\rho_a)
    \coloneq
    \bigl\{
        Q\in\mathcal{P}_1(\mathcal{Z})\colon
        W_1(Q,P_C^a)\le \rho_a
    \bigr\}
    .
\end{equation}
The condition $P_H^a \in \mathcal{U}_a(\rho_a)$ formalizes the assumption that the historical distribution drifts from the current distribution by at most $\rho_a$ in terms of the Wasserstein distance.
Within this set, we identify the worst-case mean shifts:
\begin{equation*}
    \Delta_a^+(\rho_a)
    \coloneq
    \sup_{Q\in\mathcal{U}_a(\rho_a)}\bigl(
        \mathbb{E}_Q[Y] - \mu_C^a
    \bigr)
    ,
    \quad
    \Delta_a^{-}(\rho_a)
    \coloneq
    \inf_{Q\in\mathcal{U}_a(\rho_a)}\bigl(
        \mathbb{E}_Q[Y] - \mu_C^a
    \bigr)
    .
\end{equation*}
These bounds, $\Delta_a^+$ and $\Delta_a^-$, represent the maximal positive and negative bias feasible under the constraint that the historical data are $\rho_a$-compatible with the current trial.
\section{Proposed Method}
\label{sec:proposed-method}

We refer to the proposed distributionally robust calibration-and-testing procedure as BOND.

\subsection{Worst-Case Mean Shifts over Wasserstein Balls}
\label{subsec:worst-mean}

The core of our proposal is to robustify the borrowing estimator against distributional shifts.
The first analytical step is to determine the worst-case expectation of the outcome $Y$ within the Wasserstein ambiguity set $\mathcal{U}_{a}(\rho_a)$.
Although this is generally an infinite-dimensional optimization problem, the specific structure of $W_1$-transport cost with the ground metric \eqref{eq:ground-metric} allows us to derive sharp, closed-form bounds \citep{mohajerin2018data,blanchet2019quantifying,kuhn2025distributionally}.

\begin{proposition}[Closed-form worst-case mean shifts]
\label{prop:closed-form-delta}
    Fix an arm $a\in\{0,1\}$ and let $\mathcal{U}_a(\rho_a)$ be defined by \eqref{eq:Ua} under the metric \eqref{eq:ground-metric}.
    Assume $P_C^a\in \mathcal{P}_1(\mathcal{Z})$.

    \begin{enumerate}
        \item[(i)] (Continuous outcome)
        If $\mathcal{Y} = \mathbb{R}$, then
        \begin{equation*}
            \Delta_a^+(\rho_a) = \rho_a
            ,
            \quad
            \Delta_a^-(\rho_a) = -\rho_a
            .
        \end{equation*}

        \item[(ii)] (Binary outcome)
        If $\mathcal{Y} = \{0,1\}$, then
        \begin{equation*}
            \sup_{Q\in\mathcal{U}_a(\rho_a)}\mathbb{E}_Q[Y]
            =
            \min\{\mu_C^a+\rho_a,1\}
            ,
            \quad
            \inf_{Q\in\mathcal{U}_a(\rho_a)}\mathbb{E}_Q[Y]
            =
            \max\{\mu_C^a-\rho_a,0\}
            ,
        \end{equation*}
        equivalently,
        \begin{equation*}
            \Delta_a^+(\rho_a)
            =
            \min\{\rho_a,1-\mu_C^a\}
            ,
            \quad
            \Delta_a^-(\rho_a)
            =
            -\min\{\rho_a,\mu_C^a\}
            .
        \end{equation*}
    \end{enumerate}
\end{proposition}
See Appendix~\ref{app:pf-prop-closed-form-delta} for the proof.
Proposition~\ref{prop:closed-form-delta} turns an infinite-dimensional DRO problem over a Wasserstein ball into an explicit, closed-form bound on the arm-wise mean drift.
As a result, the robust bias correction can be computed by simple arithmetic (no optimization solver is needed), which enables fast calibration of the borrowing parameters and transparent sensitivity interpretation of the radius $\rho_a$.

\subsection{Robust Bias Correction and Test Definition}
\label{subsec:test}

We consider the one-sided hypothesis
\begin{equation*}
    H_0\colon \theta_C\le 0
    \quad
    \text{vs.}
    \quad
    H_1\colon \theta_C >0
    .
\end{equation*}

Let $z_{1-\alpha}$ denote the $(1-\alpha)$ quantile of the standard normal distribution.
For each $\lambda\in\Lambda$, the worst-case bias in the rejection direction is:
\begin{equation}
    \label{eq:bplus-def}
    b_+(\lambda)
    \coloneq
    \sup_{\substack{Q_1\in\mathcal{U}_1(\rho_1)\\ Q_0\in\mathcal{U}_0(\rho_0)}}\Bigl[
        w_1(\lambda_1)\bigl(
            \mathbb{E}_{Q_1}[Y] - \mu_C^1
        \bigr)
        -
        w_0(\lambda_0)\bigl(
            \mathbb{E}_{Q_0}[Y] - \mu_C^0
        \bigr)
    \Bigr]
    .
\end{equation}

\begin{proposition}[Closed-form of $b_+(\lambda)$]
\label{prop:bplus}
    For any $\lambda\in\Lambda$ with $w_a(\lambda_a) \ge 0$,
    \begin{equation*}
        b_+(\lambda)
        =
        w_1(\lambda_1)\Delta_1^+(\rho_1)
        -
        w_0(\lambda_0)\Delta_0^-(\rho_0)
        .
    \end{equation*}
    In particular, under Proposition~\ref{prop:closed-form-delta},
    \begin{equation*}
        b_+(\lambda)=
        \begin{cases}
            w_{1}(\lambda_1)\rho_1+w_0(\lambda_0)\rho_0
            ,
            &
            \mathcal{Y} = \mathbb{R}
            ,
            \\
            w_1(\lambda_1)\min\{\rho_1,1-\mu_C^1\}+w_0(\lambda_0)\min\{\rho_0, \mu_C^0\},
            &
            \mathcal{Y} = \{0,1\}
            .
        \end{cases}
    \end{equation*}
\end{proposition}
See Appendix~\ref{app:pf-prop-bplus} for the proof.
Proposition~\ref{prop:bplus} shows that the worst-case bias in the rejection direction decomposes arm-by-arm and depends on $\lambda$ only through the EBW $w_a(\lambda_a)$.
Combined with Proposition~\ref{prop:closed-form-delta}, this yields an explicit $b_+(\lambda)$, making the robust test and subsequent $\lambda$-calibration computationally trivial to evaluate over $\Lambda$.

We define the asymptotic variance of $\hat{\theta}(\lambda)$ as:
\begin{equation}
    \label{eq:s-lambda}
    s^2(\lambda)
    \coloneq
    \mathrm{Var}\bigl(
        \hat{\theta}(\lambda)
    \bigr)
    =
    \sum_{a\in\{0,1\}}\Biggl[
        \bigl(1-w_{a}(\lambda_a)\bigr)^2
        \frac{\sigma_{C,a}^2}{n_{C,a}}
        +
        w_a(\lambda_a)^2\frac{\sigma_{H,a}^2}{n_{H,a}}
    \Biggr]
    ,
\end{equation}
where $\sigma_{j,a}^2\coloneq \mathrm{Var}(Y\mid A=a,j)$ and the convention is that the term $w_a(\lambda_a)^2\sigma_{H,a}^2/n_{H,a}$ is set to $0$ if $n_{H,a}=0$.

In practice, we estimate \eqref{eq:s-lambda} via the pooled plug-in estimator.
For $(j,a)$ with $n_{j,a} \ge 2$, let
\begin{equation*}
    \hat{\sigma}_{j,a}^2
    \coloneq
    \frac{1}{n_{j,a} - 1}\sum_{i: A_{j,i}=a}\bigl(
        Y_{j,i} - \bar{Y}_{j,a}
    \bigr)^2
    ,
\end{equation*}
and define
\begin{equation}
    \label{eq:shat}
    \hat{s}^2(\lambda)
    \coloneq
    \sum_{a\in\{0,1\}}
    \Biggl[
        \bigl(1-w_a(\lambda_a)\bigr)^2
        \frac{\hat{\sigma}_{C,a}^2}{n_{C,a}}
        +
        w_a(\lambda_a)^2\frac{\hat{\sigma}_{H,a}^2}{n_{H,a}}
    \Biggr]
    ,
    \quad
    \hat{s}(\lambda)
    \coloneq
    \sqrt{\hat{s}^2(\lambda)}
    .
\end{equation}

For binary outcomes, $b_+(\lambda)$ depends on the unknown $\mu_C^a$.
A natural implementation replaces $\mu_C^a$ by $\bar{Y}_{C,a}$ in Proposition~\ref{prop:bplus};
denote the resulting plug-in bias by $\hat{b}_{+}(\lambda)$.
For continuous outcomes, $b_+(\lambda)$ depends only on $(\rho_0,\rho_1)$ and $(w_0,w_1)$.

We propose the distributionally robust Wald test
\begin{equation}
    \label{eq:robust-test}
    \varphi_{\lambda}
    \coloneq
    \boldsymbol{1}\Biggl\{
        \frac{\hat{\theta}(\lambda)-\tilde{b}_{+}(\lambda)}{\hat{s}(\lambda)}\ge z_{1-\alpha}
    \Biggr\}
    ,
\end{equation}
where $\tilde{b}_+(\lambda) = b_{+}(\lambda)$ if $b_{+}(\lambda)$ is treated as known theoretical benchmark and $\tilde{b}_+(\lambda) = \hat{b}_+(\lambda)$ for the practical plug-in version.

\begin{remark}[Two-sided extension]
\label{rem:two-sided}
    The main text focuses on the one-sided hypothesis $H_0\colon\theta_C\le 0$.
    A two-sided test for
    \begin{equation*}        
        H_0^{\pm}\colon\theta_C=0,
        \quad
        H_1^{\pm}\colon\theta_C\neq 0
    \end{equation*}
    is obtained by introducing the worst-case bias in the negative rejection direction,
    \begin{equation*}        
        b_-(\lambda)
        \coloneq
        \inf_{\substack{Q_1\in\mathcal{U}_1(\rho_1)\\ Q_0\in\mathcal{U}_0(\rho_0)}}
        \Bigl[
            w_1(\lambda_1)\{
                \mathbb{E}_{Q_1}[Y]-\mu_C^1
            \}
            -
            w_0(\lambda_0)\{
                \mathbb{E}_{Q_0}[Y]-\mu_C^0
            \}
        \Bigr]
        ,
    \end{equation*}
    and rejecting when either tail is significant:
    \begin{equation*}        
        \varphi^{\pm}_{\lambda}
        \coloneq
        \boldsymbol{1}\Biggl\{
            \frac{\hat{\theta}(\lambda)-\tilde{b}_{+}(\lambda)}{\hat{s}(\lambda)}
            \ge z_{1-\alpha/2}
            \;\text{ or }\;
            \frac{\hat{\theta}(\lambda)-\tilde{b}_{-}(\lambda)}{\hat{s}(\lambda)}
            \le -z_{1-\alpha/2}
        \Biggr\}
        .
    \end{equation*}
    Here $\tilde{b}_-(\lambda)$ is the benchmark $b_-(\lambda)$ or its plug-in version for binary outcomes.
    The closed-form of $b_-(\lambda)$ and the full proofs of distributionally robust size control are given in Appendix~\ref{app:sec:two-sided}.
\end{remark}

\subsection{Asymptotic Guarantees and Optimal Calibration}
\label{subsec:asymptotics}
We introduce a minimal asymptotic framework.

\begin{assumption}[Sampling and moments]
\label{ass:sampling}
    For each $(j,a)$ with $n_{j,a} \ge 1$, the outcomes $\{Y_{j,i}\colon A_{j,i}=a\}$ are i.i.d.\ with mean $\mu_j^a$ and variance $\sigma_{j,a}^2<\infty$.
    Moreover, the collections from different $(j,a)$ are mutually independent.
\end{assumption}

\begin{assumption}[Asymptotic regime and nondegeneracy]
\label{ass:asymptotic}
    For each arm $a$, $n_{C,a}\to \infty$ and either $n_{H,a}\to \infty$ or $n_{H,a} = 0$.
    Additionally, $\sigma_{C,a}^2 > 0$ for $a\in\{0,1\}$.
\end{assumption}

\begin{proposition}[Asymptotic normality]
\label{prop:clt}
    Under Assumptions~\ref{ass:sampling} and \ref{ass:asymptotic}, for any fixed $\lambda\in\Lambda$,
    \begin{equation*}
        \frac{\hat{\theta}(\lambda) - \bigl(
            \theta_C+w_1(\lambda_1)\Delta_1 - w_0(\lambda_0)\Delta_0
        \bigr)}{s(\lambda)}
        \longrightarrow_d
        N(0,1)
        ,
    \end{equation*}
    where $s(\lambda)$ is defined in \eqref{eq:s-lambda}.
\end{proposition}
See Appendix~\ref{app:pf-prop-clt} for the proof.
Proposition~\ref{prop:clt} provides a Gaussian approximation for the borrowing estimator $\hat{\theta}(\lambda)$ with an explicit centering term that isolates the external-data bias component.
This normal limit justifies the Wald-type construction in \eqref{eq:robust-test} and is the key step that allows us to derive analytic size and power characterizations for each fixed $\lambda$.

We now state the robust size guarantee, defined with respect to the Wasserstein ambiguity sets.

\begin{theorem}[Asymptotic distributionally robust size control]
\label{thm:robust-size}
    Fix $\lambda\in\Lambda$.
    For any fixed $(\theta_C,P_H^0,P_H^1)$ with $\theta_C\le 0$ and $P_H^a\in\mathcal{U}_a(\rho_a)$, $a\in\{0,1\}$,
    under Assumptions~\ref{ass:sampling} and \ref{ass:asymptotic},
    \begin{equation*}
        \limsup_{\min_a n_{C,a}\to\infty}
        \mathbb{P}(
            \varphi_\lambda = 1
        )
        \le \alpha
        ,
    \end{equation*}
    where the probability is taken under the joint law induced by the i.i.d.\ sampling from $(P_C^0, P_C^1,P_H^0,P_H^1)$.
\end{theorem}
See Appendix~\ref{app:pf-thm-robust-size} for the proof.
Theorem~\ref{thm:robust-size} guarantees asymptotic type I error control uniformly over all historical-arm distributions lying in the Wasserstein ambiguity sets.
Thus, for any prespecified borrowing rule $\lambda$, the proposed bias correction provides a principled safeguard against false positives induced by lack of commensurability between current and external data.

\begin{proposition}[Tightness and minimality of the robust correction]
\label{prop:tight-min}
    Fix $\lambda\in\Lambda$ and consider the one-sided Wald-type test family
    \begin{equation*}        
        \varphi_{\lambda,c}
        \coloneq
        \boldsymbol{1}\Biggl\{
            \frac{\hat\theta(\lambda)-c}{\hat{s}(\lambda)}\ge z_{1-\alpha}
        \Biggr\}
        ,
        \quad
        c\in\mathbb{R}
        .
    \end{equation*}
    \begin{enumerate}
        \item[(i)] (Minimality) If $c<b_+(\lambda)$, there exists a null configuration $(\theta_C, P_H^0, P_H^1)$ with $\theta_C=0$ and $P_H^a\in\mathcal{U}_a(\rho_a)$
        \begin{equation*}        
            \liminf_{\min_a n_{C,a}\to\infty}
            \mathbb{P}(\varphi_{\lambda,c}=1)
            >
            \alpha
            .
        \end{equation*}
        Consequently, $b_+(\lambda)$ is the minimal constant correction required to guarantee distributionally robust size control over the joint ambiguity set $\mathcal{U}_0(\rho_0)\times\mathcal{U}_1(\rho_1)$.
        
        \item[(ii)] (Tightness) For $c=b_+(\lambda)$, the bound in Theorem~\ref{thm:robust-size} is tight in the minimax sense:
        \begin{equation*}        
            \sup_{\theta_C\le 0,P_H^a\in\mathcal{U}_a(\rho_a)}
            \lim_{\min_a n_{C,a}\to\infty}\mathbb{P}(\varphi_{\lambda}=1)
            =
            \alpha
            .
        \end{equation*}
    \end{enumerate}
\end{proposition}
See Appendix~\ref{app:pf-prop-tight-min} for the proof.
Proposition~\ref{prop:tight-min} shows that the proposed robust correction $b_+(\lambda)$ is not only sufficient but also necessary (within the class of constant-shift Wald tests) for distributionally robust size control over $\mathcal{U}_0(\rho_0)\times\mathcal{U}_1(\rho_1)$.
In particular, any smaller correction $c<b_+(\lambda)$ leads to asymptotic type I error inflation under some admissible historical drift, while the choice $c=b_+(\lambda)$ attains the nominal level $\alpha$ under an explicit least-favorable configuration, implying that Theorem~\ref{thm:robust-size} is minimax-sharp and cannot be further tightened without strengthening assumptions or enlarging the test class.

To characterize power, fix a target alternative effect $\theta_1 > 0$ and define the robust power
\begin{equation*}
    \mathrm{Pow}_{\mathrm{rob}}(\lambda;\theta_1)
    \coloneq
    \inf_{P_H^1\in\mathcal{U}_1(\rho_1), P_H^0\in\mathcal{U}_0(\rho_0)}
    \mathbb{P}_{\theta_C=\theta_1}\bigl(
        \varphi_{\lambda} = 1
    \bigr)
    .
\end{equation*}

\begin{theorem}[Asymptotic robust power and the robust noncentrality parameter]
\label{thm:robust-power}
    Fix $\theta_1 > 0$ and $\lambda\in \Lambda$.
    Under Assumptions~\ref{ass:sampling} and \ref{ass:asymptotic},
    \begin{equation*}
        \lim_{\min_{a}n_{C,a}\to \infty} \mathrm{Pow}_{\mathrm{rob}}(\lambda;\theta_1)
        =
        1-\Phi\bigl(
            z_{1-\alpha} - \kappa(\lambda)
        \bigr)
        ,
    \end{equation*}
    where $\Phi$ is the standard normal cumulative distribution function and
    \begin{equation}
    \label{eq:kappa}
        \kappa(\lambda)
        \coloneq
        \frac{\theta_1 - w_1(\lambda_1)\bigl(
            \Delta_1^+(\rho_1) - \Delta_1^-(\rho_1)
        \bigr)
        -
        w_0(\lambda_0)\bigl(
            \Delta_0^+(\rho_0) - \Delta_0^-(\rho_0)
        \bigr)}{s(\lambda)}
        .
    \end{equation}
\end{theorem}
See Appendix~\ref{app:pf-thm-robust-power} for the proof.
Theorem~\ref{thm:robust-power} yields an explicit asymptotic lower bound on power in terms of the robust noncentrality parameter $\kappa(\lambda)$.
This converts the choice of the borrowing parameter from ad-hoc scenario-based simulation to a direct, analytically tractable optimization problem: maximize $\kappa(\lambda)$ (equivalently robust power) while preserving worst-case size control.

The test \eqref{eq:robust-test} controls worst-case type I error asymptotically for each fixed $\lambda$.
Therefore, it is natural to select $\lambda$ by maximizing the robust power lower bound.

\begin{corollary}[Robust-power optimal borrowing weight exists]
\label{cor:lambda-opt}
    Assume $\Lambda=[0,\Lambda_0]\times [0,\Lambda_1]$ is compact.
    Under Assumptions~\ref{ass:sampling} and \ref{ass:asymptotic}, the map $\lambda\mapsto \kappa(\lambda)$ in \eqref{eq:kappa} is continuous on $\Lambda$.
    Hence, there exists at least one maximizer
    \begin{equation*}
        \lambda^\ast
        \in
        \argmax_{\lambda\in\Lambda} \mathrm{Pow}_{\mathrm{rob}}(\lambda;\theta_1)
        =
        \argmax_{\lambda\in\Lambda}\kappa(\lambda)
        .
    \end{equation*}
\end{corollary}
See Appendix~\ref{app:pf-cor-lambda-opt} for the proof.
Corollary~\ref{cor:lambda-opt} ensures that the proposed DRO-based calibration problem is well-posed: an optimal borrowing parameter $\lambda^\ast$ exists on the prespecified feasible set $\Lambda$.
This guarantees that the method produces a concrete, implementable recommendation rather than only a conceptual criterion.

\begin{remark}
    The optimization in Corollary~\ref{cor:lambda-opt} formalizes DRO-based calibration of the borrowing degree:
    we maximize a worst-case (over Wasserstein ambiguity sets) power functional while maintaining worst-case type I error control (Theorem~\ref{thm:robust-size}).
    This contrasts with common ad-hoc calibrations of discounting parameters by simulation for specific scenarios, as in many implementations of power priors, commensurate priors, and MAP/robust MAP priors.
\end{remark}

\subsection{Implementation and Algorithm}
\label{subsec:algorithm}

This subsection summarizes the fully implementable BOND procedure, which calibrates the borrowing parameter $\lambda$ by maximizing a plug-in version of the robust noncentrality parameter, followed by the robust Wald test at the selected $\hat{\lambda}$.
Importantly, the procedure requires only arm-level summary statistics $(n_{j,a},\bar{Y}_{j,a},\hat{\sigma}^2_{j,a})$, making it directly applicable to aggregate historical data.

\subsubsection{Plug-In Robust Power Criterion}
\label{subsubsec:plug-in}
Recall the robust noncentrality parameter $\kappa(\lambda)$ in \eqref{eq:kappa}.
Let $D_a(\rho_a) \coloneq \Delta_a^+(\rho_a)-\Delta_a^-(\rho_a)$ denote the arm-wise worst-case mean-drift range.
By Proposition~\ref{prop:closed-form-delta}, this simplifies to
\begin{equation*}
    D_a(\rho_a)
    =
    \begin{cases}
        2\rho_a,
        & \mathcal{Y}=\mathbb{R}
        ,\\
        \min\{\rho_a,1-\mu_C^a\}+\min\{\rho_a,\mu_C^a\},
        & \mathcal{Y}=\{0,1\}
        .
    \end{cases}
\end{equation*}
For the latter, we use the plug-in estimator $\hat{D}_a(\rho_a)$ obtained by replacing $\mu_C^a$ with $\bar{Y}_{C,a}$.

We define the plug-in robust noncentrality parameter as
\begin{equation}
\label{eq:kappa-hat}
    \hat{\kappa}(\lambda)
    \coloneq
    \frac{
        \theta_1
        - w_1(\lambda_1)\hat{D}_1(\rho_1)
        - w_0(\lambda_0)\hat{D}_0(\rho_0)
    }{\hat{s}(\lambda)}
    ,
\end{equation}
where $\hat{s}(\lambda)$ is given in \eqref{eq:shat}.
The corresponding robust power proxy,
\(
\widehat{\mathrm{Pow}}_{\mathrm{rob}}(\lambda;\theta_1)
\coloneq
1-\Phi(z_{1-\alpha}-\hat{\kappa}(\lambda))
\)
, is strictly monotone in $\hat{\kappa}(\lambda)$.

Because $\Lambda$ is compact and $\hat{\kappa}(\lambda)$ is continuous whenever $\hat{s}(\lambda)>0$, the maximizer set is non-empty.
To ensure a single-valued selection, we adopt a deterministic tie-breaking rule $\mathrm{Sel}(\cdot)$ that selects the maximizer with the smallest Euclidean norm, breaking remaining ties lexicographically.
The data-driven borrowing parameter is then:
\begin{equation}
\label{eq:lambda-hat-def}
    \hat{\lambda}
    \coloneq
    \mathrm{Sel}\Bigl(
        \argmax_{\lambda\in\Lambda}\hat{\kappa}(\lambda)
    \Bigr)
    .
\end{equation}
If only historical controls are available ($n_{H,1}=0$), we fix $\lambda_1=0$ and optimize over $\lambda_0\in[0,\Lambda_0]$.
Algorithm~\ref{alg:dro-calibration} details the complete end-to-end procedure.

\begin{algorithm}[t]
\caption{BOND: Borrowing under Optimal Nonparametric Distributional robustness}
\label{alg:dro-calibration}
\begin{algorithmic}[1]
    \REQUIRE
    Significance level $\alpha\in(0,1)$;
    target alternative effect $\theta_1>0$;
    radii $(\rho_0,\rho_1)$;
    feasible set $\Lambda=[0,\Lambda_0]\times[0,\Lambda_1]$;
    arm-level summaries $\{(n_{j,a},\bar{Y}_{j,a},\hat{\sigma}^2_{j,a})\}_{j\in\{C,H\},a\in\{0,1\}}$ (convention: $n_{H,a}=0$ if unavailable).

    \ENSURE
    Selected parameter $\hat{\lambda}$, test statistic $T(\hat{\lambda})$, and decision $\varphi_{\hat{\lambda}}\in\{0,1\}$.
    
    \STATE
    Compute EBWs $w_a(\lambda_a)=\lambda_a n_{H,a}/(n_{C,a}+\lambda_a n_{H,a})$ for $a\in\{0,1\}$ (set $w_a(\lambda_a)=0$ if $n_{H,a}=0$).
    
    \STATE
    For each $\lambda\in\Lambda$, compute $\hat{\mu}_a(\lambda_a)$ via \eqref{eq:mu-hat}, $\hat{\theta}(\lambda)=\hat{\mu}_1(\lambda_1)-\hat{\mu}_0(\lambda_0)$, and $\hat{s}(\lambda)$ via \eqref{eq:shat}.

    \STATE For each $\lambda\in\Lambda$, compute the bias correction $\tilde{b}_+(\lambda)$:
    \begin{equation*}        
        \tilde{b}_+(\lambda)=
        \begin{cases}
            w_1(\lambda_1)\rho_1+w_0(\lambda_0)\rho_0,
            & \mathcal{Y}=\mathbb{R}
            ,\\
            w_1(\lambda_1)\min\{\rho_1,1-\bar{Y}_{C,1}\}
            +w_0(\lambda_0)\min\{\rho_0,\bar{Y}_{C,0}\},
            & \mathcal{Y}=\{0,1\}
            .
        \end{cases}
    \end{equation*}

    \STATE Compute $\hat{\kappa}(\lambda)$ via \eqref{eq:kappa-hat} and select $\hat{\lambda}$ via \eqref{eq:lambda-hat-def}.
    
    \STATE Compute the robust Wald statistic and output the decision:
    \begin{equation*}        
        T(\hat{\lambda})
        \coloneq
        \frac{\hat{\theta}(\hat{\lambda})-\tilde{b}_+(\hat{\lambda})}{\hat{s}(\hat{\lambda})}
        ,
        \quad
        \varphi_{\hat{\lambda}}
        \coloneq
        \boldsymbol{1}\{T(\hat{\lambda})\ge z_{1-\alpha}\}
        .
    \end{equation*}
\end{algorithmic}
\end{algorithm}

\subsubsection{Asymptotic Validity of Data-Driven Calibration}
\label{subsubsec:adaptive-validity}

While Theorem~\ref{thm:robust-size} establishes size control for a fixed $\lambda$, Algorithm~\ref{alg:dro-calibration} selects $\hat{\lambda}$ adaptively.
We now formally justify that this data-driven test $\varphi_{\hat{\lambda}}$ retains asymptotic type I error control.

Let $a_n \coloneq \sqrt{n_C+n_H} \to \infty$.
Maximizing $\hat{\kappa}(\lambda)$ is equivalent to maximizing the rescaled objective $\hat{\bar{\kappa}}(\lambda) \coloneq a_n^{-1}\hat{\kappa}(\lambda)$.
Let $\bar{\kappa}(\lambda)$ denote its population counterpart, and define the theoretical optimal parameter $\lambda^\ast \coloneq \mathrm{Sel}(\argmax_{\lambda\in\Lambda}\bar{\kappa}(\lambda))$.

\begin{assumption}[Well-separated maximizer]
\label{ass:adaptive-identification}
    The population optimizer $\lambda^\ast$ is well-separated:
    for every $\varepsilon>0$, there exists $\delta_\varepsilon>0$ such that for all sufficiently large $n$, $\sup_{\|\lambda-\lambda^\ast\|_2\ge \varepsilon} \bar{\kappa}(\lambda) \le \bar{\kappa}(\lambda^\ast)-\delta_\varepsilon$.
\end{assumption}

\begin{proposition}[Robust size control with adaptive $\hat{\lambda}$]
\label{prop:adaptive-robust-size}
    Suppose Assumptions~\ref{ass:sampling}--\ref{ass:asymptotic} and \ref{ass:adaptive-identification} hold.
    Let $\hat{\lambda}$ be selected via \eqref{eq:lambda-hat-def}.
    Then, for any null configuration $(\theta_C,P_H^0,P_H^1)$ with $\theta_C\le 0$ and $P_H^a\in\mathcal{U}_a(\rho_a)$, the adaptive test
    \(
    \varphi_{\hat{\lambda}}
    \coloneq
    \boldsymbol{1}\{
        (\hat{\theta}(\hat{\lambda})-b_{+}(\hat{\lambda}))/\hat{s}(\hat{\lambda})\ge z_{1-\alpha}
    \}
    \) satisfies:
    \begin{equation*}
        \limsup_{\min_a n_{C,a}\to\infty}
        \mathbb{P}\bigl(
            \varphi_{\hat{\lambda}}=1
        \bigr)
        \le
        \alpha
        .
    \end{equation*}
\end{proposition}
See Appendix~\ref{app:pf-prop-adaptive-robust-size} for the proof.
Proposition~\ref{prop:adaptive-robust-size} formally justifies the data-driven calibration step, guaranteeing that practitioners can adaptively optimize borrowing weights to maximize power without inflating the worst-case type I error.

\begin{remark}[Safeguards for adaptive implementation]
\label{rem:adaptive-safeguards}
    Assumption~\ref{ass:adaptive-identification} generically holds when $\bar{\kappa}(\lambda)$ has a unique maximizer.
    If one wishes to avoid relying on uniqueness, a deterministic vanishing regularizer (e.g., subtracting $\eta_n\|\lambda\|_2^2$ where $\eta_n\to\infty$ but $\eta_n/a_n\to 0$) can enforce stability without asymptotically altering the objective.
    Furthermore, for binary outcomes, Proposition~\ref{prop:adaptive-robust-size} utilizes the oracle correction $b_+(\lambda)$ dependent on true $\mu_C^a$.
    A conservative, fully data-driven implementation can rely on the universal bounds $\Delta_a^+(\rho_a)\le \rho_a$ and $\Delta_a^-(\rho_a)\ge -\rho_a$, yielding 
    \(
    b_+^{\mathrm{univ}}(\lambda) \coloneq w_1(\lambda_1)\rho_1+w_0(\lambda_0)\rho_0 \ge b_+(\lambda)
    \).
    Alternatively, the plug-in correction in Algorithm~\ref{alg:dro-calibration} is exact with probability tending to one under the interior condition $\rho_a<\min\{\mu_C^a,1-\mu_C^a\}$.
\end{remark}

\section{Numerical Experiments}
\label{sec:numerical-experiments}

\subsection{Experimental Setup}
\label{subsec:experimental-setup}

We conducted extensive numerical experiments to evaluate the finite-sample operating characteristics of BOND against a suite of standard borrowing rules.
We simulated a current randomized trial ($n_C=200$) and a historical dataset ($n_H=500$) with baseline covariates $X \in \mathbb{R}^2$ for both continuous and binary outcomes.
To assess robustness against between-trial noncommensurability, we varied a scalar heterogeneity index $\gamma$ ranging from $0$ to $2$ in increments of $0.1$.

We considered three data-generating cases representing common mechanisms of incompatibility:
(i) Commensurate (no covariate shift, no drift; historical data include both arms),
(ii) Covariate shift + effect modification (shifted historical covariates with treatment-covariate interaction inducing $\gamma$-dependent marginal effect differences),
and (iii) Control drift (historical control-only) (historical controls only with a control-arm mean drift of magnitude $\gamma$).

We compared BOND against representative frequentist and Bayesian borrowing paradigms implemented in the accompanying code, including Current-only (no borrowing), Naive pooling (full borrowing), fixed-weight EBW rules, the power prior \citep{ibrahim2000power}, the commensurate prior \citep{hobbs2012commensurate}, and robust MAP priors \citep{schmidli2014robust}.
For all methods, we tested the one-sided hypothesis $H_0\colon \theta_C \le 0$ at level $\alpha=0.025$ using standardized Wald-type statistics.

For BOND, we evaluated two radius specifications:
an oracle radius (true $\rho_a = |\mu_H^a - \mu_C^a|$) and a data-driven proxy $\hat{\rho}_a = c\widehat{W}_1(\widehat{P}_{C}^a, \widehat{P}_{H}^a)$ with an inflation multiplier $c=1.5$.
Complete details of the data-generating mechanisms, hyperparameter settings, and comprehensive operating characteristic curves for all methods, including TTP and other borrowing methods, are provided in Appendix~\ref{app:sec:detailed-numerical-experiments}.

\subsection{Results}
\label{subsec:num-results}

\subsubsection{Adaptive Borrowing Behavior}
\label{subsub:adaptive-borrowing-result}

The proposed DRO calibration transparently modulates borrowing based on the specified tolerance for discrepancy.
Figure~\ref{fig:sim-dro-lambda} visualizes the borrowing levels selected by BOND under the data-driven radii for continuous outcomes.
When true noncommensurability is introduced (under the Covariate shift + effect modification and Control drift scenarios), BOND exhibits a sharp, adaptive switching behavior:
it permits near-full borrowing at $\gamma=0$, but rapidly attenuates the historical weight to zero as $\gamma$ increases beyond $0.1$, successfully guarding against bias.

\begin{figure}[tb]
\vskip 0.1in
\begin{center}
    \centerline{
    \includegraphics[width=\columnwidth]{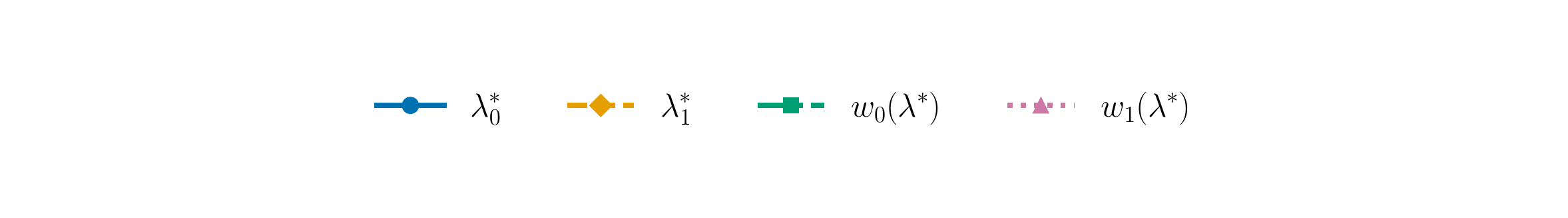}
    }
    \centerline{
    \includegraphics[width=0.49\columnwidth]{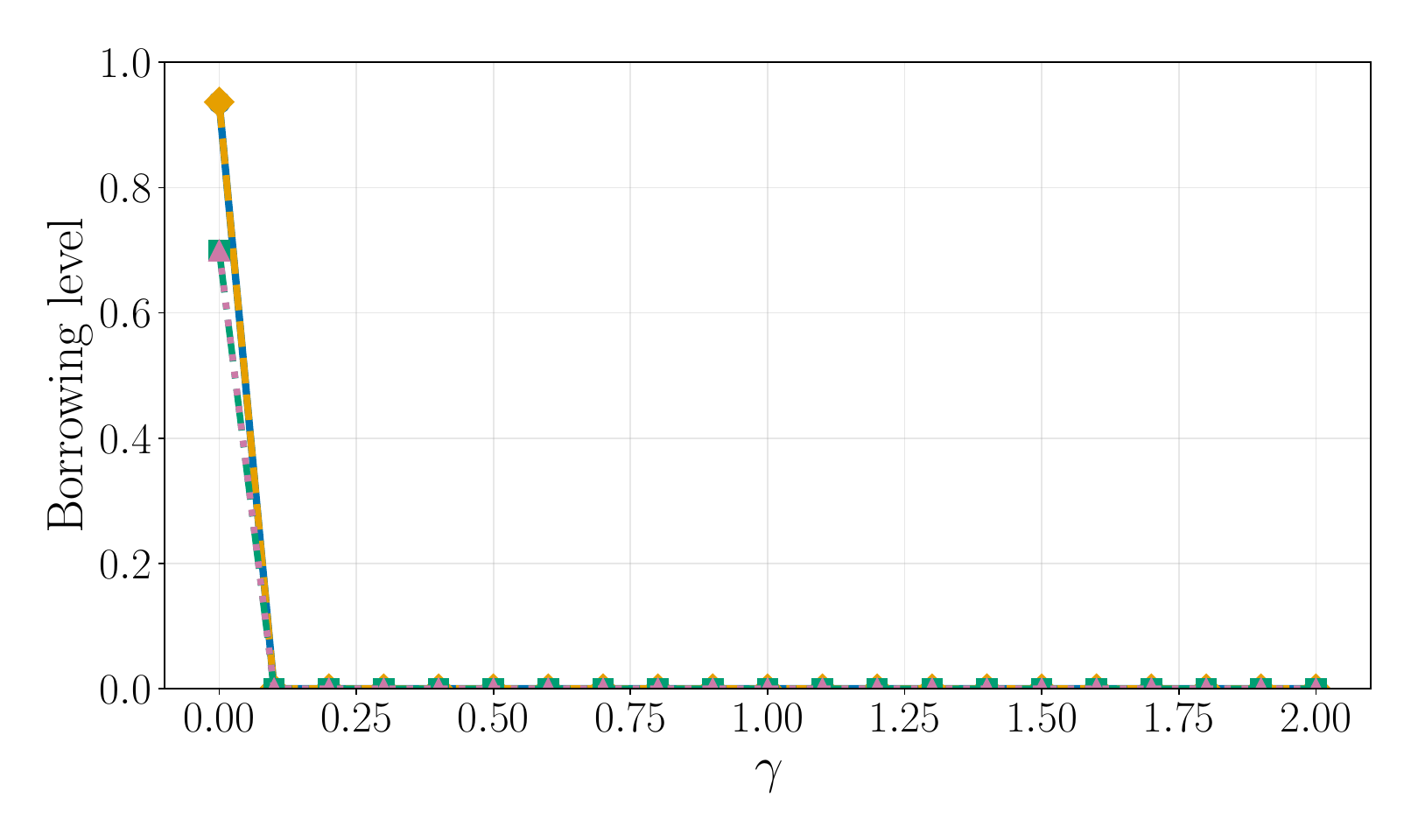}
    \hfill
    \includegraphics[width=0.49\columnwidth]{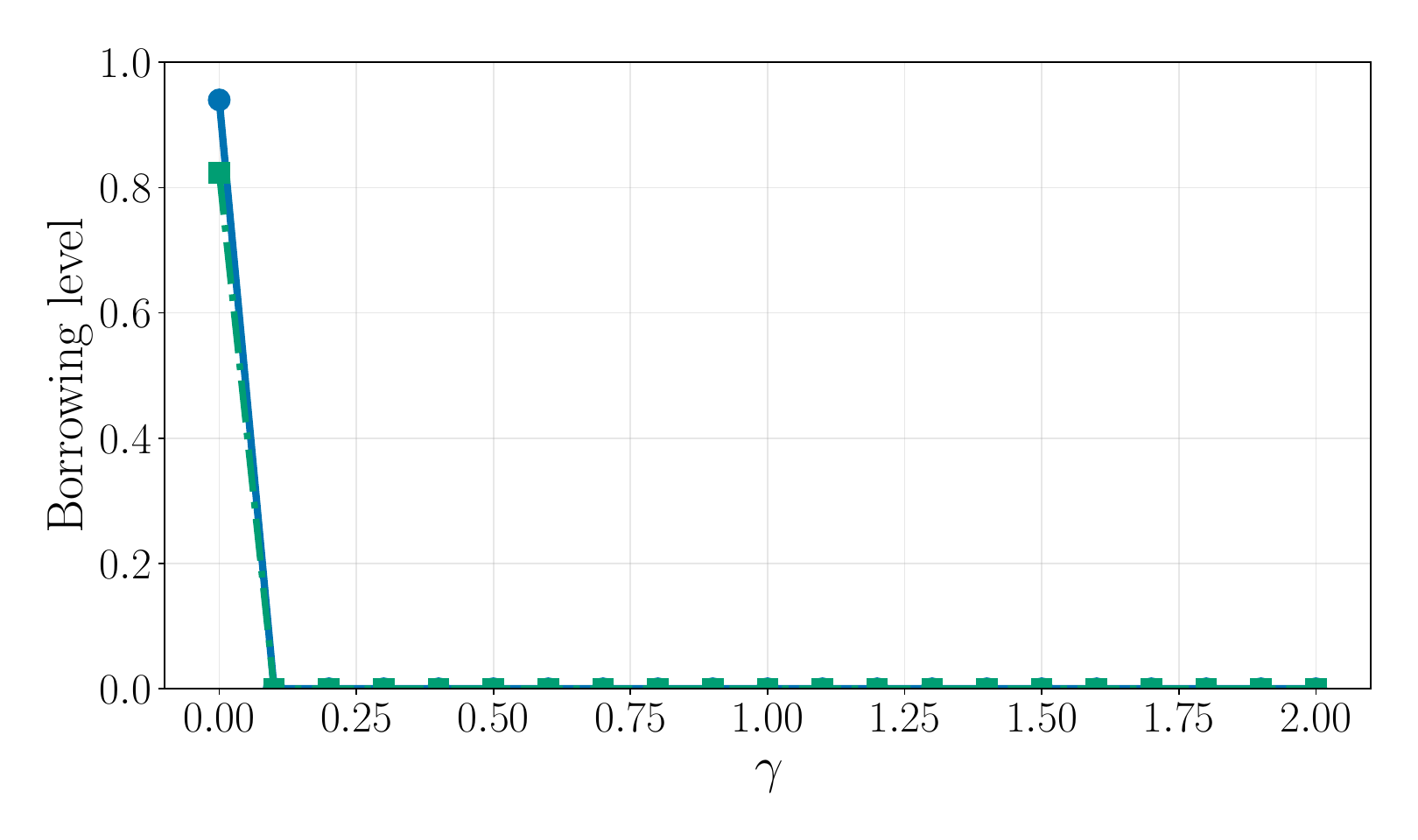}
    }
    \caption{
    BOND-calibrated borrowing levels versus $\gamma$ for continuous outcomes under data-driven radii ($c=1.5$).
    Each panel plots the calibrated discount factors $(\lambda_0^\ast,\lambda_1^\ast)$ and the implied EBW $(w_0(\lambda_0^\ast),w_1(\lambda_1^\ast))$.
    Left: covariate shift + effect modification. Right: control drift with historical controls only (so $\lambda_1^\ast\equiv 0$).
    }
\label{fig:sim-dro-lambda}
\end{center}
\vskip -0.1in
\end{figure}

\subsubsection{Type I Error and Power Trade-offs}
\label{subsub:type-i-result}

Table~\ref{tab:sim-worst-case} summarizes the worst-case operating characteristics over the $\gamma$ grid for continuous outcomes, and Figure~\ref{fig:sim-type1-power-cont} illustrates these continuous dynamics across the full range of heterogeneity.
Aggressive strategies (e.g., Naive pooling, fixed $\lambda=0.5$, or standard power priors) achieve high power under the Commensurate scenario but suffer catastrophic type I error inflation ($\approx 1.000$) as heterogeneity grows under the Covariate shift + effect modification scenario, rendering them unacceptable for regulatory purposes.
In contrast, Figure~\ref{fig:sim-type1-power-cont} (left panels) demonstrates that BOND, robust MAP, and the commensurate prior successfully maintain strict type I error control even as $\gamma$ increases.

Importantly, BOND achieves this robustness without sacrificing utility.
In the Commensurate scenario, BOND boosts the worst-case power to $0.773$, nearly doubling the efficiency of the Current-only design ($0.402$).
Furthermore, under the Control drift scenario (Figure~\ref{fig:sim-type1-power-cont}, right panels), several dynamic priors become somewhat conservative and suffer a severe power collapse due to borrowing historical controls in the wrong direction.
However, BOND consistently controls the type I error near the nominal level while preserving substantially higher power than these competing robust methods, reverting to the Current-only baseline power ($0.400$) by adaptively setting $\lambda_0^\ast \approx 0$.

\begin{table}[tb]
\centering
\small
\caption{
Worst-case operating characteristics over $\gamma$ ranging from $0$ to $2$ in increments of $0.1$ for continuous outcomes under data-driven radii ($c=1.5$).
We report the maximum empirical type I error and the minimum empirical power for a focused subset of methods.
}
\label{tab:sim-worst-case}
\begin{tabular}{lcc cc cc}
\hline
& \multicolumn{2}{c}{Commensurate} & \multicolumn{2}{c}{Covariate shift + Mod.} & \multicolumn{2}{c}{Control drift (historical control-only)} \\
Method
& $\max_\gamma \widehat{\mathrm{TypeI}}$ & $\min_\gamma \widehat{\mathrm{Power}}$
& $\max_\gamma \widehat{\mathrm{TypeI}}$ & $\min_\gamma \widehat{\mathrm{Power}}$
& $\max_\gamma \widehat{\mathrm{TypeI}}$ & $\min_\gamma \widehat{\mathrm{Power}}$ \\
\hline
Current-only
& $0.029$ & $0.402$ & $0.027$ & $0.326$ & $0.028$ & $0.400$ \\

Naive pooling
& $0.028$ & $0.892$ & $1.000$ & $0.826$ & $0.026$ & $0.000$ \\

Fixed $\lambda=0.5$
& $0.028$ & $0.854$ & $1.000$ & $0.777$ & $0.026$ & $0.000$ \\

Power prior ($\lambda=0.5$)
& $0.013$ & $0.765$ & $0.984$ & $0.669$ & $0.022$ & $0.000$ \\

Commensurate prior ($\tau=1$)
& $0.028$ & $0.406$ & $0.042$ & $0.352$ & $0.027$ & $0.348$ \\

Robust MAP ($\epsilon=0.2$)
& $0.012$ & $0.736$ & $0.081$ & $0.144$ & $0.028$ & $0.026$ \\

BOND (data-driven)
& $0.025$ & $0.773$ & $0.027$ & $0.326$ & $0.028$ & $0.400$ \\
\hline
\end{tabular}
\end{table}

\begin{figure}[tb]
\vskip 0.1in
\begin{center}
    \centerline{
    \includegraphics[width=\columnwidth]{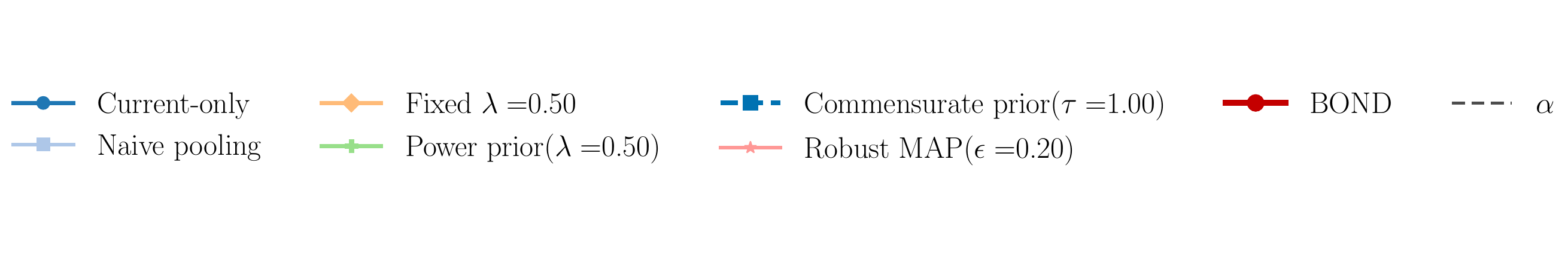}
    }
    \centerline{
    \includegraphics[width=0.49\columnwidth]{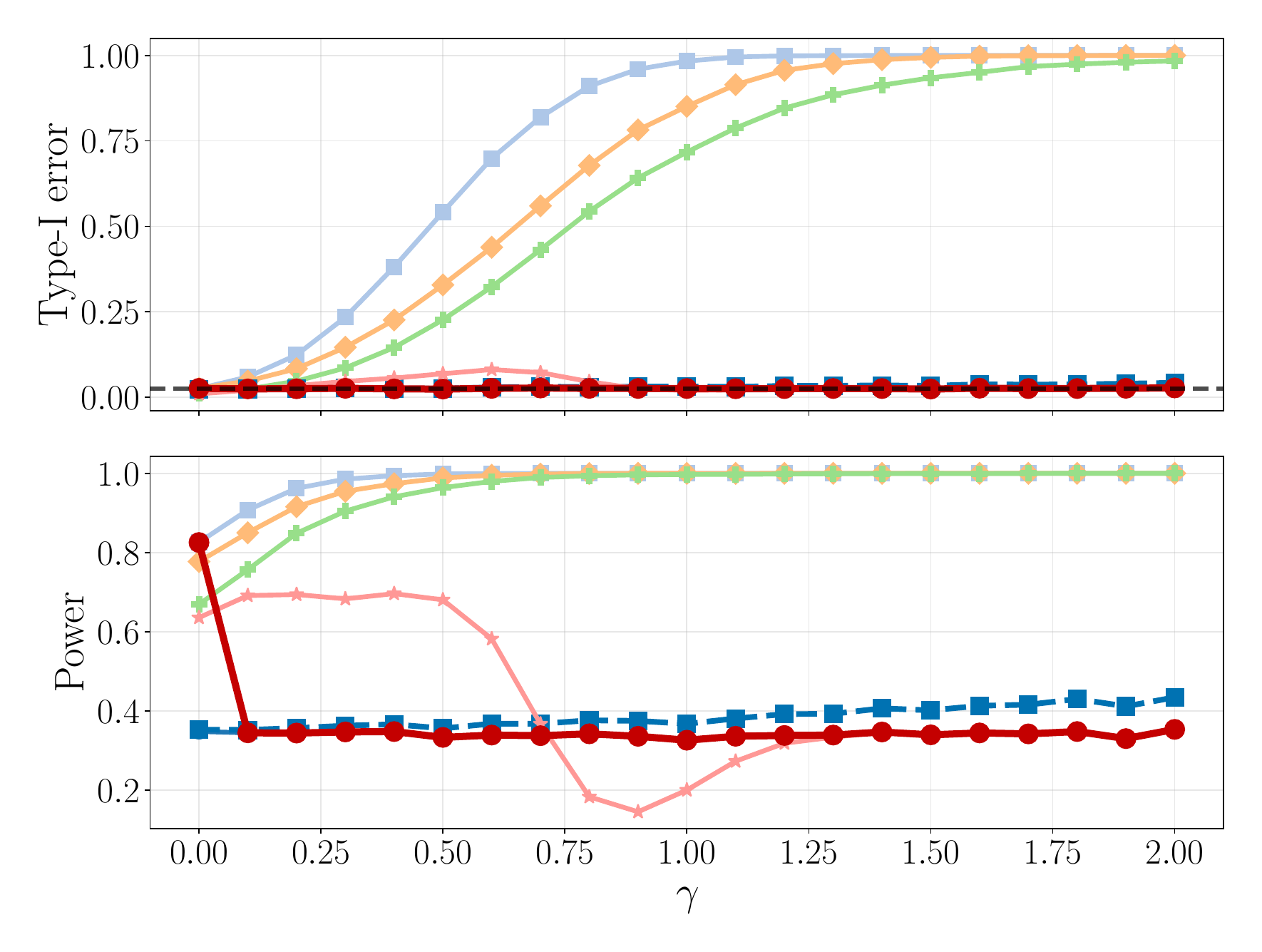}
    \hfill
    \includegraphics[width=0.49\columnwidth]{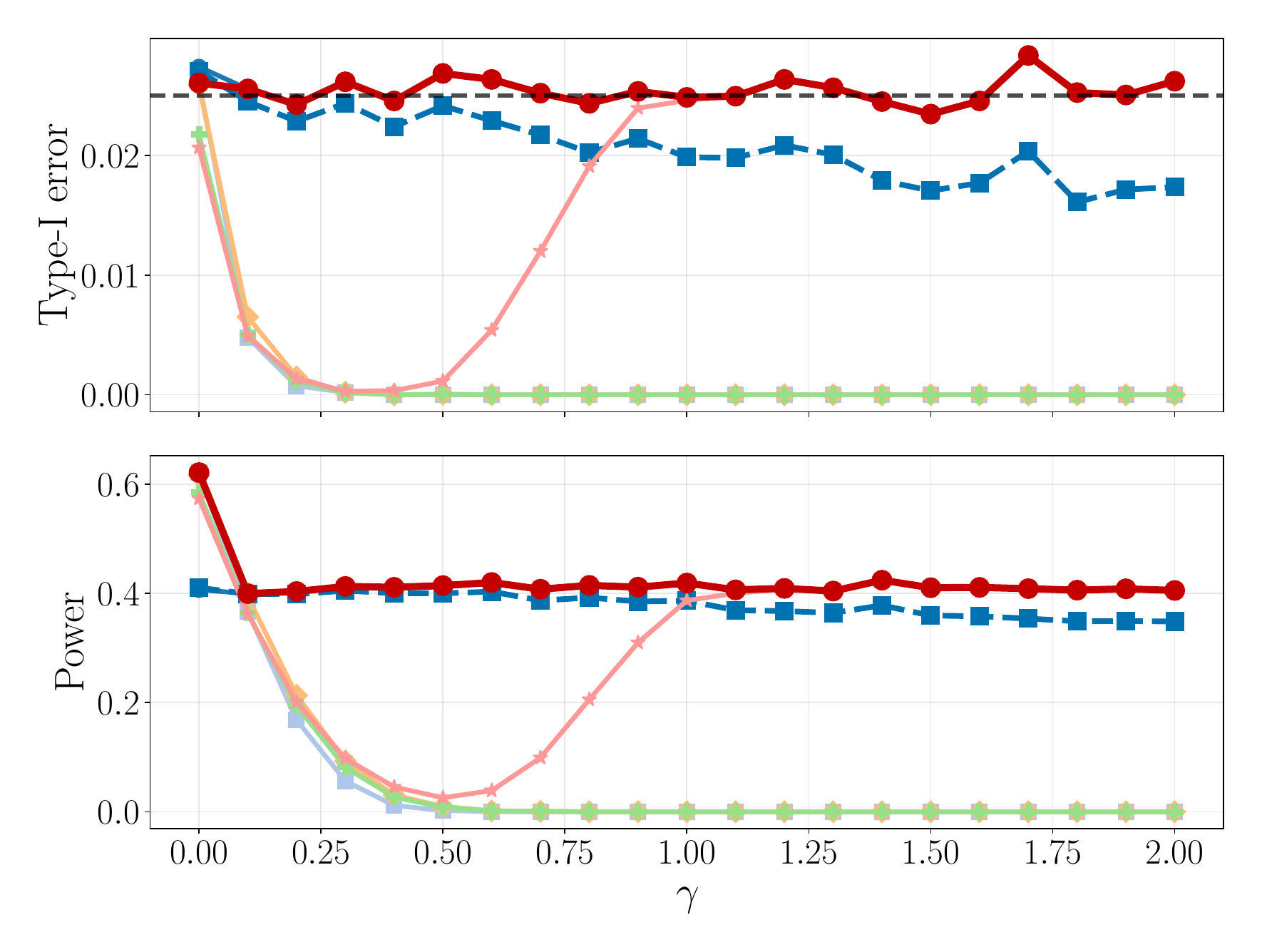}
    }
    \caption{
    Empirical type I error (top) and power (bottom) versus $\gamma$ for continuous outcomes under data-driven radii ($c=1.5$).
    Left: covariate shift + effect modification. Right: control drift with historical controls only (so $\lambda_1^\ast\equiv 0$).
    }
\label{fig:sim-type1-power-cont}
\end{center}
\vskip -0.1in
\end{figure}

Exhaustive graphical profiles and tables for all scenarios, methods, endpoints, and radius specifications (including oracle radii) are provided in Appendix~\ref{app:subsec:detailed-numerical-results}.
\section{Real-World Data Experiments}
\label{sec:real-world-experiments}

\subsection{Real-World Dataset}
\label{subsec:rw-dataset}

We illustrate the proposed method using aggregate data from two randomized studies in previously treated metastatic colorectal cancer (mCRC).
The goal is to evaluate the treatment effect of panitumumab plus FOLFIRI versus FOLFIRI alone on the binary objective response rate (ORR), utilizing the panitumumab trial \citep{peeters2010randomized} as the current study ($j=C$).
We borrow external control information from the placebo plus FOLFIRI arm of the VELOUR study \citep{tabernero2014aflibercept} ($j=H$).
Because the experimental regimens differ, we calibrate borrowing exclusively for the control arm ($\lambda_1=0$, $\lambda_0 \in [0, \Lambda_0]$).

This dataset presents a severe stress test for dynamic borrowing:
the observed ORR in the historical control is $0.367$, markedly higher than the $0.128$ observed in the current control.
This absolute gap of $0.239$ implies substantial noncommensurability (e.g., secular trends or unmeasured prognostic differences).

\subsection{Results and Interpretation}
\label{subsec:realworld-results}

Table~\ref{tab:realworld-results-focus} compares BOND (at $\rho_0=0$) against standard baselines.
The Current-only analysis demonstrates a highly significant treatment effect ($\hat{\theta}=0.156,\;p < 0.001$).
However, Naive pooling forcefully incorporates the highly discordant historical controls, artificially inflating the estimated control ORR to $0.263$.
This severely attenuates the estimated treatment effect to $\hat{\theta}=0.022$, completely destroying statistical significance ($p=0.186$) despite yielding a narrower confidence interval.
Fixed-weight priors ($\lambda=0.5$) similarly dilute the effect.
Conversely, adaptive robust priors (Commensurate, Robust MAP) detect the massive conflict, effectively shutting off borrowing ($n_{\mathrm{hist}}^{\mathrm{eff}} \approx 0$) and recovering the Current-only result, but failing to provide any efficiency gain.
Unlike these all-or-nothing adaptive methods, BOND evaluated at $\rho_0=0$ proactively optimizes for maximum power by borrowing aggressively ($n_{\mathrm{hist}}^{\mathrm{eff}}=294$).
As shown in Table~\ref{tab:realworld-results-focus}, although this inevitably attenuates the point estimate ($\hat{\theta}=0.065$), BOND uniquely secures the narrowest confidence interval (width ratio $0.930$) while successfully preserving statistical significance ($p=0.004$), thus demonstrating an optimal balance between efficiency and validity.

\begin{table}[t]
\centering
\caption{
    Real-world ORR analysis (representative methods).
    We report the estimated control response $\hat{\mu}_0$, estimated treatment effect $\hat{\theta}=\hat{\mu}_1-\hat{\mu}_0$, the 95\% interval width relative to Current-only, the effective borrowed historical control sample size $n_{\mathrm{hist}}^{\mathrm{eff}}$, and the one-sided $p$-value for $H_0\colon\theta_C\le 0$.
}
\label{tab:realworld-results-focus}
\setlength{\tabcolsep}{6pt}
\begin{tabular}{lrrrrr}
\toprule
Method & $\hat{\mu}_0$ & $\hat{\theta}$ & Width ratio & $n_{\mathrm{hist}}^{\mathrm{eff}}$ & $p$ \\
\midrule
Current-only
& $0.128$ & $0.156$ & $1.000$ & $0$ & $7.7\times 10^{-10}$ \\

Naive pooling
& $0.263$ & $0.022$ & $0.946$ & $610$ & $0.186$ \\

Fixed $\lambda=0.5$
& $0.220$ & $0.063$ & $0.930$ & $305$ & $0.005$ \\

Power prior ($\lambda=0.5$)
& $0.220$ & $0.062$ & $0.989$ & $305$ & $0.008$ \\

Commensurate prior ($\tau=1$)
& $0.132$ & $0.152$ & $1.004$ & $4$ & $2.2\times 10^{-9}$ \\

Robust MAP ($\epsilon=0.2$)
& $0.130$ & $0.155$ & $1.001$ & $0$ & $7.8\times 10^{-10}$ \\

BOND ($\rho_0=0$)
& $0.220$ & $0.065$ & $0.930$ & $294$ & $0.004$ \\
\bottomrule
\end{tabular}
\end{table}

\begin{figure*}[t]
\centering
    \includegraphics[width=0.95\textwidth]{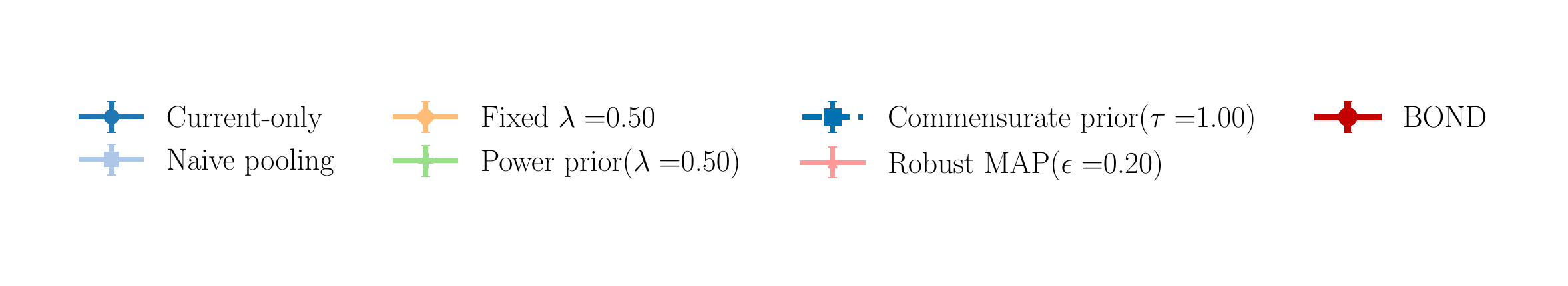}
    \begin{minipage}{0.49\textwidth}
    \centering
    \includegraphics[width=\linewidth]{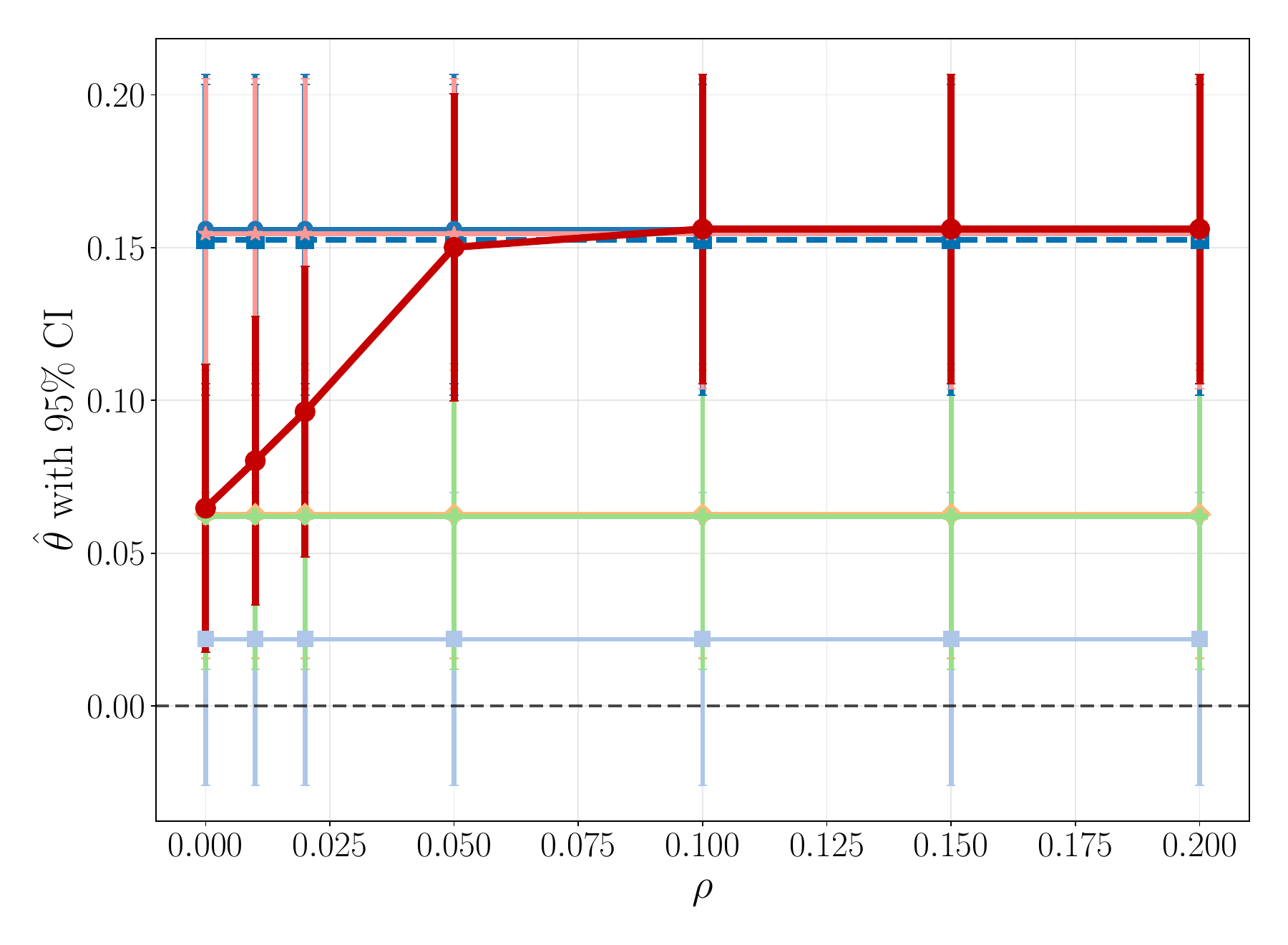}
    \end{minipage}\hfill
    \begin{minipage}{0.49\textwidth}
    \centering
    \includegraphics[width=\linewidth]{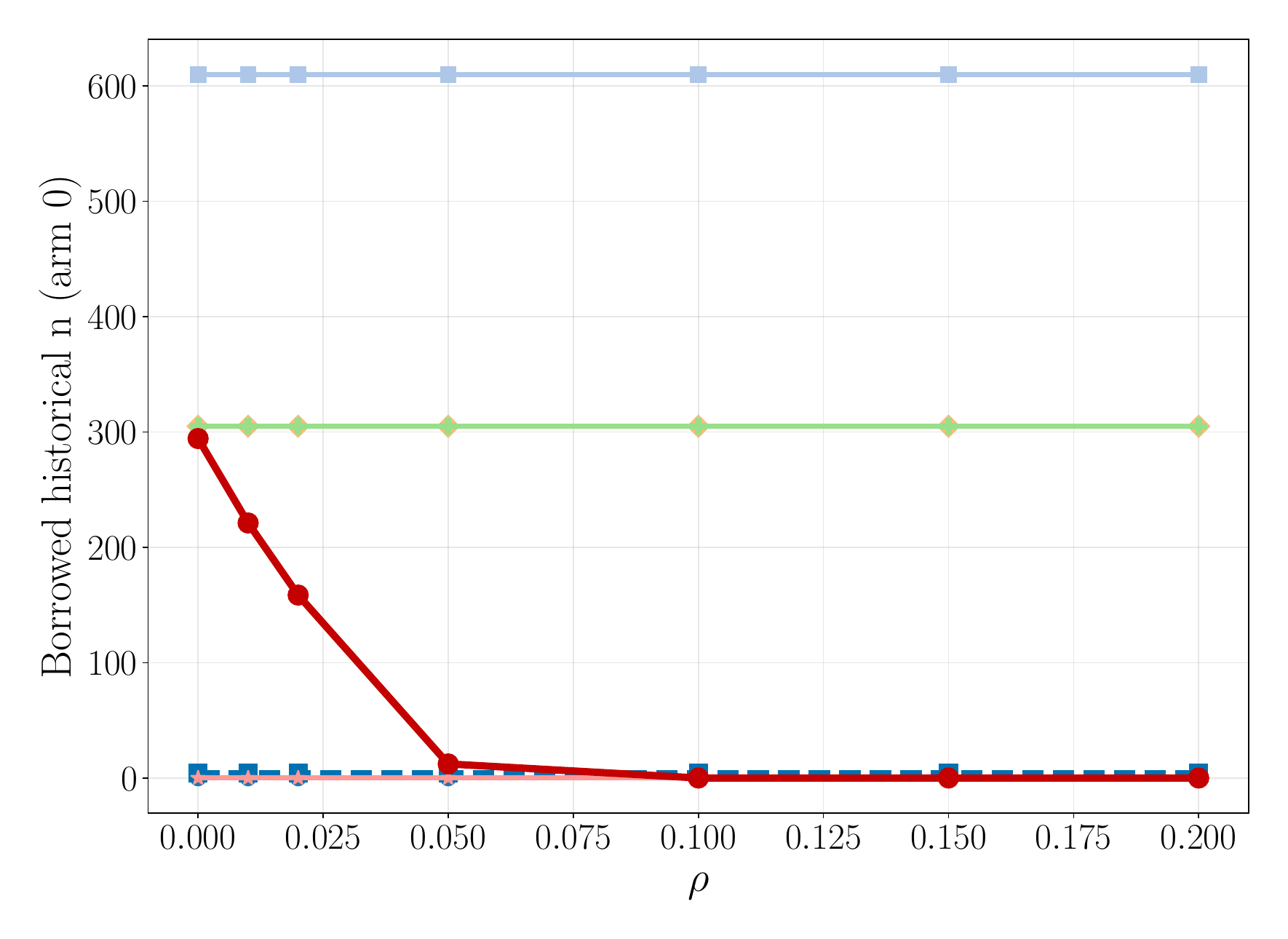}
    \end{minipage}
\\[-0.6em]
\caption{
    Real-world sensitivity to $\rho_0$ for BOND vs. baselines.
    Left: Estimated treatment effect $\hat{\theta}$ with 95\% robust confidence intervals (CIs) versus the tolerance radius $\rho_0$. 
    Right: Effective borrowed historical control sample size ($n_{\mathrm{hist}}^{\mathrm{eff}}$) versus $\rho_0$.
}
\label{fig:realworld-focus-rho}
\end{figure*}

Figure~\ref{fig:realworld-focus-rho} illustrates the unique value of the BOND framework via sensitivity analysis over $\rho_0$.
At $\rho_0=0$, BOND optimizes for maximum power, borrowing aggressively ($n_{\mathrm{hist}}^{\mathrm{eff}}=294$) and yielding a tight, but attenuated, estimate ($\hat{\theta}=0.065$).
As the tolerance for unmeasured drift increases (e.g., $\rho_0 \to 0.05$), BOND mathematically acknowledges the inflated worst-case bias and automatically attenuates the borrowing weight to protect the type I error ($n_{\mathrm{hist}}^{\mathrm{eff}}$ drops to $12$).
The effect estimate correspondingly climbs back to $0.150$.
For $\rho_0 \ge 0.10$, BOND recognizes that the potential bias overwhelms any variance reduction, setting $\lambda_0^\ast=0$ and perfectly reverting to the Current-only analysis. 

This analysis demonstrates that BOND provides a principled, continuous mechanism to negotiate the trade-off between efficiency gains from historical data and rigorous protection against population drift, translating abstract prior tuning into a clinically meaningful discussion over $\rho_0$.
Complete numerical results for this application are detailed in Appendix~\ref{app:sec:additional-real-world}.
\section{Discussion}
\label{sec:discussion}

In this paper, we proposed BOND, a distributionally robust framework for calibrating the borrowing of information from external arms.
The central motivation of this work addresses a persistent gap in the design of clinical trials that borrow external data:
while numerous Bayesian and frequentist methods exist to facilitate borrowing, the determination of the borrowing intensity (e.g., power prior exponents or mixture weights) has largely relied on ad hoc, scenario-based simulations.
BOND replaces this heuristic tuning with a principled optimization procedure.
By modeling the noncommensurability between current and historical data as a drift within a Wasserstein ambiguity set, we derived a sharp, closed-form bias correction that guarantees asymptotic type I error control.
Furthermore, by maximizing a robust noncentrality parameter subject to this error constraint, BOND uniquely identifies an optimal borrowing weight, thereby converting the calibration problem into a transparent trade-off between robustness and efficiency.

A distinguishing feature of our approach is its generality through the EBW representation.
As demonstrated in Appendix~\ref{app:sec:incl-ebw}, a wide class of dynamic borrowing methods, including power priors, commensurate priors, and robust MAP priors, can be characterized by an implied arm-specific weight.
Consequently, BOND acts not merely as a standalone estimator but as a universal robustness wrapper.
Practitioners can utilize their preferred approach for estimation while employing our DRO-based criterion to calibrate the hyperparameters (e.g., the precision parameter in commensurate priors).
This bridges the gap between flexible modeling and the rigorous error rate control required for regulatory decision-making.

Our framework shifts the focus of sensitivity analysis from the borrowing parameter $\lambda$, which lacks a direct physical interpretation, to the ambiguity radius $\rho$.
In standard practice, selecting $\lambda=0.5$ implies a specific sample size discount, but it says little about the assumed population differences.
In contrast, $\rho$ directly quantifies the tolerance for distributional divergence (e.g., the maximum admissible difference in response rates or mean outcomes).
This geometric perspective aligns well with clinical reasoning; regulators and sponsors can debate the plausible magnitude of outcome drift (the radius $\rho$) rather than the abstract mathematical weight of historical data.
The closed-form solutions we derived for both continuous and binary outcomes ensure that these operating characteristics can be computed instantaneously without complex numerical solvers, facilitating real-time sensitivity checks across a range of radii.

Several directions for future research emerge from this framework.
First, while we focused on binary and continuous outcomes, extending BOND to time-to-event endpoints is of high practical value, particularly for oncology trials.
This would require extending the Wasserstein bound derivation to hazard rates or survival functions, potentially utilizing martingale-based concentration inequalities.
Second, the current framework utilizes summary statistics to bound the mean shift.
When individual patient data are available, incorporating covariate-adjusted estimators (such as distributionally robust inverse probability weighting) would allow for a more granular handling of observed covariate shifts versus unobserved residual bias.
Finally, while we treated the radius $\rho$ as a fixed sensitivity parameter, data-driven methods to estimate the ambiguity set size, for example by using empirical Wasserstein distances with appropriate confidence inflation, warrant further theoretical development to ensure valid post-selection inference.

In conclusion, borrowing information from external sources invariably incurs a risk of bias due to distributional heterogeneity.
Rather than ignoring this risk or managing it through arbitrary discounting, BOND formalizes it via DRO.
By providing a mathematically rigorous way to borrow under uncertainty, this method offers a viable path toward more efficient and defensible clinical trial designs in the era of real-world evidence.

\section*{Code Availability}
\label{sec:code-availability}
The Python implementation of the proposed method and simulation experiments in this study are available at \href{https://github.com/shutech2001/bond-experiments}{\texttt{https://github.com/shutech2001/bond-experiments}}.
\section*{Acknowledgements}
This article is based on research using information obtained from \href{https://data.projectdatasphere.org/projectdatasphere/html/home}{\texttt{www.projectdatasphere.org}}, which is maintained by Project Data Sphere, LLC. Neither Project Data Sphere, LLC nor the owner(s) of any information from the website have contributed to, approved, or are in any way responsible for the contents of this article.
Shu Tamano was supported by JSPS KAKENHI Grant Numbers 25K24203.
\section*{Conflict of Interest}
Yui Kimura is an employee of Novartis Pharma K.K.
This work was conducted independently and outside the scope of the author’s employment.
No financial support or other funding was received from the company for this study.

\bibliographystyle{apalike}
\bibliography{bibliography}

@article{ibrahim2000power,
    author    = {Ibrahim, J.~G. and Chen, M.-H.},
    journal   = {Statistical Science},
    number    = {1},
    pages     = {46--60},
    publisher = {JSTOR},
    title     = {Power prior distributions for regression models},
    volume    = {15},
    year      = {2000}
}

@article{ibrahim2003optimality,
    author    = {Ibrahim, J.~G. and Chen, M.-H. and Sinha, D.},
    journal   = {Journal of the American Statistical Association},
    number    = {461},
    pages     = {204--213},
    publisher = {Taylor \& Francis},
    title     = {On optimality properties of the power prior},
    volume    = {98},
    year      = {2003}
}

@article{ibrahim2015power,
    author    = {Ibrahim, J.~G. and Chen, M.-H. and Gwon, Y. and Chen, F.},
    journal   = {Statistics in Medicine},
    number    = {28},
    pages     = {3724--3749},
    publisher = {Wiley Online Library},
    title     = {The power prior: Theory and applications},
    volume    = {34},
    year      = {2015}
}

@article{pawel2023normalized,
    author    = {Pawel, S. and Aust, F. and Held, L. and Wagenmakers, E.-J.},
    journal   = {Stat},
    number    = {1},
    pages     = {e591},
    publisher = {Wiley Online Library},
    title     = {Normalized power priors always discount historical data},
    volume    = {12},
    year      = {2023}
}

@article{lu2022propensity,
    author    = {Lu, N. and Wang, C. and Chen, W.-C. and Li, H. and Song, C. and Tiwari, R. and Xu, Y. and Yue, L.~Q.},
    journal   = {Journal of Biopharmaceutical Statistics},
    number    = {1},
    pages     = {158--169},
    publisher = {Taylor \& Francis},
    title     = {Propensity score-integrated power prior approach for augmenting the control arm of a randomized controlled trial by incorporating multiple external data sources},
    volume    = {32},
    year      = {2022}
}

@article{kwiatkowski2024case,
    author    = {Kwiatkowski, E. and Zhu, J. and Li, X. and Pang, H. and Lieberman, G. and Psioda, M.~A.},
    journal   = {Biometrics},
    number    = {2},
    pages     = {ujae019},
    publisher = {Oxford University Press},
    title     = {Case weighted power priors for hybrid control analyses with time-to-event data},
    volume    = {80},
    year      = {2024}
}

@article{duan2006evaluating,
    author    = {Duan, Y. and Ye, K. and Smith, E.~P.},
    journal   = {Environmetrics: The Official Journal of the International Environmetrics Society},
    number    = {1},
    pages     = {95--106},
    publisher = {Wiley Online Library},
    title     = {Evaluating water quality using power priors to incorporate historical information},
    volume    = {17},
    year      = {2006}
}

@article{hupf2021bayesian,
    author    = {Hupf, B. and Bunn, V. and Lin, J. and Dong, C.},
    journal   = {Statistics in Medicine},
    number    = {14},
    pages     = {3385--3399},
    publisher = {Wiley Online Library},
    title     = {{Bayesian} semiparametric meta-analytic-predictive prior for historical control borrowing in clinical trials},
    volume    = {40},
    year      = {2021}
}

@article{schmidli2014robust,
    author    = {Schmidli, H. and Gsteiger, S. and Roychoudhury, S. and O'Hagan, A. and Spiegelhalter, D. and Neuenschwander, B.},
    journal   = {Biometrics},
    number    = {4},
    pages     = {1023--1032},
    publisher = {Oxford University Press},
    title     = {Robust meta-analytic-predictive priors in clinical trials with historical control information},
    volume    = {70},
    year      = {2014}
}

@article{schmidli2020beyond,
    author    = {Schmidli, H. and H{\"a}ring, D.~A. and Thomas, M. and Cassidy, A. and Weber, S. and Bretz, F.},
    journal   = {Clinical Pharmacology \& Therapeutics},
    number    = {4},
    pages     = {806--816},
    publisher = {Wiley Online Library},
    title     = {Beyond randomized clinical trials: Use of external controls},
    volume    = {107},
    year      = {2020}
}

@article{neuenschwander2010summarizing,
    author    = {Neuenschwander, B. and Capkun-Niggli, G. and Branson, M. and Spiegelhalter, D.~J.},
    journal   = {Clinical Trials},
    number    = {1},
    pages     = {5--18},
    publisher = {SAGE Publications Sage UK: London, England},
    title     = {Summarizing historical information on controls in clinical trials},
    volume    = {7},
    year      = {2010}
}

@article{yang2023sam,
    author    = {Yang, P. and Zhao, Y. and Nie, L. and Vallejo, J. and Yuan, Y.},
    journal   = {Biometrics},
    number    = {4},
    pages     = {2857--2868},
    publisher = {Oxford University Press},
    title     = {{SAM}: Self-adapting mixture prior to dynamically borrow information from historical data in clinical trials},
    volume    = {79},
    year      = {2023}
}

@article{wang2022propensity,
    author    = {Wang, X. and Suttner, L. and Jemielita, T. and Li, X.},
    journal   = {Journal of Biopharmaceutical Statistics},
    number    = {1},
    pages     = {170--190},
    publisher = {Taylor \& Francis},
    title     = {Propensity score-integrated {Bayesian} prior approaches for augmented control designs: A simulation study},
    volume    = {32},
    year      = {2022}
}

@article{liu2021propensity,
    author    = {Liu, M. and Bunn, V. and Hupf, B. and Lin, J. and Lin, J.},
    journal   = {Statistics in Medicine},
    number    = {22},
    pages     = {4794--4808},
    publisher = {Wiley Online Library},
    title     = {Propensity-score-based meta-analytic predictive prior for incorporating real-world and historical data},
    volume    = {40},
    year      = {2021}
}

@article{jiang2023elastic,
    author    = {Jiang, L. and Nie, L. and Yuan, Y.},
    journal   = {Biometrics},
    number    = {1},
    pages     = {49--60},
    publisher = {Wiley Online Library},
    title     = {Elastic priors to dynamically borrow information from historical data in clinical trials},
    volume    = {79},
    year      = {2023}
}

@article{jin2021unit,
    author    = {Jin, H. and Yin, G.},
    journal   = {Statistics in Medicine},
    number    = {25},
    pages     = {5657--5672},
    publisher = {Wiley Online Library},
    title     = {Unit information prior for adaptive information borrowing from multiple historical datasets},
    volume    = {40},
    year      = {2021}
}

@article{viele2014use,
    author    = {Viele, K. and Berry, S. and Neuenschwander, B. and Amzal, B. and Chen, F. and Enas, N. and Hobbs, B. and Ibrahim, J.~G. and Kinnersley, N. and Lindborg, S. and Micallef, S. and Roychoudhury, S. and Thompson, L.},
    journal   = {Pharmaceutical Statistics},
    number    = {1},
    pages     = {41--54},
    publisher = {Wiley Online Library},
    title     = {Use of historical control data for assessing treatment effects in clinical trials},
    volume    = {13},
    year      = {2014}
}

@article{van2018including,
    author    = {van Rosmalen, J. and Dejardin, D. and van Norden, Y. and L{\"o}wenberg, B. and Lesaffre, E.},
    journal   = {Statistical Methods in Medical Research},
    number    = {10},
    pages     = {3167--3182},
    publisher = {SAGE Publications Sage UK: London, England},
    title     = {Including historical data in the analysis of clinical trials: Is it worth the effort?},
    volume    = {27},
    year      = {2018}
}

@article{li2020revisit,
    author    = {Li, W. and Liu, F. and Snavely, D.},
    journal   = {Pharmaceutical Statistics},
    number    = {5},
    pages     = {498--517},
    publisher = {Wiley Online Library},
    title     = {Revisit of test-then-pool methods and some practical considerations},
    volume    = {19},
    year      = {2020}
}

@article{okada2024decoupling,
    author    = {Okada, K. and Tanaka, S. and Matsubayashi, J. and Takahashi, K. and Yokota, I.},
    journal   = {Biometrical Journal},
    number    = {1},
    pages     = {2200312},
    publisher = {Wiley Online Library},
    title     = {Decoupling power and type {I} error rate considerations when incorporating historical control data using a test-then-pool approach},
    volume    = {66},
    year      = {2024}
}

@article{guo2024adaptive,
    author    = {Guo, B. and Laird, G. and Song, Y. and Chen, J. and Yuan, Y.},
    journal   = {Journal of the Royal Statistical Society Series C: Applied Statistics},
    number    = {2},
    pages     = {444--459},
    publisher = {Oxford University Press US},
    title     = {Adaptive hybrid control design for comparative clinical trials with historical control data},
    volume    = {73},
    year      = {2024}
}

@article{hobbs2011hierarchical,
    author    = {Hobbs, B.~P. and Carlin, B.~P. and Mandrekar, S.~J. and Sargent, D.~J.},
    journal   = {Biometrics},
    number    = {3},
    pages     = {1047--1056},
    publisher = {Oxford University Press},
    title     = {Hierarchical commensurate and power prior models for adaptive incorporation of historical information in clinical trials},
    volume    = {67},
    year      = {2011}
}

@article{hobbs2012commensurate,
    author    = {Hobbs, B.~P. and Sargent, D.~J. and Carlin, B.~P.},
    journal   = {Bayesian Analysis},
    number    = {3},
    pages     = {639--674},
    title     = {Commensurate priors for incorporating historical information in clinical trials using general and generalized linear models},
    volume    = {7},
    year      = {2012}
}

@article{khanal2025commensurate,
    author    = {Khanal, M. and Logan, B.~R. and Banerjee, A. and Fang, X. and Ahn, K.~W.},
    journal   = {Pharmaceutical Statistics},
    number    = {1},
    pages     = {e2464},
    publisher = {Wiley Online Library},
    title     = {A commensurate prior model with random effects for survival and competing risk outcomes to accommodate historical controls},
    volume    = {24},
    year      = {2025}
}

@article{kuhn2025distributionally,
    author    = {Kuhn, D. and Shafiee, S. and Wiesemann, W.},
    journal   = {Acta Numerica},
    pages     = {579--804},
    publisher = {Cambridge University Press},
    title     = {Distributionally robust optimization},
    volume    = {34},
    year      = {2025}
}

@article{mohajerin2018data,
    author    = {Mohajerin Esfahani, P. and Kuhn, D.},
    journal   = {Mathematical Programming},
    number    = {1},
    pages     = {115--166},
    publisher = {Springer},
    title     = {Data-driven distributionally robust optimization using the {Wasserstein} metric: Performance guarantees and tractable reformulations},
    volume    = {171},
    year      = {2018}
}

@article{blanchet2019quantifying,
    author    = {Blanchet, J. and Murthy, K.},
    journal   = {Mathematics of Operations Research},
    number    = {2},
    pages     = {565--600},
    publisher = {INFORMS},
    title     = {Quantifying distributional model risk via optimal transport},
    volume    = {44},
    year      = {2019}
}

@article{gao2023distributionally,
    author    = {Gao, R. and Kleywegt, A.},
    journal   = {Mathematics of Operations Research},
    number    = {2},
    pages     = {603--655},
    publisher = {INFORMS},
    title     = {Distributionally robust stochastic optimization with {Wasserstein} distance},
    volume    = {48},
    year      = {2023}
}

@article{alt2024leap,
    author    = {Alt, E.~M. and Chang, X. and Jiang, X. and Liu, Q. and Mo, M. and Xia, H.~A. and Ibrahim, J.~G.},
    journal   = {Biometrics},
    number    = {3},
    pages     = {ujae083},
    publisher = {Oxford University Press},
    title     = {{LEAP}: The latent exchangeability prior for borrowing information from historical data},
    volume    = {80},
    year      = {2024}
}

@article{kaizer2018bayesian,
    author    = {Kaizer, A.~M. and Koopmeiners, J.~S. and Hobbs, B.~P.},
    journal   = {Biostatistics},
    number    = {2},
    pages     = {169--184},
    publisher = {Oxford University Press},
    title     = {{Bayesian} hierarchical modeling based on multisource exchangeability},
    volume    = {19},
    year      = {2018}
}

@article{lu2025overlapping,
    author    = {Lu, X. and Lee, J.~J.},
    journal   = {Journal of Computational and Graphical Statistics},
    pages     = {1--15},
    publisher = {Taylor \& Francis},
    title     = {Overlapping indices for dynamic information borrowing in {Bayesian} hierarchical modeling},
    year      = {2025}
}

@article{ohigashi2025nonparametric,
    author    = {Ohigashi, T. and Maruo, K. and Sozu, T. and Gosho, M.},
    journal   = {Biometrics},
    number    = {3},
    pages     = {ujaf118},
    publisher = {Oxford University Press},
    title     = {Nonparametric {Bayesian} approach for dynamic borrowing of historical control data},
    volume    = {81},
    year      = {2025}
}

@article{neuenschwander2016robust,
    author    = {Neuenschwander, B. and Wandel, S. and Roychoudhury, S. and Bailey, S.},
    journal   = {Pharmaceutical Statistics},
    number    = {2},
    pages     = {123--134},
    publisher = {Wiley Online Library},
    title     = {Robust exchangeability designs for early phase clinical trials with multiple strata},
    volume    = {15},
    year      = {2016}
}

@book{villani2009optimal,
    author    = {Villani, C.},
    publisher = {Springer},
    title     = {Optimal Transport: Old and New},
    volume    = {338},
    year      = {2009}
}

@article{jahanshahi2021use,
    author    = {Jahanshahi, M. and Gregg, K. and Davis, G. and Ndu, A. and Miller, V. and Vockley, J. and Ollivier, C. and Franolic, T. and Sakai, S.},
    journal   = {Therapeutic Innovation \& Regulatory Science},
    number    = {5},
    pages     = {1019--1035},
    publisher = {Springer},
    title     = {The use of external controls in {FDA} regulatory decision making},
    volume    = {55},
    year      = {2021}
}

@article{goring2019characteristics,
    author    = {Goring, S. and Taylor, A. and M{\"u}ller, K. and Li, T.~J.~J. and Korol, E.~E. and Levy, A.~R. and Freemantle, N.},
    journal   = {{BMJ} Open},
    number    = {2},
    pages     = {e024895},
    publisher = {British Medical Journal Publishing Group},
    title     = {Characteristics of non-randomised studies using comparisons with external controls submitted for regulatory approval in the {USA} and {Europe}: A systematic review},
    volume    = {9},
    year      = {2019}
}

@article{rippin2022review,
    author    = {Rippin, G. and Ballarini, N. and Sanz, H. and Largent, J. and Quinten, C. and Pignatti, F.},
    journal   = {Drug Safety},
    number    = {8},
    pages     = {815--837},
    publisher = {Springer},
    title     = {A review of causal inference for external comparator arm studies},
    volume    = {45},
    year      = {2022}
}

@article{liu2025design,
    author    = {Liu, J. and Yao, M. and Wang, M. and Jie, W. and Liu, Y. and Luo, X. and Huan, J. and Deng, K. and Deng, K. and Zou, K. and Zhang, Y. and Li, L. and Sun, X.},
    journal   = {{JAMA} Network Open},
    number    = {9},
    pages     = {e2530277},
    title     = {Design, conduct, and analysis of externally controlled trials},
    volume    = {8},
    year      = {2025}
}

@article{bennett2021novel,
    author    = {Bennett, M. and White, S. and Best, N. and Mander, A.},
    journal   = {Pharmaceutical Statistics},
    number    = {3},
    pages     = {462--484},
    publisher = {Wiley Online Library},
    title     = {A novel equivalence probability weighted power prior for using historical control data in an adaptive clinical trial design: A comparison to standard methods},
    volume    = {20},
    year      = {2021}
}

@misc{fda2023externally_controlled,
    author    = {{U.S. Food and Drug Administration}},
    title     = {Considerations for the design and conduct of externally controlled trials for drug and biological products: Guidance for industry (draft guidance)},
    year      = {2023},
    month     = feb,
    note      = {Draft Guidance},
    url       = {https://www.fda.gov/media/164960/download}
}

@misc{ich2000e10,
    title     = {{ICH} harmonised tripartite guideline {E10}: Choice of control group and related issues in clinical trials},
    author    = {{International Conference on Harmonisation (ICH)}},
    year      = {2000},
    month     = jul,
    note      = {Step 4 version dated 20 July 2000},
    url       = {https://database.ich.org/sites/default/files/E10_Guideline.pdf}
}

@misc{ema2025external_controls_concept,
    author    = {{European Medicines Agency}},
    title     = {Draft concept paper on the development of a reflection paper on the use of external controls for evidence generation in regulatory decision-making},
    year      = {2025},
    note      = {Reference: EMA/CHMP/225255/2025},
    url       = {https://www.ema.europa.eu/en/development-reflection-paper-use-external-controls-evidence-generation-regulatory-decision-making-scientific-guideline}
}

@article{gao2025control,
    author    = {Gao, P. and Ni, X. and Li, J. and Chu, R.},
    journal   = {Pharmaceutical Statistics},
    number    = {3},
    pages     = {e70011},
    publisher = {Wiley Online Library},
    title     = {Control of unconditional type {I} error in clinical trials with external control borrowing—a two-stage adaptive design perspective},
    volume    = {24},
    year      = {2025}
}

@article{psioda2019bayesian,
    author    = {Psioda, M.~A. and Ibrahim, J.~G.},
    journal   = {Biostatistics},
    number    = {3},
    pages     = {400--415},
    publisher = {Oxford University Press},
    title     = {{Bayesian} clinical trial design using historical data that inform the treatment effect},
    volume    = {20},
    year      = {2019}
}

@article{kopp2020power,
    author    = {Kopp-Schneider, A. and Calderazzo, S. and Wiesenfarth, M.},
    journal   = {Biometrical Journal},
    number    = {2},
    pages     = {361--374},
    publisher = {Wiley Online Library},
    title     = {Power gains by using external information in clinical trials are typically not possible when requiring strict type {I} error control},
    volume    = {62},
    year      = {2020}
}

@article{lee2024using,
    author    = {Lee, S.~Y.},
    journal   = {{BMC} Medical Research Methodology},
    pages     = {110},
    publisher = {Springer},
    title     = {Using {Bayesian} statistics in confirmatory clinical trials in the regulatory setting: A tutorial review},
    volume    = {24},
    year      = {2024}
}

@article{pan2017calibrated,
    author    = {Pan, H. and Yuan, Y. and Xia, J.},
    journal   = {Journal of the Royal Statistical Society Series C: Applied Statistics},
    number    = {5},
    pages     = {979--996},
    publisher = {Oxford University Press},
    title     = {A calibrated power prior approach to borrow information from historical data with application to biosimilar clinical trials},
    volume    = {66},
    year      = {2017}
}

@article{eggleston2021bayesctdesign,
    author    = {Eggleston, B.~S. and Ibrahim, J.~G. and McNeil, B. and Catellier, D.},
    journal   = {Journal of Statistical Software},
    number    = {21},
    pages     = {1--51},
    title     = {{BayesCTDesign}: An {R} package for {Bayesian} trial design using historical control data},
    volume    = {100},
    year      = {2021}
}

@article{ling2021calibrated,
    author    = {Ling, S.~X. and Hobbs, B.~P. and Kaizer, A.~M. and Koopmeiners, J.~S.},
    journal   = {Journal of Biopharmaceutical Statistics},
    number    = {6},
    pages     = {852--867},
    publisher = {Taylor \& Francis},
    title     = {Calibrated dynamic borrowing using capping priors},
    volume    = {31},
    year      = {2021}
}

@article{demartino2025eliciting,
    author    = {Demartino, R.~M. and Egidi, L. and Torelli, N. and Ntzoufras, I.},
    journal   = {Computational Statistics \& Data Analysis},
    pages     = {108180},
    publisher = {Elsevier},
    title     = {Eliciting prior information from clinical trials via calibrated {Bayes} factor},
    volume    = {209},
    year      = {2025}
}

@article{peeters2010randomized,
    author    = {Peeters, M. and Price, T.~J. and Cervantes, A. and Sobrero, A.~F. and Ducreux, M. and Hotko, Y. and Andr{\'e}, T. and Chan, E. and Lordick, F. and Punt, C.~J. and Strickland, A.~H. and Wilson, G. and Ciuleanu, T.~E. and Roman, L. and Van Cutsem, E. and Tzekova, V. and Collins, S. and Oliner, K.~S. and Rong, A. and Gansert, J.},
    journal   = {Journal of Clinical Oncology},
    number    = {31},
    pages     = {4706--4713},
    publisher = {American Society of Clinical Oncology},
    title     = {Randomized phase {III} study of panitumumab with fluorouracil, leucovorin, and irinotecan ({FOLFIRI}) compared with folfiri alone as second-line treatment in patients with metastatic colorectal cancer},
    volume    = {28},
    year      = {2010}
}

@article{tabernero2014aflibercept,
    author    = {Tabernero, J. and Van Cutsem, E. and Lakom{\'y}, R. and Prausov{\'a}, J. and Ruff, P. and van Hazel, G.~A. and Moiseyenko, V.~M. and Ferry, D.~R. and McKendrick, J.~J. and Soussan-Lazard, K. and Chevalier, S. and Allegra, C.~J.},
    journal   = {European Journal of Cancer},
    number    = {2},
    pages     = {320--331},
    publisher = {Elsevier},
    title     = {Aflibercept versus placebo in combination with fluorouracil, leucovorin and irinotecan in the treatment of previously treated metastatic colorectal cancer: Prespecified subgroup analyses from the {VELOUR} trial},
    volume    = {50},
    year      = {2014}
}

\newpage
\appendix
\section{Connections to Existing Adaptive Borrowing Priors}
\label{app:sec:incl-ebw}

This section establishes a structural equivalence between a wide array of existing Bayesian and frequentist adaptive borrowing methods and the EBW class defined in \eqref{eq:mu-hat}.
The core analytical insight is that, under standard conjugate modeling for a single arm, the posterior mean derived from many prominent dynamic borrowing procedures collapses exactly to an affine combination of the current and historical sample means.
By systematically identifying the implied effective weight $w$ (or equivalently, the borrowing parameter $\lambda$) in each framework, we demonstrate that the proposed DRO calibration method (BOND) is not merely a standalone estimator.
Rather, it serves as a universal robustness wrapper:
practitioners can retain their preferred Bayesian modeling machinery while employing our DRO-based criterion to rigorously calibrate its associated hyperparameters.

\subsection{Effective Borrowing Form for a Single Arm}
\label{app:subsec:ebw-setup}

We first focus on a single arm, suppressing the arm index $a$ for clarity.
Let $n_j\coloneq n_{j,a}$ and $\bar{Y}_j\coloneq \bar{Y}_{j,a}$ for $j \in \{C, H\}$.
For any borrowing parameter $\lambda\ge 0$, recall the EBW estimator:
\begin{equation*}
    \hat{\mu}(\lambda)
    \coloneq
    \frac{n_C\bar{Y}_C+\lambda n_H\bar{Y}_H}{n_C + \lambda n_H}
    =
    (1-w(\lambda))\bar{Y}_C + w(\lambda)\bar{Y}_H
    ,
\end{equation*}
where the EBW is defined as
\begin{equation*}
    w(\lambda)
    \coloneq
    \frac{\lambda n_H}{n_C + \lambda n_H}
    \in
    [0,1)
    .
\end{equation*}

\begin{lemma}[One-to-one mapping]
\label{app:lem:w-lambda}
    Assume $n_C > 0$ and $n_H > 0$.
    The mapping $w(\lambda) = \lambda n_H/(n_C + \lambda n_H)$ is a bijection from $[0,\infty)$ to $[0,1)$.
    Its inverse is given by
    \begin{equation*}
        \lambda(w)
        =
        \frac{n_C}{n_H} \cdot \frac{w}{1-w}
        ,
        \quad
        w\in[0,1)
        .
    \end{equation*}
\end{lemma}
See Appendix~\ref{app:pf-lem-w-lambda} for the proof.
Lemma~\ref{app:lem:w-lambda} guarantees that any procedure producing an estimator of the form $(1-w)\bar{Y}_C + w\bar{Y}_H$ for some $w\in[0,1)$ can be exactly parameterized as $\hat{\mu}(\lambda(w))$ within our framework.

\subsection{Test-Then-Pool (TTP) Procedures}
\label{app:subsec:test-then-pool}

TTP procedures are frequentist dynamic borrowing rules that determine whether to incorporate historical information via a preliminary commensurability screen, subsequently performing estimation using either pooled or current-only data \citep{viele2014use,li2020revisit}.
In the canonical externally controlled setting, if the screen declares the two sources sufficiently similar, the controls are pooled;
otherwise, the historical data are discarded.

Maintaining the single-arm notation from Appendix~\ref{app:subsec:ebw-setup}, let $D_C=(Y_{C,1},\dots,Y_{C,n_C})$ and $D_H=(Y_{H,1},\dots,Y_{H,n_H})$, and define the filtration $\mathcal{F}\coloneq\sigma(D_C,D_H)$.
A generic TTP rule is governed by an $\mathcal{F}$-measurable pooling indicator
\begin{equation*}
    \hat{\eta}
    =
    \hat{\eta}(D_C,D_H)
    \in
    \{0,1\}
    ,
\end{equation*}
where $\hat{\eta}=1$ dictates pooling and $\hat{\eta}=0$ dictates no borrowing.
Given a nominal pooling intensity $\lambda_{\mathrm{pool}}\ge 0$ (with $\lambda_{\mathrm{pool}}=1$ corresponding to full pooling of the historical sample), the TTP estimator of the arm mean is the dichotomous rule:
\begin{equation}
\label{eq:ttp-mu-def}
    \hat{\mu}_{\mathrm{TTP}}
    \coloneq
    \hat{\eta}\hat{\mu}(\lambda_{\mathrm{pool}})
    +
    (1-\hat{\eta})\hat{\mu}(0)
    =
    \begin{cases}
        \hat{\mu}(\lambda_{\mathrm{pool}}), & \hat{\eta}=1,\\
        \bar{Y}_C, & \hat{\eta}=0.
    \end{cases}
\end{equation}

\begin{lemma}[Test-then-pool implies a data-adaptive EBW]
\label{app:lem:ttp-ebw}
    For any $\mathcal{F}$-measurable $\hat{\eta}\in\{0,1\}$ and fixed $\lambda_{\mathrm{pool}}\ge 0$, the estimator $\hat{\mu}_{\mathrm{TTP}}$ in \eqref{eq:ttp-mu-def} is exactly an EBW estimator governed by the data-adaptive borrowing parameter
    \begin{equation*}
        \hat{\lambda}
        \coloneq
        \hat{\eta}\lambda_{\mathrm{pool}}
        .
    \end{equation*}
    Consequently,
    \begin{equation*}
        \hat{\mu}_{\mathrm{TTP}}
        =
        \hat{\mu}(\hat{\lambda})
        =
        (1-\hat{w})\bar{Y}_C+\hat{w}\bar{Y}_H,
        \quad
        \hat{w}\coloneq w(\hat{\lambda}).
    \end{equation*}
    Under full pooling ($\lambda_{\mathrm{pool}}=1$), this simplifies to $\hat{w}=\hat{\eta} n_H/(n_C+n_H)$.
\end{lemma}
See Appendix~\ref{app:pf-lem-ttp-ebw} for the proof.

\begin{remark}[Screening tests and post-selection inference]
\label{app:rem:ttp-postselection}
    The EBW representation above concerns the point estimator.
    In the full TTP testing procedure, the stage-2 test statistic and its critical value are applied after a data-dependent pooling decision.
    Consequently, unconditional operating characteristics need not coincide with those of either always pool or never pool rules.
    For example, \citet{li2020revisit} show that, without calibration, the nominal stage-2 level can yield inflated type I error even under perfect commensurability, and propose adjusting the stage-2 level.
    This motivates calibration frameworks such as BOND that determine borrowing (or its tuning parameter) by optimizing power subject to explicit size control under a prespecified heterogeneity tolerance set.
\end{remark}

\begin{remark}[A concrete screening test]
\label{app:rem:ttp-test-example}
    For continuous outcomes, a common original TTP screen uses the two-sample Wald statistic
    \begin{equation*}        
        T_{\mathrm{pool}}
        =
        \frac{\bar{Y}_H-\bar{Y}_C}{\sqrt{\hat{\sigma}_H^2/n_H+\hat{\sigma}_C^2/n_C}}
        ,
        \quad
        \hat{\eta}
        =
        \boldsymbol{1}\bigl\{
            |T_{\mathrm{pool}}|\le z_{1-\alpha_{\mathrm{pool}}/2}
        \bigr\}
        ,
    \end{equation*}
    i.e., pooling occurs upon fail-to-reject at level $\alpha_{\mathrm{pool}}$.
    Equivalence-based TTP instead pools upon rejection of nonequivalence within a margin; see \citet{li2020revisit} for details and analogues for binary outcomes.
\end{remark}

\subsection{Power Prior}
\label{app:subsec:power-prior}

The power prior \citep{ibrahim2000power,ibrahim2015power} modulates the influence of historical information by raising the historical likelihood to a power parameter.
To maintain notation, we denote this power parameter as $\lambda \in [0,1]$.
For a single-arm parameter $\mu$ and historical dataset $D_H$, the conditional power prior is defined as:
\begin{equation*}
    \pi(\mu\mid D_H,\lambda)
    \propto
    L_H(\mu)^{\lambda}\pi_0(\mu)
    ,
\end{equation*}
where $L_H(\mu)$ is the historical likelihood and $\pi_0(\mu)$ is a baseline prior.
After observing current data $D_C$, the posterior updates to:
\begin{equation}
\label{eq:pp-post}
    \pi(\mu\mid D_C,D_H,\lambda)
    \propto
    L_C(\mu)L_H(\mu)^{\lambda}\pi_0(\mu)
    .
\end{equation}
Under standard conjugate models with weakly informative baseline priors, the posterior mean resolves precisely to the EBW form $\mathbb{E}[\mu\mid D_C,D_H,\lambda]=\hat{\mu}(\lambda)$.

\subsubsection{Bernoulli Likelihood with Beta Base Prior}
\label{app:subsubsec:pp-ber}

Assume $Y_{j,i}\mid \mu \overset{\mathrm{i.i.d.}}{\sim} \mathrm{Bernoulli}(\mu)$.
Let $S_j\coloneq\sum_{i=1}^{n_j}Y_{j,i}$ and $\bar{Y}_j=S_j/n_j$.
With baseline prior $\mu\sim\mathrm{Beta}(\alpha_0,\beta_0)$, the posterior becomes:
\begin{equation*}
    \mu\mid D_C,D_H,\lambda
    \sim
    \mathrm{Beta}\Bigl(
        \alpha_0+S_C+\lambda S_H
        ,
        \;
        \beta_0+(n_C-S_C)+\lambda(n_H-S_H)
    \Bigr)
    .
\end{equation*}
In the weakly informative limit ($\alpha_0,\beta_0\to 0$), the posterior mean simplifies exactly to the EBW estimator:
\begin{equation*}
    \mathbb{E}[\mu\mid D_C,D_H,\lambda]
    =
    \frac{S_C+\lambda S_H}{n_C+\lambda n_H}
    =
    \frac{n_C\bar{Y}_C+\lambda n_H\bar{Y}_H}{n_C+\lambda n_H}
    =\hat{\mu}(\lambda)
    .
\end{equation*}

\subsubsection{Normal Likelihood with Flat Prior and Known Variance}
\label{app:subsubsec:normal-flat}

Assume $Y_{j,i}\mid \mu \overset{\mathrm{i.i.d.}}{\sim} N(\mu,\sigma^2)$ with a known common variance $\sigma^2$ and flat baseline prior $\pi_0(\mu)\propto 1$.
Applying the power $\lambda$ effectively scales the historical precision by $\lambda$.
The combined posterior is:
\begin{equation*}
    \mu\mid D_C,D_H,\lambda
    \sim
    N\Biggl(
        \frac{n_C\bar{Y}_C+\lambda n_H\bar{Y}_H}{n_C+\lambda n_H}
        ,\;
        \frac{\sigma^2}{n_C+\lambda n_H}
    \Biggr)
    ,
\end{equation*}
which once again yields $\mathbb{E}[\mu\mid D_C,D_H,\lambda]=\hat{\mu}(\lambda)$.

\subsection{Modified and Normalized Power Priors}
\label{app:subsec:mod-power}

A common hierarchical extension treats the power parameter $\lambda$ as unknown and assigns it a prior $\pi_0(\lambda)$. 
There are two standard formulations when $\lambda$ is random \citep{ibrahim2015power,duan2006evaluating,pawel2023normalized}.

\subsubsection{Joint Power Prior}
\label{app:subsubsec:joint-pp}
The joint power prior specifies a joint prior on $(\mu,\lambda)$ directly:
\begin{equation}
    \label{eq:jpp}
    \pi(\mu,\lambda\mid D_H)
    \propto
    L_H(\mu)^{\lambda}\pi_0(\mu)\pi_0(\lambda)
    ,
    \quad
    \lambda
    \in
    [0,1]
    .
\end{equation}

\subsubsection{Normalized Power Prior}
\label{app:subsubsec:normalized-pp}
The normalized power prior explicitly normalizes the conditional prior for $\mu$ given $\lambda$:
\begin{equation*}
    \pi(\mu,\lambda\mid D_H)
    =
    \pi(\mu\mid D_H,\lambda)\pi_0(\lambda)
    =
    \frac{L_H(\mu)^{\lambda}\pi_0(\mu)}{m^\ast(\lambda)}\pi_0(\lambda)
    ,
\end{equation*}
where $m^\ast(\lambda) \coloneq \int L_H(u)^{\lambda}\pi_0(u)\dd u < \infty$ for $\lambda\in(0,1]$.
Compared to \eqref{eq:jpp}, the critical difference is the $\lambda$-dependent normalizing constant $m^\ast(\lambda)$, which modulates the induced marginal posterior for $\lambda$ without altering the conditional form $\pi(\mu\mid D_H,\lambda)$.

After observing current data $D_C$, both formulations yield a joint posterior of the form:
\begin{equation*}
    \pi(\mu,\lambda\mid D_C,D_H)
    \propto
    L_C(\mu)L_H(\mu)^{\lambda}\pi_0(\mu)\pi_0(\lambda)\times
    \begin{cases}
        1, & \text{(joint power prior)}
        ,\\
        m^\ast(\lambda)^{-1}, & \text{(normalized power prior)}
        .
    \end{cases}
\end{equation*}
Crucially, conditional on $\lambda$, the posterior for $\mu$ coincides exactly with the fixed-$\lambda$ power-prior update in \eqref{eq:pp-post}. 
Therefore, under the conjugate setups in Section~\ref{app:subsec:power-prior}, the conditional expectation remains $\mathbb{E}[\mu\mid D_C,D_H,\lambda]=\hat{\mu}(\lambda)$.

By the law of total expectation, the marginal posterior mean under either formulation satisfies:
\begin{equation*}
    \mathbb{E}[\mu\mid D_C,D_H]
    =
    \mathbb{E}\bigl[
        \hat{\mu}(\lambda)
        \bigm| D_C,D_H
    \bigr]
    =
    \bar{Y}_C+\mathbb{E}\bigl[
        w(\lambda)
        \bigm| D_C,D_H
    \bigr](\bar{Y}_H-\bar{Y}_C)
    .
\end{equation*}
Thus, hierarchical power priors remain strictly within the EBW class. They are simply characterized by a data-adaptive EBW $w_{\mathrm{eff}} \coloneq \mathbb{E}[w(\lambda)\mid D_C,D_H] \in [0,1)$. By Lemma~\ref{app:lem:w-lambda}, this implicitly defines an adaptive tuning parameter $\lambda_{\mathrm{eff}} = \lambda(w_{\mathrm{eff}})$, demonstrating that the distinction between joint and normalized formulations merely shifts the realized value of $w_{\mathrm{eff}}$, rather than fundamentally changing the EBW structure.

\subsection{Commensurate Prior}
\label{app:subsec:commensurate}

Commensurate priors formalize the degree of agreement between current and historical parameters via a precision parameter $\tau$ \citep{hobbs2012commensurate}.
To explicitly connect this to the EBW class, consider the conjugate normal-mean setting with common variance $\sigma^2$.
A location commensurate prior is specified as $\mu_C \mid \mu_H,\tau \sim N(\mu_H,\tau^{-1})$ with a flat base prior $\pi(\mu_H) \propto 1$.

\begin{lemma}[Commensurate prior implies an EBW representation]
\label{app:lem:commensurate-ebw}
    Fix $\tau>0$.
    The posterior mean of $\mu_C$ resolves to an affine combination:
    \begin{equation*}
        \mathbb{E}[\mu_C\mid D_C,D_H,\tau]
        =
        (1-w(\tau))\bar{Y}_C + w(\tau)\bar{Y}_H
        ,
    \end{equation*}
    where the implied effective weight in EBW form is governed by the effective historical sample size $m_{\mathrm{eff}}(\tau)$:
    \begin{equation*}
        w(\tau)
        =
        \frac{\lambda(\tau)n_H}{n_C+\lambda(\tau)n_H}
        ,
        \quad\text{with}\quad
        \lambda(\tau)
        \coloneq
        \frac{m_{\mathrm{eff}}(\tau)}{n_H}
        =
        \frac{\sigma^2\tau}{\sigma^2\tau+n_H}
        \in
        (0,1)
        .
    \end{equation*}
\end{lemma}
See Appendix~\ref{app:pf-lem-commensurate-ebw} for the proof.
Lemma~\ref{app:lem:commensurate-ebw} demonstrates that tuning the commensurability precision $\tau$ is mathematically equivalent to calibrating an EBW discount factor $\lambda$.
Consequently, the proposed DRO framework can systematically optimize the effective borrowing behavior governed by $\tau$.

\begin{remark}[Beyond the conjugate normal-mean case]
\label{app:rem:commensurate-general}
    Although the exact EBW derivation above relies on the conjugate normal-mean model, this affine shrinkage structure extends naturally to broader settings. 
    For generalized linear models, the conditional posterior mean under a commensurate prior asymptotically approximates a precision-weighted average of the current and historical maximum likelihood estimators. In this regime, raw sample sizes are effectively replaced by Fisher information matrices \citep{hobbs2012commensurate}. 
    This asymptotic behavior formally justifies interpreting the commensurate prior's precision $\tau$ through the lens of an EBW---even in nonconjugate settings via standard Gaussian or Laplace approximations.
\end{remark}

\subsection{MAP and Robust MAP Priors}
\label{app:subsec:map}

MAP priors provide predictive distributions for parameters in a new study, often approximated by finite mixtures of conjugate distributions \citep{neuenschwander2010summarizing,schmidli2014robust}.

\subsubsection{MAP as a Finite Mixture of Conjugates}
\label{app:subsubsec:map-finite-mixture}

In the single-arm notation of Appendix~\ref{app:subsec:ebw-setup}, assume the MAP prior given historical data admits the mixture representation:
\begin{equation}
\label{app:eq:map-mixture-prior}
    \pi_{\mathrm{MAP}}(\mu\mid D_H)
    =
    \sum_{k=1}^K \omega_k\pi_k(\mu)
    ,
    \quad
    \omega_k\ge 0
    ,\;
    \sum_{k=1}^K \omega_k=1
    ,
\end{equation}
where each $\pi_k$ is conjugate.
Under conjugacy, the $k$th component posterior mean exhibits affine shrinkage: $\mathbb{E}_{\pi_k}[\mu\mid D_C] = (1-w_k)\bar{Y}_C + w_k m_k$, where $m_k = \mathbb{E}_{\pi_k}[\mu]$.
Let $f_k(D_C) \coloneq \int L_C(\mu)\pi_k(\mu)\dd\mu$ denote the marginal likelihood.
The posterior mixture weights are updated via Bayes' rule:
\begin{equation}
\label{app:eq:map-post-weights}
    \omega_k^{\mathrm{post}}
    =
    \frac{\omega_k f_k(D_C)}{\sum_{\ell=1}^K \omega_\ell f_\ell(D_C)}
    .
\end{equation}

\begin{lemma}[Mixture-of-conjugates implies an EBW representation]
\label{app:lem:map-ebw}
    Under \eqref{app:eq:map-mixture-prior} and \eqref{app:eq:map-post-weights}, the MAP marginal posterior mean satisfies:
    \begin{equation*}
        \mathbb{E}[\mu\mid D_C,D_H]
        =
        (1-w_{\mathrm{eff}})\bar{Y}_C + w_{\mathrm{eff}}m_{\mathrm{eff}}
        ,
    \end{equation*}
    where $w_{\mathrm{eff}} \coloneq \sum_{k=1}^K \omega_k^{\mathrm{post}}w_k \in [0,1)$, and $m_{\mathrm{eff}}$ is the posterior-weighted aggregate of the historical component means.
\end{lemma}
See Appendix~\ref{app:pf-lem-map-ebw} for the proof.
When all mixture components share a common prior mean (e.g., mixing only over precision parameters) such that $m_{\mathrm{eff}}=\bar{Y}_H$, Lemma~\ref{app:lem:w-lambda} yields an exactly equivalent $\lambda = \lambda(w_{\mathrm{eff}})$.

\subsubsection{Robust MAP}
\label{app:subsubsec:r-map}
Robust MAP priors explicitly protect against prior-data conflict by mixing the MAP prior with a vague (weakly informative) conjugate prior $\pi_V$ \citep{schmidli2014robust}:
\begin{equation*}
    \pi_{\mathrm{RMAP}}(\mu\mid D_H)
    =
    (1-\epsilon)\pi_{\mathrm{MAP}}(\mu\mid D_H)
    +
    \epsilon\pi_V(\mu)
    ,
    \quad
    \epsilon\in(0,1)
    .
\end{equation*}
Let $f_V(D_C)$ denote the vague-component marginal likelihood.
The posterior probability assigned to the vague component becomes:
\begin{equation}
\label{app:eq:rmap-vague-post}
    \epsilon^{\mathrm{post}}
    =
    \frac{\epsilon f_V(D_C)}{\epsilon f_V(D_C) + (1-\epsilon)\sum_{k=1}^K \omega_k f_k(D_C)}
    .
\end{equation}
Following the same logic as Lemma~\ref{app:lem:map-ebw}, the Robust MAP posterior mean resolves to an affine form with an overall effective weight $w_{\mathrm{eff}}^{\mathrm{RMAP}} = (1-\epsilon^{\mathrm{post}})\sum_{k=1}^K \omega_k^{\mathrm{post}} w_k + \epsilon^{\mathrm{post}} w_V$.

Crucially, the portion of borrowing attributable to the historical data is $(1-\epsilon^{\mathrm{post}})w_{\mathrm{eff}}$. Under severe prior-data conflict, the informative marginal likelihood $\sum \omega_k f_k(D_C)$ diminishes relative to $f_V(D_C)$.
Consequently, \eqref{app:eq:rmap-vague-post} mathematically drives $\epsilon^{\mathrm{post}}$ toward $1$, heavily shrinking the effective historical borrowing weight toward $0$.
Assuming $w_V$ is negligible, the robust MAP estimator collapses to an EBW estimator, meaning its inherent robustness behavior can be directly audited and tuned using the proposed BOND framework.

\subsection{Elastic Prior}
\label{app:subsec:elastic}

The elastic prior proposed by \citet{jiang2023elastic} functions as an empirical-Bayes dynamic borrowing mechanism that modulates the integration of historical information via an elastic function $g(T) \in [0,1]$.
Here, $T=T(D_C,D_H)$ denotes a prespecified congruence statistic evaluating the compatibility between the current dataset $D_C$ and the historical dataset $D_H$ (e.g., a two-sample $t$-statistic for Gaussian outcomes or a chi-square statistic for binary outcomes).
The function $g(\cdot)$ is a monotonically decreasing map calibrated to yield $g(T) \approx 1$ under near-commensurability and $g(T) \approx 0$ under severe prior-data conflict.
Operationally, the method first derives the historical posterior under a vague base prior and subsequently inflates its variance by a factor of $g(T)^{-1}$ (equivalently, scaling its precision by $g(T)$).
This smoothly interpolates between full pooling ($g(T)=1$) and essentially no borrowing ($g(T)=0$).
Under standard conjugate models, we show that the resulting posterior mean naturally admits an EBW representation.

\subsubsection{Bernoulli Likelihood with Beta Base Prior}
\label{app:subsubsec-elastic-ber}

Assume $Y_{H,i}\mid \mu_H \overset{\text{i.i.d.}}{\sim}\mathrm{Bernoulli}(\mu_H)$ and $Y_{C,i}\mid \mu \overset{\text{i.i.d.}}{\sim}\mathrm{Bernoulli}(\mu)$.
Let $S_H=\sum_{i=1}^{n_H}Y_{H,i}$ and $S_C=\sum_{i=1}^{n_C}Y_{C,i}$.
Under a Beta base prior $\mu_H\sim \mathrm{Beta}(\alpha_0,\beta_0)$, the historical posterior is $\mathrm{Beta}(\alpha_0+S_H,\beta_0+n_H-S_H)$.
The elastic prior for $\mu$ is then constructed by discounting the posterior information by $g(T)$:
\begin{equation*}
    \mu\mid D_H,T
    \sim
    \mathrm{Beta}\bigl(
        g(T)(\alpha_0+S_H)
        ,\;
        g(T)(\beta_0+n_H-S_H)
    \bigr)
    .
\end{equation*}
Updating this prior with the current data $D_C$ yields the posterior:
\begin{equation*}
    \mu\mid D_C,D_H,T
    \sim
    \mathrm{Beta}\bigl(
        S_C+g(T)(\alpha_0+S_H)
        ,\;
        (n_C-S_C)+g(T)(\beta_0+n_H-S_H)
    \bigr)
    .
\end{equation*}
Evaluating the posterior mean under the noninformative limit $\alpha_0, \beta_0 \to 0$ gives:
\begin{equation*}
    \mathbb{E}[\mu\mid D_C,D_H,T]
    =
    \frac{S_C+g(T)S_H}{n_C+g(T)n_H}
    =
    \frac{n_C\bar{Y}_C+g(T)n_H\bar{Y}_H}{n_C+g(T)n_H}
    =
    \hat{\mu}\bigl(g(T)\bigr)
    ,
\end{equation*}
which aligns exactly with the EBW estimator utilizing the borrowing parameter $\lambda=g(T)$.

\subsubsection{Gaussian Likelihood}
\label{app:subsubsec:elastic-gaussian}

Assume $Y_{C,i}\mid \mu \sim N(\mu,\sigma_C^2)$ and $Y_{H,i}\mid \mu_H \sim N(\mu_H,\sigma_H^2)$ with known variances (or consistent plug-in estimates).
Under a flat base prior $\pi_0(\mu_H) \propto 1$, the historical posterior is $\mu_H\mid D_H \sim N(\bar{Y}_H,\sigma_H^2/n_H)$.
The elastic prior incorporates the congruence statistic by inflating this variance by $g(T)^{-1}$:
\begin{equation*}
    \mu\mid D_H,T \sim N\Bigl(
        \bar{Y}_H,\frac{\sigma_H^2}{g(T)n_H}
    \Bigr)
    .
\end{equation*}
Combining this with the current likelihood yields the posterior mean:
\begin{equation*}
    \mathbb{E}[\mu\mid D_C,D_H,T]
    =
    \frac{\frac{n_C}{\sigma_C^2}\bar{Y}_C+\frac{g(T)n_H}{\sigma_H^2}\bar{Y}_H}
    {\frac{n_C}{\sigma_C^2}+\frac{g(T)n_H}{\sigma_H^2}}
    .
\end{equation*}
Multiplying the numerator and the denominator by $\sigma_C^2$ reformulates this expectation into the EBW structure:
\begin{equation*}
    \mathbb{E}[\mu\mid D_C,D_H,T]
    =
    \frac{n_C\bar{Y}_C+\lambda n_H\bar{Y}_H}{n_C+\lambda n_H}
    =
    \hat{\mu}(\lambda),
    \quad
    \lambda
    \coloneq
    g(T)\frac{\sigma_C^2}{\sigma_H^2}
    .
\end{equation*}
In the homoskedastic setting where $\sigma_C^2=\sigma_H^2$, the borrowing parameter simplifies directly to $\lambda=g(T)$.

\subsubsection{Implication for EBW Calibration}
\label{app:subsubsec:elastic-implication}
Consequently, conditional on the realized congruence statistic $T$, the elastic prior consistently yields an EBW estimator governed by an implied borrowing parameter $\lambda$, mapped directly from $g(T)$ (subject to a variance-ratio rescaling for Gaussian outcomes).
Because $T$ is purely a function of the observed data $(D_C,D_H)$, the elastic prior fundamentally operates as a data-adaptive EBW rule.
This structural equivalence places it within the same unified robust calibration framework as test-then-pool procedures and modified power priors, allowing BOND to systematically evaluate its worst-case properties.

\subsection{Unit-Information Prior}
\label{app:subsec_uip}

The unit-information prior (UIP) proposed by \citet{jin2021unit} constructs an informative prior by explicitly calibrating the amount of Fisher information contributed by multiple historical datasets.
Let $D_C$ denote the current data and let $D_{H,1},\dots,D_{H,K}$ denote $K$ independent historical datasets.
The UIP framework introduces two key components: (i) dataset relevance weights $\omega=(\omega_1,\dots,\omega_K)^\top$ on the simplex ($\omega_k\ge 0$ and $\sum_{k=1}^K \omega_k=1$), and (ii) a scalar precision parameter $M\ge 0$ interpreted as the total number of units of information to be borrowed from the aggregated historical sources.

Let $\hat\theta_k$ be a consistent point estimate from $D_{H,k}$ (typically the maximum likelihood estimate).
We define the unit information (the observed Fisher information per observation) as
\begin{equation*}
    I_U(\hat\theta_k)
    \coloneq
    -\frac{1}{n_k}
    \frac{\partial^2}{\partial\theta^2}\log L^{(k)}(\theta\mid D_{H,k})
    \bigg|_{\theta=\hat\theta_k}
    ,
\end{equation*}
where $n_k$ is the sample size and $L^{(k)}$ is the likelihood associated with $D_{H,k}$.
Leveraging a normal approximation, the UIP specifies a prior for the target parameter $\theta$ in the current study matching the following moments:
\begin{equation}
\label{eq:uip-moments}
    \mathbb{E}(\theta\mid M,\omega,D_{H,1:K})
    =
    \sum_{k=1}^K \omega_k \hat\theta_k,
    \quad
    \mathrm{Var}(\theta\mid M,\omega,D_{H,1:K})
    =
    \Bigl\{
        M\sum_{k=1}^K \omega_k I_U(\hat\theta_k)
    \Bigr\}^{-1}
    .
\end{equation}
This construction separates the prior location (accuracy via the weighted historical mean) from its scale (informativeness via the Fisher-information budget $M$).

\subsubsection{Continuous Outcomes}
\label{app:subsubsec:uip-cont}
Consider the conjugate normal-mean setting $Y_{C,i}\mid \mu \overset{\mathrm{i.i.d.}}{\sim} N(\mu,\sigma^2)$ with a known variance $\sigma^2$, and let the historical estimates be the sample means $\hat\mu_k=\bar{Y}_{H,k}$.
For a Gaussian mean, the unit information is simply $I_U(\hat\mu_k)=1/\sigma_k^2$ (where $\sigma_k^2$ is replaced by the sample variance $\hat\sigma_k^2$ in practice).
Consequently, the UIP is exactly Gaussian:
\begin{equation*}
    \mu \mid (M,\omega,D_{H,1:K})
    \sim
    N\bigl(
        \mu_{\omega},
        \eta^2
    \bigr),
    \quad
    \mu_{\omega}
    \coloneq
    \sum_{k=1}^K \omega_k \bar{Y}_{H,k}
    ,
    \quad
    \eta^2
    =
    \Bigl(
        M\sum_{k=1}^K \omega_k/\sigma_k^2
    \Bigr)^{-1}
    .
\end{equation*}
Updating this prior with the current likelihood yields a posterior mean that is strictly an affine combination of the current sample mean and the UIP historical center:
\begin{equation}
\label{eq:uip-post-mean-normal}
    \mathbb{E}[\mu\mid D_C,D_{H,1:K},M,\omega]
    =
    (1-w_{\mathrm{UIP}})\bar{Y}_C
    +
    w_{\mathrm{UIP}}\mu_{\omega}
    ,
    \quad
    w_{\mathrm{UIP}}
    =
    \frac{\sigma^2/\eta^2}{n_C+\sigma^2/\eta^2}
    .
\end{equation}
In the homoscedastic special case where $\sigma_k^2=\sigma^2$ (which implies $\eta^2=\sigma^2/M$), \eqref{eq:uip-post-mean-normal} simplifies directly to
\begin{equation*}
    \mathbb{E}[\mu\mid D_C,D_{H,1:K},M,\omega]
    =
    \frac{n_C\bar{Y}_C + M\mu_{\omega}}{n_C+M},
    \quad
    w_{\mathrm{UIP}}
    =
    \frac{M}{n_C+M}
    .
\end{equation*}
Therefore, conditional on $(M,\omega)$, the UIP perfectly induces an EBW estimator where the effective historical center is $\mu_{\omega}$ and the effective historical sample size is $\sigma^2/\eta^2$ (which equals the information budget $M$ under homoscedasticity).

\subsubsection{Binary Outcomes}
\label{app:subsubsec:uip-binary}
For a binary outcome $Y\in\{0,1\}$ with success probability $p$, the unit information evaluated at the historical maximum likelihood estimate $\hat p_k$ is $I_U(\hat p_k)=\{ \hat p_k(1-\hat p_k)\}^{-1}$.
To satisfy the UIP moments in \eqref{eq:uip-moments} within a conjugate framework, a Beta prior $p\mid (M,\omega,D_{H,1:K}) \sim \mathrm{Beta}(\alpha,\beta)$ is utilized.
Defining the target moments as $\mu_{\omega}\coloneq \sum_{k=1}^K \omega_k \hat p_k$ and $\eta^2 \coloneq \{M\sum_{k=1}^K \omega_k I_U(\hat p_k)\}^{-1}$, we match the mean and variance of the Beta distribution to obtain:
\begin{equation*}
    \alpha+\beta
    =
    \frac{\mu_{\omega}(1-\mu_{\omega})}{\eta^2}-1,
    \quad
    \alpha = \mu_{\omega}(\alpha+\beta),
    \quad
    \beta=(1-\mu_{\omega})(\alpha+\beta)
    .
\end{equation*}
After observing the current data of size $n_C$, the posterior mean again resolves into an EBW form:
\begin{equation}
\label{eq:uip-post-mean-bin}
    \mathbb{E}[p\mid D_C,D_{H,1:K},M,\omega]
    =
    (1-w_{\mathrm{UIP}})\bar{Y}_C
    +
    w_{\mathrm{UIP}}\mu_{\omega}
    ,
    \quad
    w_{\mathrm{UIP}}
    =
    \frac{\alpha+\beta}{n_C+\alpha+\beta}
    .
\end{equation}
Thus, for binary endpoints, the UIP functions as an EBW estimator where the effective historical sample size corresponds to the prior effective sample size $\alpha+\beta$, which is deterministic given the information budget $M$ and the historical unit information.

\subsubsection{Implication for EBW Calibration.}
\label{app:subsubsec:uip-implication}

\eqref{eq:uip-post-mean-normal} and \eqref{eq:uip-post-mean-bin} demonstrate that, conditional on the UIP hyperparameters $(M,\omega)$, the posterior mean fundamentally resides within the EBW class defined in Appendix~\ref{app:subsec:ebw-setup}, shrinking the current estimate toward a composite historical center.
In fully Bayesian implementations where $(M,\omega)$ are assigned hyperpriors (such as the UIP-Dirichlet or UIP-JS specifications in \citet{jin2021unit}), marginalizing over these parameters yields a posterior mean that is an integral of conditional affine forms.
This maintains its identity as an EBW estimator with a data-adaptive effective weight, thereby allowing BOND to provide distributionally robust operating characteristics for the UIP framework.

\subsection{Multisource Exchangeability Models}
\label{app:subsec:mem}

Multisource exchangeability models (MEMs), introduced by \citet{kaizer2018bayesian}, generalize the EXNEX paradigm \citep{neuenschwander2016robust} by performing Bayesian model averaging over all possible exchangeability patterns across multiple historical sources.
Consider a single trial arm, letting $D_C$ denote the current dataset and $\{D_h\}_{h=1}^H$ denote $H$ independent historical datasets, with corresponding sample means $\bar{Y}_C$ and $\bar{Y}_h$.
We define exchangeability indicators $S=(S_1,\ldots,S_H)\in\{0,1\}^H$, where $S_h=1$ indicates that the historical source $h$ is entirely exchangeable with the current population (thus sharing a common mean parameter $\mu$), whereas $S_h=0$ assigns source $h$ a fully independent mean parameter $\mu_h$.
Each binary configuration $s\in\{0,1\}^H$ identifies a distinct model $\Omega_s$, yielding a discrete model space of size $2^H$.

Assuming standard conjugate specifications alongside a diffuse base prior for $\mu$, the model-specific posterior for $\mu$ remains conjugate.
Crucially, its expectation simply pools the current data with the subset of historical sources declared exchangeable under $\Omega_s$.
For instance, in a Gaussian setting with known sampling variances for the sample means (matching the setup in \citet{kaizer2018bayesian}), let
$v_C \equiv \mathrm{Var}(\bar{Y}_C \mid \mu)$ and $v_h \equiv \mathrm{Var}(\bar{Y}_h \mid \mu_h)$.
The model-specific posterior mean then takes the precision-weighted form:
\begin{equation}
\label{eq:mem-conditional-mean}
    m(s)
    \coloneq
    \mathbb{E}[\mu \mid D_C, D_{1:H}, \Omega_s]
    =
    \frac{\bar{Y}_C/v_C + \sum_{h=1}^H s_h\bar{Y}_h/v_h}{1/v_C + \sum_{h=1}^H s_h/v_h}
    =
    (1-w(s))\bar{Y}_C + w(s)\bar{Y}_H(s)
    ,
\end{equation}
where the model-specific effective weight and effective historical mean are defined as
\begin{equation*}
    w(s)
    \coloneq
    \frac{\sum_{h=1}^H s_h/v_h}{1/v_C + \sum_{h=1}^H s_h/v_h}
    \in
    [0,1)
    ,
    \quad
    \bar{Y}_H(s)
    \coloneq
    \frac{\sum_{h=1}^H s_h\bar{Y}_h/v_h}{\sum_{h=1}^H s_h/v_h}
    .
\end{equation*}
We adopt the convention $w(s)=0$ (and consequently $m(s)=\bar{Y}_C$) when $\sum_{h=1}^H s_h/v_h=0$ (i.e., the entirely non-exchangeable model).
In the homoscedastic special case where $v_C=\sigma^2/n_C$ and $v_h=\sigma^2/n_h$, \eqref{eq:mem-conditional-mean} simplifies to standard sample-size pooling strictly over the exchangeable subset.

Let $\omega(s)\coloneq \mathbb{P}(\Omega_s\mid D_C,D_{1:H})$ represent the posterior model probability, derived via Bayes' theorem from the marginal likelihood and the prior $\pi(\Omega_s)$ (In the standard MEM framework, $\pi(\Omega_s)$ is typically formulated through independent inclusion probabilities on each $\{S_h\}_{h=1}^H$;
see \citet{kaizer2018bayesian}).
The marginal MEM posterior is constructed via Bayesian model averaging:
\begin{equation*}
    p(\mu\mid D_C,D_{1:H})
    =
    \sum_{s\in\{0,1\}^H}\omega(s)p(\mu\mid D_C,D_{1:H},\Omega_s)
    .
\end{equation*}
Consequently, the unconditional posterior mean is simply the weighted average of the model-specific posterior means:
\begin{equation}
\label{eq:mem-ebw}
    \mathbb{E}[\mu\mid D_C,D_{1:H}]
    =
    \sum_{s\in\{0,1\}^H}\omega(s)m(s)
    =
    (1-w_{\mathrm{MEM}})\bar{Y}_C + w_{\mathrm{MEM}}m_{\mathrm{MEM}}
    ,
\end{equation}
where the aggregated EBW is $w_{\mathrm{MEM}}\coloneq \sum_{s}\omega(s)w(s)=\mathbb{E}[w(S)\mid D_C,D_{1:H}]$, and the composite historical center is
\begin{equation*}
    m_{\mathrm{MEM}}
    \coloneq
    \frac{\sum_{s}\omega(s)w(s)\bar{Y}_H(s)}{w_{\mathrm{MEM}}}
    \quad
    \text{when } w_{\mathrm{MEM}}>0
    \quad
    (\text{otherwise $m_{\mathrm{MEM}}$ is arbitrary})
    .
\end{equation*}

\eqref{eq:mem-ebw} explicitly demonstrates that MEMs fundamentally operate as data-adaptive EBW estimators.
The current sample mean $\bar{Y}_C$ is systematically shrunk toward an effective historical composite mean, governed by the data-driven total weight $w_{\mathrm{MEM}} \in [0,1)$.
This exact structural correspondence confirms that MEMs, despite their complex model-averaging mechanics, fall seamlessly within the analytical scope of the proposed EBW robust calibration framework.

\subsection{Latent Exchangeability Prior}
\label{app:subsec:leap}

The latent exchangeability prior (LEAP) \citep{alt2024leap} enables individual-level dynamic borrowing by introducing latent class indicators for each historical observation.
In its general formulation, LEAP models the historical data $D_H=\{Y_{H,i}\}_{i=1}^{n_H}$ via a $K$-component mixture:
\begin{equation*}
    Y_{H,i}\mid (c_i=k,\theta_k)\sim f(\cdot\mid \theta_k),
    \quad
    \mathbb{P}(c_i=k\mid \gamma)=\gamma_k
    ,
    \quad
    k=1,\dots,K
    ,
\end{equation*}
where $c_i\in\{1,\dots,K\}$ are i.i.d.\ latent allocations and $\gamma=(\gamma_1,\dots,\gamma_K)$ lies on the simplex.
Following \citet{alt2024leap}, the first mixture component aligns with the current-data sampling model, such that the current data $D_C=\{Y_{C,i}\}_{i=1}^{n_C}$ satisfy $Y_{C,i}\sim f(\cdot\mid \theta_1)$.
Consequently, $\gamma_1$ represents the marginal probability that a historical individual is exchangeable with the current population (i.e., belongs to the component sharing $\theta_1$).

Let $c=(c_1,\dots,c_{n_H})$ denote the allocation vector and define the exchangeable historical subset as
\begin{equation*}
    \mathcal{I}_{\mathrm{ex}}(c)\coloneq\{i: c_i=1\}
    ,
    \quad
    n_H^{\mathrm{ex}}(c)
    \coloneq
    |\mathcal{I}_{\mathrm{ex}}(c)|
    .
\end{equation*}
As noted by \citet{alt2024leap}, $n_H^{\mathrm{ex}}(c)$ represents the sample size contribution (SSC) of the historical data to the posterior of $\theta_1$.
Let $\bar{Y}_C$ be the current sample mean, and let $\bar{Y}_H^{\mathrm{ex}}(c)$ be the empirical mean over the exchangeable subset:
\begin{equation*}
    \bar{Y}_H^{\mathrm{ex}}(c)\coloneq
    \begin{cases}
        \frac{1}{n_H^{\mathrm{ex}}(c)}\sum_{i\in\mathcal{I}_{\mathrm{ex}}(c)} Y_{H,i},
        & n_H^{\mathrm{ex}}(c)>0,
        \\
        0,
        & n_H^{\mathrm{ex}}(c)=0.
    \end{cases}
\end{equation*}

Conditional on $c$, the likelihood contribution for the current parameter $\theta_1$ involves only the current data and the exchangeable historical observations.
Therefore, letting the arm mean be denoted by $\mu$ (corresponding to $\theta_1$), under a conjugate one-parameter setting with a flat base prior (e.g., a normal mean with known variance, or a Bernoulli outcome with an improper $\mathrm{Beta}(0,0)$ limit), the conditional posterior mean takes the pooled form:
\begin{equation*}
    \mathbb{E}[\mu\mid D_C,D_H,c]
    =
    \frac{n_C\bar{Y}_C+n_H^{\mathrm{ex}}(c)\bar{Y}_H^{\mathrm{ex}}(c)}{n_C+n_H^{\mathrm{ex}}(c)}
    =
    (1-w(c))\bar{Y}_C+w(c)\bar{Y}_H^{\mathrm{ex}}(c),
    \quad
    w(c)
    \coloneq
    \frac{n_H^{\mathrm{ex}}(c)}{n_C+n_H^{\mathrm{ex}}(c)}
    .
\end{equation*}
Thus, for each realized allocation $c$, LEAP yields an EBW estimator whose EBW is determined by the latent SSC.

Marginalizing over the posterior distribution of $c$ (and any hyperparameters such as $\gamma$) yields a partition-averaged EBW representation for the marginal posterior mean:
\begin{equation}
\label{eq:leap-marg-ebw}
    \mathbb{E}[\mu\mid D_C,D_H]
    =
    \mathbb{E}\Bigl[
        (1-w(c))\bar{Y}_C+w(c)\bar{Y}_H^{\mathrm{ex}}(c)
        \Bigm|D_C,D_H
    \Bigr]
    =
    (1-w_{\mathrm{LEAP}})\bar{Y}_C+w_{\mathrm{LEAP}}\bar{Y}_{H}^{\mathrm{LEAP}}
    ,
\end{equation}
where
\begin{equation*}
    w_{\mathrm{LEAP}}
    \coloneq
    \mathbb{E}\bigl[
        w(c)\mid D_C,D_H
    \bigr],
    \quad
    \bar{Y}_{H}^{\mathrm{LEAP}}
    \coloneq
    \frac{\mathbb{E}\bigl[
        w(c)\bar{Y}_H^{\mathrm{ex}}(c)\mid D_C,D_H
    \bigr]}{w_{\mathrm{LEAP}}}
\end{equation*}
with the convention $\bar{Y}_{H}^{\mathrm{LEAP}}=\bar{Y}_C$ if $w_{\mathrm{LEAP}}=0$.
\eqref{eq:leap-marg-ebw} demonstrates that LEAP operates fundamentally as a data-adaptive EBW estimator.
Its effective historical mean is driven by the posterior-weighted exchangeable subset, seamlessly integrating it into the EBW-based calibration perspective adopted in this paper.

\subsection{Bayesian Hierarchical Modeling with Overlapping Indices}
\label{app:subsec:bhmoi}

\citet{lu2025overlapping} propose Bayesian Hierarchical Modeling with Overlapping Indices (BHMOI), a two-stage dynamic borrowing framework designed for settings with multiple related cohorts (e.g., basket trials).
First, BHMOI performs distribution clustering. Letting $f_i$ denote a subgroup-specific reference distribution (typically a noninformative posterior proxy for the subgroup parameter $\theta_i$), the method selects a partition $S^{(\mathrm{oci})*}=\{S_1,\dots,S_K\}$ that maximizes the Overlapping Clustering Index (OCI), a metric derived from the overlap coefficient $OVL(\cdot,\cdot)$ between distributions.
Given this selected partition, BHMOI computes the Overlapping Borrowing Index (OBI) for each cluster $S_m$, $m=1,\dots,K$, which quantifies within-cluster homogeneity on the standardized $[0,1]$ scale.
The core modeling step then links the borrowing-strength hyperparameters to this homogeneity measure. In the notation of \citet{lu2025overlapping}, the cluster-dependent prior $p(\eta_{mb}\mid S^{(\mathrm{oci})*})$ is replaced by $p(\eta_{mb}\mid s(\mathrm{OBI}_m))$ using a user-specified mapping $s(\cdot)$, thereby calibrating the degree of within-cluster borrowing according to the observed $\mathrm{OBI}_m$.

To connect BHMOI to the effective-borrowing-weight (EBW) perspective, consider the conjugate normal-endpoint specification used in \citet{lu2025overlapping}.
For a subgroup $i\in S_m$ with sample mean $\bar{Y}_i$ based on $n_i$ observations and known sampling variance $\sigma^2$, the model specifies:
\begin{equation*}
    \bar{Y}_i \mid \theta_i
    \sim
    N(\theta_i,\sigma^2/n_i),
    \quad
    \theta_i \mid \mu_m,\tau_m
    \sim
    N(\mu_m,\tau_m^{-1})
    ,
\end{equation*}
where $\mu_m$ is the cluster mean and $\tau_m$ is a cluster-specific precision parameter governing the extent of borrowing.
Conditional on $(\mu_m,\tau_m)$, the posterior mean of $\theta_i$ takes the standard affine shrinkage form:
\begin{equation}
\label{eq:bhmoi-shrinkage}
    \mathbb{E}[\theta_i\mid \bar{Y}_i,\mu_m,\tau_m]
    =
    (1-w_{i,m})\bar{Y}_i + w_{i,m}\mu_m,
    \quad
    w_{i,m}
    =
    \frac{\tau_m}{\tau_m+n_i/\sigma^2}
    .
\end{equation}
This demonstrates that BHMOI induces an EBW $w_{i,m}$ directed toward the cluster mean $\mu_m$, with borrowing effectively restricted to members of the selected cluster $S_m$.

Ultimately, the overlapping indices regulate the borrowing behavior through two distinct channels:
(i) the discrete partition $S^{(\mathrm{oci})*}$ determines from whom information is borrowed, and
(ii) the prior on $\tau_m$ conditionally depends on $s(\mathrm{OBI}_m)$, dictating how much borrowing occurs.
For instance, in their simulation study, $\tau_m$ is assigned a Gamma prior with a shape parameter $\alpha_m=s(\mathrm{OBI}_m)$;
consequently, higher within-cluster overlap yields a stochastically larger $\tau_m$, leading to stronger shrinkage in \eqref{eq:bhmoi-shrinkage}.
While a closed-form posterior mean is generally unavailable for nonconjugate endpoints (e.g., the binomial-logit model in \citealp{lu2025overlapping}), the underlying hierarchical structure functionally preserves this cluster-restricted shrinkage.
Therefore, BHMOI can be characterized by a data-adaptive EBW, conceptually aligning with the EBW formulation utilized in our DRO-based calibration.

\subsection{Nonparametric Bayesian Borrowing via Dirichlet Process Mixtures}
\label{app:subsec:npb}

Dirichlet process mixture (DPM) models offer a flexible, nonparametric mechanism for adaptive borrowing across multiple historical sources.
They achieve this by clustering study-specific parameters, effectively restricting information borrowing to historical studies that are empirically commensurate with the current study \citep{hupf2021bayesian,ohigashi2025nonparametric}.
Consider a single arm with one current dataset $D_C$ and $K$ historical datasets $D_1,\dots,D_K$.
Let $\theta_j$ denote the arm-level parameter in dataset $j \in \{C, 1, \dots, K\}$ (e.g., a mean for continuous outcomes or a response probability for binary outcomes), and write $(n_j,\bar{Y}_j)$ for the corresponding sample size and sample mean (or sufficient statistic) in dataset $j$.

A canonical DPM borrowing specification takes the form:
\begin{equation*}
    Y_{j,i}\mid \theta_j
    \sim
    f(\cdot \mid \theta_j),
    \quad
    \theta_j \mid G \stackrel{\mathrm{i.i.d.}}{\sim} G,
    \quad
    G \sim \mathrm{DP}(M,G_0),
    \quad
    j=0,1,\dots,K
    .
\end{equation*}
Because realizations of $G$ from a Dirichlet process are almost surely discrete, this formulation induces a random partition of the study indices $\{C, 1, \dots, K\}$.
Equivalently, there exist latent cluster labels $c_j$ and unique atoms $\{\theta_c^\star\}$ such that $\theta_j=\theta_{c_j}^\star$, meaning that datasets assigned to the same cluster share a common parameter.

Let $\Pi$ denote the induced partition, and define the subset of historical studies that are assigned to the same cluster as the current study by
\begin{equation*}
    S(\Pi)
    \coloneq
    \{k\in\{1,\dots,K\}\colon c_k=c_0\}
    .
\end{equation*}
Conditional on $\Pi$, borrowing is strictly selective:
only datasets within $S(\Pi)$ are pooled with the current dataset.
Under conjugate exponential-family sampling, the conditional posterior mean of $\theta_C$ is the standard conjugate update based exclusively on the data in the cluster containing $C$.
Specifically, when the base measure $G_0$ is specified as diffuse (or its contribution is subsumed into pseudo-counts), this conditional mean reduces to the pooled estimator:
\begin{equation}
\label{eq:dpm-conditional-mean}
    \mathbb{E}[\theta_0 \mid D_{0:K},\Pi]
    =
    \frac{
    n_0\bar{Y}_0+\sum_{k\in S(\Pi)} n_k \bar{Y}_k
    }{
    n_0+\sum_{k\in S(\Pi)} n_k
    }
    =
    (1-w(\Pi))\bar{Y}_0 + w(\Pi)\bar{Y}_{S(\Pi)}
    ,
\end{equation}
where $n_{S(\Pi)}\coloneq \sum_{k\in S(\Pi)} n_k$,
$w(\Pi)\coloneq n_{S(\Pi)}/(n_C+n_{S(\Pi)})$, and
$\bar{Y}_{S(\Pi)}\coloneq \sum_{k\in S(\Pi)} n_k\bar{Y}_k/n_{S(\Pi)}$ (with the convention $w(\Pi)=0$ and $\bar{Y}_{S(\Pi)}$ arbitrary when $S(\Pi)=\emptyset$).

Marginalizing over the posterior distribution of partitions yields a mixture over all possible subsets $S\subseteq\{1,\dots,K\}$ that may cluster with the current study:
\begin{equation*}
    \mathbb{E}[\theta_0 \mid D_{0:K}]
    =
    \sum_{S\subseteq\{1,\dots,K\}}
    \pi_S^{\mathrm{post}}
    \hat{\mu}_{\mathrm{pool}}(S)
    ,
    \quad
    \pi_S^{\mathrm{post}}
    \coloneq
    \mathbb{P}\{S(\Pi)=S \mid D_{0:K}\}
    ,
\end{equation*}
where $\hat{\mu}_{\mathrm{pool}}(S)$ is the pooled estimator in \eqref{eq:dpm-conditional-mean} with $S(\Pi)$ replaced by $S$.
Mirroring the MEM representation in \eqref{eq:mem-ebw}, this mixture expectation admits an affine formulation:
\begin{equation*}
    \mathbb{E}[\theta_0 \mid D_{0:K}]
    =
    (1-w_{\mathrm{eff}})\bar{Y}_0
    +
    w_{\mathrm{eff}} m_{\mathrm{eff}}
    ,
    \quad
    w_{\mathrm{eff}}
    \coloneq
    \sum_S \pi_S^{\mathrm{post}} w(S)
    ,
\end{equation*}
with
\(
    m_{\mathrm{eff}}
    \coloneq
    \{\sum_S \pi_S^{\mathrm{post}} w(S)\bar{Y}_S\}/w_{\mathrm{eff}}
\)
when $w_{\mathrm{eff}}>0$.
Consequently, DPM-based borrowing can be elegantly summarized by a single data-adaptive EBW $w_{\mathrm{eff}}$.
This establishes structural compatibility with the EBW formulation, enabling the direct application of the DRO-based calibration framework proposed in this paper.

\subsection{Extension to Two Arms and Control-Only Borrowing}
\label{app:subsec:two-arms}

The reductions outlined throughout this section operate independently within each specific trial arm.
Consequently, returning to the multi-arm target parameter $\hat{\theta}$, any borrowing procedure characterizing an arm-specific posterior mean as $(1-w_a)\bar{Y}_{C,a} + w_{a}\tilde{m}_{H,a}$ naturally aligns with the formulation $\hat{\theta}(\lambda) = \hat{\mu}_1(\lambda_1) - \hat{\mu}_0(\lambda_0)$ via Lemma~\ref{app:lem:w-lambda}.

This unified structure readily accommodates the frequent scenario of control-only borrowing.
When historical data is exclusively utilized for the control arm ($n_{H,1}=0$), we simply fix the experimental arm's effective weight to zero ($\lambda_1 = 0$) and apply the DRO calibration strictly to $\lambda_0$.
This demonstrates that BOND acts as a universal robust calibrator, capable of computing tight, worst-case uniform bounds regardless of whether the underlying dynamic borrowing architecture is symmetric, asymmetric, or single-armed.
\section{Extensions: Multi-Arm Trials, Multiple Historical Sources, and General Treatment Indices}
\label{app:sec:multiarm-multisource}

This section generalizes the core framework developed in Sections~\ref{sec:preliminaries} and \ref{sec:proposed-method} to accommodate 
(i) multi-arm trials with finitely many treatment levels,
(ii) dynamic borrowing from multiple independent historical sources, and
(iii) general treatment indices via measurable coarsening.
A key analytical insight is that the robust bias correction, driven by arm-wise worst-case mean shifts over Wasserstein balls, preserves its separable structure across multiple arms and data sources, keeping the optimization computationally trivial.

\subsection{Unified Data Structure: Arms and Sources}
\label{app:subsec:multiarm-data}

Let $\mathcal{A}$ be a finite set of treatment levels with $|\mathcal{A}|=K\ge 2$ (e.g., $\mathcal{A}=\{0,1,\dots,K-1\}$).
We define the set of available data sources as
\begin{equation*}
    \mathcal{J}
    \coloneq
    \{C\}\cup\{H_1,\dots,H_J\}
    ,
\end{equation*}
where $C$ denotes the current randomized trial and $H_1,\dots,H_J$ denote distinct historical sources.
For subject $i$ in source $j\in\mathcal{J}$, we observe
\begin{equation*}
    Z_{j,i}
    \coloneq
    (A_{j,i},X_{j,i},Y_{j,i})
    \in
    \mathcal{A}\times\mathcal{X}\times\mathcal{Y}
    .
\end{equation*}
For each arm $a\in\mathcal{A}$ and source $j\in\mathcal{J}$, let the arm-specific conditional law of $(X,Y)$ be
\begin{equation*}
    P_{j}^a
    \coloneq
    \mathcal{L}\bigl(
        (X,Y)\mid A=a,j
    \bigr)
    ,
\end{equation*}
with its corresponding mean outcome denoted by $\mu_j^a\coloneq \mathbb{E}_{P_j^a}[Y]$.
To capture potential noncommensurability, we define the mean shift of each historical source $H_k$ relative to the current trial for arm $a$ as:
\begin{equation*}
    \Delta_{k,a}
    \coloneq
    \mu_{H_k}^a-\mu_C^a
    .
\end{equation*}

For any prespecified contrast vector $c=(c_a)_{a\in\mathcal{A}}\in\mathbb{R}^{K}$, the target parameter in the current population is defined as
\begin{equation}
\label{app:eq:theta-contrast}
    \theta_C(c)
    \coloneq
    \sum_{a\in\mathcal{A}} c_a \mu_C^a
    .
\end{equation}
This general formulation subsumes standard two-arm comparisons;
for example, comparing an experimental arm $t$ against a control arm $0$ corresponds to setting $c_t=1$, $c_0=-1$, and $c_a=0$ otherwise.
We consider the one-sided hypothesis:
\begin{equation*}
    H_0(c)\colon \theta_C(c)\le 0
    \quad\text{vs.}\quad
    H_1(c)\colon \theta_C(c)>0
    .
\end{equation*}

\subsection{Multi-Source Effective Borrowing Estimators}
\label{app:subsec:multisource-ebw}

For each source $j\in\mathcal{J}$ and arm $a\in\mathcal{A}$, let $n_{j,a} \coloneq \sum_{i=1}^{n_j}\boldsymbol{1}\{A_{j,i}=a\}$ denote the sample size, and let $\bar{Y}_{j,a}$ be the sample mean (defined when $n_{j,a}\ge 1$).
By convention, if a specific arm is not included in a historical source, we set $n_{j,a}=0$.

We introduce a matrix of borrowing parameters:
\begin{equation*}
    \lambda
    \coloneq
    (\lambda_{k,a})_{k=1,\dots,J,\;a\in\mathcal{A}}
    \in
    \Lambda
    \coloneq
    \prod_{k=1}^J\prod_{a\in\mathcal{A}} [0,\Lambda_{k,a}]
    ,
\end{equation*}
where $\Lambda_{k,a}\in(0,\infty)$ are prespecified upper bounds.

For each arm $a\in\mathcal{A}$ with $n_{C,a}\ge 1$, the multi-source EBW estimator for the mean is:
\begin{equation*}
    \hat{\mu}_a(\lambda)
    \coloneq
    \frac{
        n_{C,a}\bar{Y}_{C,a}
        +
        \sum_{k=1}^J \lambda_{k,a} n_{H_k,a}\bar{Y}_{H_k,a}
    }{
        n_{C,a}
        +
        \sum_{k=1}^J \lambda_{k,a} n_{H_k,a}
    }
    .
\end{equation*}
This can be rewritten as a convex combination $\hat{\mu}_a(\lambda) = w_{C,a}(\lambda)\bar{Y}_{C,a} + \sum_{k=1}^J w_{k,a}(\lambda)\bar{Y}_{H_k,a}$, where the source-specific EBWs are:
\begin{equation*}
    w_{k,a}(\lambda)
    \coloneq
    \frac{\lambda_{k,a} n_{H_k,a}}{
        n_{C,a}+\sum_{\ell=1}^J \lambda_{\ell,a}n_{H_\ell,a}
    }
    \in[0,1)
    ,
    \quad
    k=1,\dots,J
    ,
\end{equation*}
and the implied weight for the current data is $w_{C,a}(\lambda) \coloneq 1-\sum_{k=1}^J w_{k,a}(\lambda) \in (0,1]$.

The induced estimator for the contrast is $\hat{\theta}(\lambda;c) \coloneq \sum_{a\in\mathcal{A}} c_a \hat{\mu}_a(\lambda)$.
A straightforward calculation reveals its expectation under heterogeneity:
\begin{equation*}
    \mathbb{E}\bigl[
        \hat{\theta}(\lambda;c)
    \bigr]
    =
    \theta_C(c)
    +
    \sum_{a\in\mathcal{A}}\sum_{k=1}^J c_a w_{k,a}(\lambda) \Delta_{k,a}
    .
\end{equation*}
This decomposition highlights that biases from multiple sources accumulate linearly according to the contrast vector and the EBWs.

\subsection{Wasserstein Ambiguity Sets for Multiple Sources}
\label{app:subsec:multisource-wass}

We utilize the same additive ground metric $d$ on $\mathcal{Z}=\mathcal{X}\times\mathcal{Y}$ as defined in \eqref{eq:ground-metric}.
For each historical source $k$ and arm $a$, we specify a tolerance radius $\rho_{k,a}\ge 0$ and define the Wasserstein ambiguity set centered at the corresponding current arm's distribution:
\begin{equation}
\label{app:eq:U-ka}
    \mathcal{U}_{k,a}(\rho_{k,a})
    \coloneq
    \bigl\{
        Q\in\mathcal{P}_1(\mathcal{Z})
        \colon
        W_1(Q,P_C^a)\le \rho_{k,a}
    \bigr\}
    .
\end{equation}
The foundational admissibility assumption is that $P_{H_k}^a \in \mathcal{U}_{k,a}(\rho_{k,a})$ for all $k=1,\dots,J$ and $a\in\mathcal{A}$.
The single-source formulation from the main text naturally emerges when $J=1$.

We define the corresponding worst-case mean shifts over \eqref{app:eq:U-ka}:
\begin{equation*}
    \Delta_{k,a}^{+}(\rho_{k,a})
    \coloneq
    \sup_{Q\in\mathcal{U}_{k,a}(\rho_{k,a})}
    \bigl(
        \mathbb{E}_Q[Y]-\mu_C^a
    \bigr)
    ,
    \quad
    \Delta_{k,a}^{-}(\rho_{k,a})
    \coloneq
    \inf_{Q\in\mathcal{U}_{k,a}(\rho_{k,a})}
    \bigl(
        \mathbb{E}_Q[Y]-\mu_C^a
    \bigr)
    .
\end{equation*}
Because the outcome space and metric remain identical, Proposition~\ref{prop:closed-form-delta} applies directly to each pair $(k,a)$, yielding explicit bounds in terms of $\rho_{k,a}$.

\subsection{Robust Bias Correction for a General Contrast}
\label{app:subsec:multisource-bias}

To ensure valid hypothesis testing, we must correct for the worst-case mean shift in the rejection direction.
For a fixed contrast $c$ and borrowing parameter $\lambda$, this is defined as:
\begin{equation}
\label{app:eq:bplus-multisource-def}
    b_{+}(\lambda;c)
    \coloneq
    \sup_{\substack{
        Q_{k,a}\in\mathcal{U}_{k,a}(\rho_{k,a})\\
        k=1,\dots,J,
        \;
        a\in\mathcal{A}
    }}
    \sum_{a\in\mathcal{A}}\sum_{k=1}^J
    c_a w_{k,a}(\lambda)
    \bigl(
        \mathbb{E}_{Q_{k,a}}[Y]-\mu_C^a
    \bigr)
    .
\end{equation}

\begin{proposition}[Closed-form of $b_+(\lambda;c)$ for multi-arm/multi-source]
\label{app:prop:bplus-multisource}
    For any $\lambda\in\Lambda$ and any contrast $c\in\mathbb{R}^K$, the worst-case bias is analytically tractable:
    \begin{equation*}
        b_{+}(\lambda;c)
        =
        \sum_{a\in\mathcal{A}}\sum_{k=1}^J
        c_a w_{k,a}(\lambda)
        \Delta_{k,a}^{\mathrm{sgn}(c_a)}(\rho_{k,a})
        ,
    \end{equation*}
    where $\Delta_{k,a}^{\mathrm{sgn}(c_a)}(\rho_{k,a}) = \Delta_{k,a}^{+}(\rho_{k,a})$ if $c_a\ge 0$, and $\Delta_{k,a}^{\mathrm{sgn}(c_a)}(\rho_{k,a}) = \Delta_{k,a}^{-}(\rho_{k,a})$ if $c_a<0$.
\end{proposition}
See Appendix~\ref{app:pf-prop-bplus-multisource} for the proof.
The logic mirrors Proposition~\ref{prop:bplus}.
Because the ambiguity sets are defined independently for each source and arm (a product ambiguity set), the global supremum decomposes into a sum of independent suprema, with the sign of the contrast coefficient $c_a$ dictating whether the positive or negative maximal drift is selected.

\subsection{Asymptotic Variance, Robust Size, and Robust Power}
\label{app:subsec:multisource-asymptotics}

The asymptotic variance of the contrast estimator $\hat{\theta}(\lambda;c)$ is given by:
\begin{equation*}
    s^2(\lambda;c)
    \coloneq
    \mathrm{Var}\bigl(\hat{\theta}(\lambda;c)\bigr)
    =
    \sum_{a\in\mathcal{A}} c_a^2
    \Biggl[
        w_{C,a}(\lambda)^2 \frac{\sigma_{C,a}^2}{n_{C,a}}
        +
        \sum_{k=1}^J
        w_{k,a}(\lambda)^2 \frac{\sigma_{H_k,a}^2}{n_{H_k,a}}
    \Biggr]
    ,
\end{equation*}
with the convention that a term with $n_{H_k,a}=0$ is set to $0$.
We denote its plug-in estimator by $\hat{s}(\lambda;c)$, utilizing sample variances $\hat{\sigma}_{j,a}^2$.

\begin{assumption}[Multi-source sampling and moments]
\label{app:ass:multisource-sampling}
    For each $(j,a)$ with $n_{j,a}\ge 1$, the outcomes $\{Y_{j,i}:A_{j,i}=a\}$ are i.i.d.\ with mean $\mu_j^a$ and variance $\sigma_{j,a}^2<\infty$.
    Data from different sources or arms are mutually independent.
\end{assumption}

\begin{assumption}[Multi-source asymptotic regime and nondegeneracy]
\label{app:ass:multisource-asymptotic}
    For each $a\in\mathcal{A}$, $n_{C,a}\to\infty$, and for each historical source $k$, either $n_{H_k,a}\to\infty$ or $n_{H_k,a}=0$.
    Furthermore, $\sigma_{C,a}^2>0$ for all $a$ where $c_a\neq 0$.
\end{assumption}

\begin{theorem}[Asymptotic distributionally robust size control: multi-arm/multi-source]
\label{app:thm:multisource-size}
    Fix $c\in\mathbb{R}^K$ and $\lambda\in\Lambda$.
    Define the robust Wald test:
    \begin{equation*}
        \varphi_{\lambda,c}
        \coloneq
        \boldsymbol{1}\Biggl\{
            \frac{\hat{\theta}(\lambda;c)-\tilde{b}_{+}(\lambda;c)}{\hat{s}(\lambda;c)}
            \ge
            z_{1-\alpha}
        \Biggr\}
        ,
    \end{equation*}
    where $\tilde{b}_{+}(\lambda;c)$ represents the benchmark bias bound $b_+(\lambda;c)$ or its valid plug-in estimate for binary outcomes.
    If the null $H_0(c)\colon\theta_C(c)\le 0$ holds and $P_{H_k}^a\in\mathcal{U}_{k,a}(\rho_{k,a})$ for all $(k,a)$,
    then under Assumptions~\ref{app:ass:multisource-sampling} and~\ref{app:ass:multisource-asymptotic},
    \begin{equation*}
        \limsup_{\min_{a}n_{C,a}\to\infty}
        \mathbb{P}(\varphi_{\lambda,c}=1)
        \le
        \alpha
        .
    \end{equation*}
\end{theorem}
See Appendix~\ref{app:pf-thm-multisource-size} for the proof.

Crucially, extending this framework to multiple sources does not change the core mechanism for selecting $\lambda$.
Fix a target alternative $\theta_1>0$ for $\theta_C(c)$.
Let $D_{k,a}(\rho_{k,a}) \coloneq \Delta_{k,a}^{+}(\rho_{k,a})-\Delta_{k,a}^{-}(\rho_{k,a})$ denote the drift range.
The robust noncentrality parameter becomes:
\begin{equation*}
    \kappa(\lambda;c)
    \coloneq
    \frac{
        \theta_1
        -
        \sum_{a\in\mathcal{A}} |c_a|
        \sum_{k=1}^J w_{k,a}(\lambda) D_{k,a}(\rho_{k,a})
    }{
        s(\lambda;c)
    }.
\end{equation*}
In line with our primary motivation, selecting the optimal borrowing parameters entails maximizing $\kappa(\lambda;c)$ over the compact set $\Lambda$.
This objective seamlessly manages the trade-off among various historical sources, dynamically down-weighting sources with large specified discrepancies $\rho_{k,a}$ or high variances, while retaining those that offer meaningful efficiency gains.

\subsection{Categorical and Continuous Treatment Indices}
\label{app:subsec:treatment-index}

The main text assumes a binary arm label $A\in\{0,1\}$;
the multi-arm extension above covers any finite categorical treatment.
If the original treatment or exposure $\tilde{A}$ takes values in a continuous or complex measurable space, the BOND methodology can still be applied through measurable coarsening.

\begin{remark}[Measurable coarsening of a general treatment index]
\label{app:rem:coarsening}
    Let $(\tilde{\mathcal{A}},\mathcal{B}_{\tilde{\mathcal{A}}})$ be a measurable space and let $\tilde{A}\in\tilde{\mathcal{A}}$ be the original continuous exposure (e.g., dosage).
    Fix a finite set $\mathcal{A}$ and a measurable function $g\colon\tilde{\mathcal{A}}\to\mathcal{A}$ mapping exposures to discrete bins or strata.
    By defining the induced categorical treatment $A\coloneq g(\tilde{A})\in\mathcal{A}$, all theoretical guarantees established in Appendix~\ref{app:sec:multiarm-multisource} remain valid.
    The target parameter \eqref{app:eq:theta-contrast} subsequently represents a contrast between these coarsened exposure strata.
\end{remark}
This coarsening strategy provides a straightforward pathway to accommodate continuous doses without necessitating complex functional extensions of the EBW principle, preserving the robustness guarantees.

\section{Two-Sided Distributionally Robust Inference}
\label{app:sec:two-sided}

This section formalizes the extension of Section~\ref{subsec:test} to handle the two-sided hypothesis test $H_0^{\pm}\colon\theta_C=0$ versus $H_1^{\pm}\colon\theta_C\neq 0$, and provides the construction of distributionally robust confidence intervals.

\subsection{Worst-Case Bias in the Negative Rejection Direction}
\label{app:subsec:two-side-worst-case}

Recall the positive-direction worst-case bias $b_+(\lambda)$ defined in \eqref{eq:bplus-def}.
By symmetry, we define the negative-direction worst-case bias as:
\begin{equation*}
    b_-(\lambda)
    \coloneq
    \inf_{\substack{Q_1\in\mathcal{U}_1(\rho_1)\\ Q_0\in\mathcal{U}_0(\rho_0)}}\Bigl[
        w_1(\lambda_1)\bigl(
            \mathbb{E}_{Q_1}[Y] - \mu_C^1
        \bigr)
        -
        w_0(\lambda_0)\bigl(
            \mathbb{E}_{Q_0}[Y] - \mu_C^0
        \bigr)
    \Bigr]
    .
\end{equation*}
This quantity represents the worst-case bias for rejections in the lower tail.

\begin{proposition}[Closed-form of $b_-(\lambda)$]
\label{app:prop:bminus}
    For any $\lambda\in\Lambda$ with $w_a(\lambda_a)\ge 0$,
    \begin{equation*}
        b_-(\lambda)
        =
        w_1(\lambda_1)\Delta_1^-(\rho_1)
        -
        w_0(\lambda_0)\Delta_0^+(\rho_0)
        .
    \end{equation*}
    Specifically, applying the bounds from Proposition~\ref{prop:closed-form-delta} yields:
    \begin{equation*}        
        b_-(\lambda)=
        \begin{cases}
            -\{w_{1}(\lambda_1)\rho_1+w_0(\lambda_0)\rho_0\},
            & \mathcal{Y}=\mathbb{R},
            \\[4pt]
            -\Bigl[
                w_1(\lambda_1)\min\{\rho_1,\mu_C^1\}
                +
                w_0(\lambda_0)\min\{\rho_0,1-\mu_C^0\}
            \Bigr],
            & \mathcal{Y}=\{0,1\}.
        \end{cases}
    \end{equation*}
\end{proposition}
See Appendix~\ref{app:pf-prop-bminus} for the proof.

\subsection{Two-Sided Robust Wald Test and Confidence Intervals}
\label{app:subsec:two-side-test}

Let $z_{1-\alpha/2}$ denote the $(1-\alpha/2)$ quantile of the standard normal distribution.
For a fixed $\lambda\in\Lambda$, the two-sided robust Wald test is defined as rejecting $H_0^\pm$ if:
\begin{equation}
\label{app:eq:robust-test-two-sided}
    \varphi^{\pm}_{\lambda}
    \coloneq
    \boldsymbol{1}\Biggl\{
        \frac{\hat{\theta}(\lambda)-\tilde{b}_{+}(\lambda)}{\hat{s}(\lambda)}
        \ge
        z_{1-\alpha/2}
        \;\text{ or }\;
        \frac{\hat{\theta}(\lambda)-\tilde{b}_{-}(\lambda)}{\hat{s}(\lambda)}
        \le
        -z_{1-\alpha/2}
    \Biggr\}
    ,
\end{equation}
where $\tilde{b}_+(\lambda)=b_+(\lambda)$ and $\tilde{b}_-(\lambda)=b_-(\lambda)$ for the benchmark formulations.
For binary outcomes, the practical implementations replace the unknown parameters $\mu_C^a$ with their sample analogues $\bar{Y}_{C,a}$ within the closed-form bounds.

Equivalently, this test allows us to construct a distributionally robust two-sided $(1-\alpha)$ confidence interval for $\theta_C$:
\begin{equation}
\label{app:eq:robust-ci-two-sided}
    \mathrm{CI}^{\pm}_{\lambda}
    \coloneq
    \Bigl[
        \hat{\theta}(\lambda)-\tilde{b}_+(\lambda)-z_{1-\alpha/2}\hat{s}(\lambda),
        \;
        \hat{\theta}(\lambda)-\tilde{b}_-(\lambda)+z_{1-\alpha/2}\hat{s}(\lambda)
    \Bigr]
    .
\end{equation}
The duality between testing and interval estimation holds exactly:
we reject $H_0^\pm\colon \theta_C=0$ if and only if $0\notin \mathrm{CI}^{\pm}_{\lambda}$. 
Because $\tilde{b}_+(\lambda) \ge 0$ and $\tilde{b}_-(\lambda) \le 0$, the interval is strictly wider than a naive, uncorrected Wald interval, explicitly reflecting the epistemic uncertainty originating from the ambiguous external data.

\subsection{Asymptotic Distributionally Robust Size Control}
\label{app:subsec:two-side-worst-size-control}

To ensure the validity of the testing procedure utilizing plug-in parameters for binary outcomes, we first verify consistency.

\begin{lemma}[Consistency of plug-in $b_-(\lambda)$ for binary outcomes]
\label{app:lem:consistency-bminus}
    Assume $\mathcal{Y}=\{0,1\}$ and fix $\lambda\in\Lambda$.
    Let $\hat b_-(\lambda)$ be the plug-in version of $b_-(\lambda)$ obtained by replacing $\mu_C^a$ with $\bar{Y}_{C,a}$ in the expressions of Proposition~\ref{app:prop:bminus}.
    Under Assumptions~\ref{ass:sampling} and \ref{ass:asymptotic},
    \begin{equation*}        
        \hat b_-(\lambda)
        \longrightarrow_p
        b_-(\lambda)
        .
    \end{equation*}
\end{lemma}
See Appendix~\ref{app:pf-lem-consistency-bminus} for the proof.

This consistency allows us to formalize the robust control of the type I error rate in the two-sided paradigm.

\begin{theorem}[Asymptotic distributionally robust size control: two-sided]
\label{app:thm:robust-size-two-sided}
    Fix $\lambda\in\Lambda$ and let $\varphi_\lambda^{\pm}$ be defined as in \eqref{app:eq:robust-test-two-sided}.
    For any valid underlying configuration satisfying $\theta_C=0$ and $P_H^a\in\mathcal{U}_a(\rho_a)$ for $a\in\{0,1\}$, under Assumptions~\ref{ass:sampling} and \ref{ass:asymptotic},
    \begin{equation*}        
        \limsup_{\min_a n_{C,a}\to\infty}
        \mathbb{P}\bigl(
            \varphi_\lambda^{\pm}=1
        \bigr)
        \le
        \alpha
        .
    \end{equation*}
\end{theorem}
See Appendix~\ref{app:pf-thm-robust-size-two-sided} for the proof.
\section{Proofs}
\label{app:sec:proofs}

\subsection{Technical Lemmas}
\label{app:subsec:technical-lemmas}

\begin{lemma}[Lipschitz bound for expectation differences]
\label{lem:lipschitz-bound}
    Let $(\mathcal{Z},d)$ be a metric space and let $P,Q\in\mathcal{P}_1(\mathcal{Z})$.
    If $f\colon \mathcal{Z}\to\mathbb{R}$ is $L$-Lipschitz, then
    \begin{equation*}
        \bigl|
            \mathbb{E}_P[f] - \mathbb{E}_Q[f]
        \bigr|
        \le
        LW_1(P,Q)
        .
    \end{equation*}
\end{lemma}
Lemma~\ref{lem:lipschitz-bound} connects distributional discrepancy, measured by the 1-Wasserstein distance, to worst-case perturbations of expectations for Lipschitz functionals.
This provides the fundamental technical device for deriving explicit, nonparametric bias bounds over Wasserstein ambiguity sets.

\begin{proof}[Proof of Lemma~\ref{lem:lipschitz-bound}]
    Let $\pi\in\Pi(P,Q)$ be an arbitrary coupling with marginals $P$ and $Q$, and let $(Z,Z')\sim\pi$.
    By definition, $\mathbb{E}_P[f]-\mathbb{E}_Q[f] = \mathbb{E}_\pi[f(Z)-f(Z')]$.
    Applying Jensen's inequality and the $L$-Lipschitz property of $f$, we obtain
    \begin{equation*}        
        \bigl|
            \mathbb{E}_P[f]-\mathbb{E}_Q[f]
        \bigr|
        \le
        \mathbb{E}_\pi\bigl[
            |f(Z)-f(Z')|
        \bigr]
        \le
        L\mathbb{E}_\pi\bigl[
            d(Z,Z')
        \bigr]
        .
    \end{equation*}
    Taking the infimum over all valid couplings $\pi\in\Pi(P,Q)$ yields the desired bound $\bigl|\mathbb{E}_P[f]-\mathbb{E}_Q[f]\bigr|\le LW_1(P,Q)$.
\end{proof}

\begin{lemma}[Consistency]
\label{lem:consistency}
    Under Assumptions~\ref{ass:sampling} and \ref{ass:asymptotic}, for any fixed $\lambda\in\Lambda$,
    \begin{equation*}
        \hat{s}(\lambda)
        \longrightarrow_p
        s(\lambda)
        .
    \end{equation*}
    Furthermore, if $\mathcal{Y} = \{0,1\}$ and $\hat{b}_+(\lambda)$ is the plug-in bias correction obtained by replacing $\mu_C^a$ with $\bar{Y}_{C,a}$ in Proposition~\ref{prop:bplus}, then
    \begin{equation*}
        \hat{b}_+(\lambda)
        \longrightarrow_p
        b_+(\lambda)
        .
    \end{equation*}
\end{lemma}
Lemma~\ref{lem:consistency} ensures that the feasible, fully data-driven implementation is asymptotically equivalent to the oracle version.

\begin{proof}[Proof of Lemma~\ref{lem:consistency}]
    Consider any observed group $(j,a)$ such that $n_{j,a}\to\infty$.
    By Assumption~\ref{ass:sampling}, $\mathbb{E}[Y^2\mid A=a,j]<\infty$.
    The sample variance can be expressed as
    \begin{equation*}        
        \hat{\sigma}_{j,a}^2
        =
        \frac{n_{j,a}}{n_{j,a}-1}\Biggl(
            \frac{1}{n_{j,a}}\sum_{i:A_{j,i}=a}Y_{j,i}^2-\bar{Y}_{j,a}^2
        \Biggr)
        .
    \end{equation*}
    By the Weak Law of Large Numbers (WLLN), $(1/n_{j,a})\sum Y_{j,i}^2 \to_{p}\mathbb{E}[Y^2\mid A=a,j]$ and $\bar{Y}_{j,a}\to_p\mu_j^a$.
    Because $n_{j,a}/(n_{j,a}-1)\to 1$, Slutsky's theorem dictates that
    \begin{equation*}
        \hat{\sigma}_{j,a}^2
        \longrightarrow_{p}
        \mathbb{E}[Y^2\mid A=a,j]-(\mu_j^a)^2
        =
        \sigma_{j,a}^2
        .
    \end{equation*}
    
    Fixing $\lambda\in\Lambda$, the weights $w_a(\lambda_a)$ are deterministic constants.
    Thus, each summand in $\hat{s}^2(\lambda)$ converges in probability to the corresponding summand in $s^2(\lambda)$.
    Since the sum is finite, $\hat{s}^2(\lambda)\to_{p}s^2(\lambda)$.
    Assumption~\ref{ass:asymptotic} ensures $s^2(\lambda)>0$, making the square-root function continuous at $s^2(\lambda)$.
    The continuous mapping theorem (CMT) therefore yields $\hat{s}(\lambda)\to_{p}s(\lambda)$.
    
    For binary outcomes, $\bar{Y}_{C,a}\to_{p}\mu_C^a$.
    The functions $\mu\mapsto \min\{\rho,1-\mu\}$ and $\mu\mapsto \min\{\rho,\mu\}$ are continuous on $[0,1]$.
    By the CMT, the plug-in estimator $\hat{b}_+(\lambda)$ evaluated via Proposition~\ref{prop:bplus} converges in probability to the theoretical bound $b_+(\lambda)$.
\end{proof}

\subsection{Proof of Proposition~\ref{prop:closed-form-delta}}
\label{app:pf-prop-closed-form-delta}
\begin{proof}[Proof of Proposition~\ref{prop:closed-form-delta}]
    The analysis proceeds arm-by-arm;
    we suppress the arm index $a$ for brevity.
    
    \textit{Part (i): $\mathcal{Y}=\mathbb{R}$.}
    
    Under the ground metric \eqref{eq:ground-metric}, the projection $f(x,y)=y$ is 1-Lipschitz because $|y-y'|\le d_{\mathcal{X}}(x,x')+|y-y'|=d((x,y),(x',y'))$.
    By Lemma~\ref{lem:lipschitz-bound}, for any $Q\in\mathcal{U}(\rho)$,
    \begin{equation*}
        \mathbb{E}_Q[Y]-\mu_C
        \le
        W_1(Q,P_C)
        \le
        \rho
        ,        
    \end{equation*}
    yielding $\Delta^+(\rho)\le \rho$.
    Applying the same logic to $-f$ yields $\Delta^-(\rho)\ge -\rho$.
    
    To demonstrate attainability, define the translation map $T_+(x,y)\coloneq (x,y+\rho)$ and construct $Q_+\coloneqq P_C\circ T_+^{-1}$.
    Using the deterministic coupling $(Z,Z')$ where $Z=(X,Y)\sim P_C$ and $Z'=T_+(Z)$, we observe
    \begin{equation*}        
        d(Z,Z')
        =
        d_{\mathcal{X}}(X,X)+|Y-(Y+\rho)|
        =
        \rho
        \quad
        \text{a.s.}
    \end{equation*}
    This implies $W_1(Q_+,P_C)\le \rho$, verifying $Q_+\in\mathcal{U}(\rho)$.
    Furthermore, $\mathbb{E}_{Q_+}[Y] = \mathbb{E}_{P_C}[Y+\rho] = \mu_C+\rho$, establishing $\Delta^+(\rho)\ge \rho$.
    Consequently, $\Delta^+(\rho)=\rho$.
    An analogous argument using $T_-(x,y)=(x,y-\rho)$ confirms $\Delta^-(\rho)=-\rho$.
    
    \textit{Part (ii): $\mathcal{Y}=\{0,1\}$.}

    Let $p \coloneq \mu_C = P_C(\mathcal{X}\times\{1\}) \in[0,1]$.
    For any $Q\in\mathcal{U}(\rho)$, the 1-Lipschitz property of $f(x,y)=y$ guarantees $\mathbb{E}_Q[Y]-p \le \rho$.
    Because $Y$ is binary, $\mathbb{E}_Q[Y]\le 1$ trivially holds.
    Thus,
    \begin{equation*}
        \sup_{Q\in\mathcal{U}(\rho)}\mathbb{E}_Q[Y]
        \le
        \min\{p+\rho,1\}
        .
    \end{equation*}
    If $p=1$, the bound is trivially $1$, achieved by $Q=P_C$.
    Suppose $p<1$.
    Define $t \coloneq \min\{\rho,1-p\} \in [0,1-p]$ and the transition probability $\eta \coloneq t/(1-p) \in [0,1]$.
    Construct a Markov kernel $K$ on $\mathcal{Z}$ that deterministically maps $(x,1)$ to $(x,1)$, and maps $(x,0)$ to $(x,1)$ with probability $\eta$ and to $(x,0)$ with probability $1-\eta$.
    Let $Q(E)\coloneq \int K(z,E)P_C(\dd z)$ be the induced probability measure.
    By construction, the marginal distribution of $X$ is preserved, and
    \begin{equation*}
        Q(\mathcal{X}\times\{1\})
        =
        P_C(\mathcal{X}\times\{1\})
        +
        \eta P_C(\mathcal{X}\times\{0\})
        =
        p+\eta(1-p)
        =
        p+t
        =
        \min\{p+\rho,1\}
        .
    \end{equation*}
    
    To verify $Q\in\mathcal{U}(\rho)$, consider the standard coupling $\pi(\dd z,\dd z') \coloneq P_C(\dd z)K(z,\dd z')$.
    Because $K$ perturbs only the outcome from $0$ to $1$ with probability $\eta$, the transportation cost is strictly determined by this mass shift:
    \begin{equation*}
        \int_{\mathcal{Z}\times\mathcal{Z}} d(z,z')\pi(\dd z,\dd z')
        =
        \eta P_C(\mathcal{X}\times\{0\})
        =
        \eta(1-p)
        =
        t
        \le
        \rho
        .
    \end{equation*}
    Thus $W_1(Q,P_C)\le \rho$, placing $Q\in\mathcal{U}(\rho)$ and confirming that the supremum evaluates to $\min\{p+\rho,1\}$.

    By symmetric reasoning applied to $-f$, we have $\mathbb{E}_Q[Y]\ge p-\rho$ alongside the inherent bound $\mathbb{E}_Q[Y]\ge 0$.
    Constructing a kernel $K_-$ that shifts mass from $Y=1$ to $Y=0$ with probability $\eta_- \coloneq \min\{\rho,p\}/p$ yields a valid $Q_-\in\mathcal{U}(\rho)$ attaining $\max\{p-\rho,0\}$, thereby completing the proof.
\end{proof}

\subsection{Proof of Proposition~\ref{prop:bplus}}
\label{app:pf-prop-bplus}
\begin{proof}[Proof of Proposition~\ref{prop:bplus}]
    By definition \eqref{eq:bplus-def}, we evaluate a supremum over the product ambiguity set $\mathcal{U}_1(\rho_1)\times \mathcal{U}_0(\rho_0)$.
    Because the objective is completely separable into arm-specific terms, we have:
    \begin{equation*}        
        b_+(\lambda)
        =
        \sup_{Q_1\in\mathcal{U}_1(\rho_1)} \Bigl[
            w_1(\lambda_1)\bigl(
                \mathbb{E}_{Q_1}[Y]-\mu_C^1
            \bigr) 
        \Bigr]
        +
        \sup_{Q_0\in\mathcal{U}_0(\rho_0)} \Bigl[
            -w_0(\lambda_0)\bigl(
                \mathbb{E}_{Q_0}[Y]-\mu_C^0
            \bigr)
        \Bigr]
        .
    \end{equation*}
    Given that the effective weights are non-negative ($w_a(\lambda_a) \ge 0$), the first supremum trivially evaluates to $w_1(\lambda_1) \Delta_1^+(\rho_1)$.
    For the second term, we apply the elementary identity $\sup_{x \in S} (-cx) = -c \inf_{x \in S} x$ for $c \ge 0$, yielding $-w_0(\lambda_0) \Delta_0^{-}(\rho_0)$.
    Substituting the analytic bounds derived in Proposition~\ref{prop:closed-form-delta} directly provides the stated forms.
\end{proof}

\subsection{Proof of Proposition~\ref{prop:clt}}
\label{app:pf-prop-clt}
\begin{proof}[Proof of Proposition~\ref{prop:clt}]
    Fix $\lambda\in\Lambda$.
    Define deterministic coefficients $c_{C,1} \coloneq 1-w_1(\lambda_1)$, $c_{H,1} \coloneq w_1(\lambda_1)$, $c_{C,0} \coloneq -(1-w_0(\lambda_0))$, and $c_{H,0} \coloneq -w_0(\lambda_0)$.
    Let $\mathcal{I} \coloneq \{(C,0),(C,1)\} \cup \{(H,a)\colon n_{H,a}\ge 1\}$ index the observed trial arms.
    The treatment effect estimator is the linear combination $\hat{\theta}(\lambda) = \sum_{(j,a)\in\mathcal{I}} c_{j,a}\bar{Y}_{j,a}$.
    Its expectation resolves straightforwardly to $\mathbb{E}[\hat{\theta}(\lambda)] = \theta_C + w_1(\lambda_1)\Delta_1 - w_0(\lambda_0)\Delta_0$, aligning with \eqref{eq:theta-mean}.
    
    To establish asymptotic normality, we represent the centered statistic as a sum of independent random variables:
    \begin{equation*}
        \hat{\theta}(\lambda)
        -
        \mathbb{E}[\hat{\theta}(\lambda)]
        =
        \sum_{(j,a)\in\mathcal{I}}\sum_{i=1}^{n_{j,a}} \xi_{j,a,i}
        ,
    \end{equation*}
    where
    \begin{equation*}
        \xi_{j,a,i}
        \coloneq
        \frac{c_{j,a}}{n_{j,a}} \bigl(
            Y_{j,a,i}-\mu_j^a
        \bigr)
        .
    \end{equation*}
    The terms $\xi_{j,a,i}$ are independent, mean-zero, and possess variance $\mathrm{Var}(\xi_{j,a,i}) = c_{j,a}^2\sigma_{j,a}^2/n_{j,a}^2$.
    Summing these variances confirms that $\mathrm{Var}(\sum \xi_{j,a,i}) = s^2(\lambda)$, which is strictly positive under Assumption~\ref{ass:asymptotic}.
    
    We verify the Lindeberg condition for this triangular array.
    Fix $\varepsilon>0$ and define the Lindeberg sum:
    \begin{equation*}
        L_n
        \coloneq
        \frac{1}{s^2(\lambda)}
        \sum_{(j,a)\in\mathcal{I}}\sum_{i=1}^{n_{j,a}}
        \mathbb{E}\Bigl[
            \xi_{j,a,i}^2 \boldsymbol{1}\bigl\{
                |\xi_{j,a,i}|>\varepsilon s(\lambda)
            \bigr\}
        \Bigr]
        .
    \end{equation*}
    Consider any $(j,a) \in \mathcal{I}$ where $\sigma_{j,a}^2 > 0$.
    Because $s^2(\lambda) \ge c_{j,a}^2\sigma_{j,a}^2/n_{j,a}$, we have $s(\lambda) \ge |c_{j,a}|\sigma_{j,a}/\sqrt{n_{j,a}}$.
    Consequently, the indicator condition implies $|Y_{j,a,i}-\mu_j^a| > \varepsilon s(\lambda) n_{j,a} / |c_{j,a}| \ge \varepsilon \sigma_{j,a} \sqrt{n_{j,a}}$.
    Exploiting the identically distributed nature of the observations within each group, the group's contribution to $L_n$ is bounded by:
    \begin{equation*}
        \frac{n_{j,a}}{s^2(\lambda)}
        \mathbb{E}\Bigl[
            \xi_{j,a,1}^2 \boldsymbol{1}\bigl\{
                |\xi_{j,a,1}| > \varepsilon s(\lambda)
            \bigr\}
        \Bigr]
        \le
        \frac{1}{\sigma_{j,a}^2}
        \mathbb{E}\Bigl[
            (Y_{j,a,1}-\mu_j^a)^2 \boldsymbol{1}\bigl\{
                |Y_{j,a,1}-\mu_j^a| > \varepsilon \sigma_{j,a}\sqrt{n_{j,a}}
            \bigr\}
        \Bigr]
        .
    \end{equation*}
    Since $\mathbb{E}[(Y_{j,a,1}-\mu_j^a)^2]=\sigma_{j,a}^2<\infty$ (Assumption~\ref{ass:sampling}) and $\sqrt{n_{j,a}} \to \infty$ (Assumption~\ref{ass:asymptotic}), the Dominated Convergence Theorem ensures this upper bound vanishes as $n \to \infty$.
    Because $\mathcal{I}$ is finite, $L_n \to 0$.
    The Lindeberg-Feller Central Limit Theorem therefore dictates that
    \begin{equation*}
        \frac{\hat{\theta}(\lambda)-\mathbb{E}[\hat{\theta}(\lambda)]}{s(\lambda)}
        \longrightarrow_d
        N(0,1)
        ,
    \end{equation*}
    completing the proof.
\end{proof}

\subsection{Proof of Theorem~\ref{thm:robust-size}}
\label{app:pf-thm-robust-size}
\begin{proof}[Proof of Theorem~\ref{thm:robust-size}]
    Fix $\lambda\in\Lambda$, and assume an arbitrary null configuration where $\theta_C\le 0$ and $P_H^a\in\mathcal{U}_a(\rho_a)$.
    Let $\Delta_a=\mu_H^a-\mu_C^a$ denote the resulting drifts.
    By Proposition~\ref{prop:clt} and Lemma~\ref{lem:consistency}, Slutsky's theorem implies:
    \begin{equation*}        
        \frac{\hat{\theta}(\lambda)-\bigl(\theta_C+w_1\Delta_1-w_0\Delta_0\bigr)}{\hat{s}(\lambda)}
        \longrightarrow_d
        N(0,1)
        .
    \end{equation*}
    We analyze the benchmark bias-corrected statistic:
    \begin{equation*}        
        \frac{\hat{\theta}(\lambda)-b_+(\lambda)}{\hat{s}(\lambda)}
        =
        \frac{\hat{\theta}(\lambda)-\bigl(\theta_C+w_1\Delta_1-w_0\Delta_0\bigr)}{\hat{s}(\lambda)}
        +
        \frac{\theta_C+w_1\Delta_1-w_0\Delta_0-b_+(\lambda)}{\hat{s}(\lambda)}
        .
    \end{equation*}
    By the supremum definition of $b_+(\lambda)$ in \eqref{eq:bplus-def} and the hypothesis $\theta_C \le 0$, the drift term $\theta_C+w_1\Delta_1-w_0\Delta_0-b_+(\lambda)$ is strictly non-positive.
    Coupled with $\hat{s}(\lambda) \to_p s(\lambda) > 0$, the statistic converges in distribution to $N(m,1)$ for some $m \le 0$.
    Since the upper-tail probability $\mathbb{P}(N(m,1)\ge z_{1-\alpha})$ is non-increasing in $m$, we have:
    \begin{equation*}        
        \limsup_{\min n_{C,a} \to \infty}
        \mathbb{P}\Biggl(
            \frac{\hat{\theta}(\lambda)-b_+(\lambda)}{\hat{s}(\lambda)} \ge z_{1-\alpha}
        \Biggr)
        \le
        \alpha
        .
    \end{equation*}

    For the practical implementation utilizing the plug-in $\hat{b}_+(\lambda)$ (binary outcomes), Lemma~\ref{lem:consistency} yields $\hat{b}_+(\lambda)-b_+(\lambda)\xrightarrow{p}0$.
    Applying Slutsky's theorem, replacing $b_+(\lambda)$ with $\hat{b}_+(\lambda)$ alters the test statistic by an $o_p(1)$ term, leaving the asymptotic upper-bound strictly preserved.
\end{proof}

\subsection{Proof of Proposition~\ref{prop:tight-min}}
\label{app:pf-prop-tight-min}
\begin{proof}[Proof of Proposition~\ref{prop:tight-min}]
    Fix $\lambda\in\Lambda$.
    To demonstrate both tightness and minimality, we explicitly construct a least-favorable null configuration satisfying $\theta_C=0$ and $P_H^a\in\mathcal{U}_a(\rho_a)$ that exactly attains the worst-case bias $w_1\Delta_1-w_0\Delta_0=b_+(\lambda)$.

    Fix any valid marginals $P_C^a \in \mathcal{P}_1(\mathcal{Z})$ ensuring $\mu_C^1 = \mu_C^0$ ($\theta_C=0$).
    If $\mathcal{Y}=\mathbb{R}$, we construct $P_H^1$ and $P_H^0$ using the translation maps $T_{1,+}(x,y)=(x,y+\rho_1)$ and $T_{0,-}(x,y)=(x,y-\rho_0)$ respectively.
    As verified in Proposition~\ref{prop:closed-form-delta} (i), this guarantees $P_H^a\in\mathcal{U}_a(\rho_a)$ and induces shifts $\Delta_1 = \rho_1 = \Delta_1^+(\rho_1)$ and $\Delta_0 = -\rho_0 = \Delta_0^-(\rho_0)$.

    If $\mathcal{Y}=\{0,1\}$, we apply the specific Markov kernels $K_{1,+}$ and $K_{0,-}$ detailed in the proof of Proposition~\ref{prop:closed-form-delta} (ii).
    This rigorously constructs $P_H^1, P_H^0$ within the Wasserstein balls that explicitly attain $\Delta_1 = \Delta_1^+(\rho_1)$ and $\Delta_0 = \Delta_0^-(\rho_0)$.

    In both scenarios, $\mathbb{E}[\hat{\theta}(\lambda)] = w_1\Delta_1^+(\rho_1) - w_0\Delta_0^-(\rho_0) = b_+(\lambda)$.
    Note also that $s^2(\lambda) \to 0$ as sample sizes diverge, ensuring $\hat{s}(\lambda) \to_p 0$.

    For part (i) (minimality), consider any constant $c < b_+(\lambda)$.
    The test statistic decomposes as:
    \begin{equation*}
        \frac{\hat{\theta}(\lambda)-c}{\hat{s}(\lambda)}
        =
        \frac{\hat{\theta}(\lambda)-b_+(\lambda)}{\hat{s}(\lambda)}
        +
        \frac{b_+(\lambda)-c}{\hat{s}(\lambda)}
        .
    \end{equation*}
    Under the constructed configuration, the first term converges in distribution to $N(0,1)$.
    The second term features a strictly positive numerator ($b_+(\lambda)-c > 0$) and a denominator vanishing in probability ($\hat{s}(\lambda) \to_p 0$), driving the ratio to $+\infty$ in probability.
    Consequently, the test statistic diverges to $+\infty$, leading to a rejection probability of 1.
    Thus, $\liminf \mathbb{P}(\varphi_{\lambda,c}=1) = 1 > \alpha$.

    For part (ii) (tightness), utilizing $c = b_+(\lambda)$ under the identically constructed configuration zeroes out the drift fraction.
    The statistic converges strictly to $N(0,1)$, meaning $\lim \mathbb{P}(\varphi_{\lambda}=1) = \alpha$.
    Combining this with the uniform upper bound established in Theorem~\ref{thm:robust-size} confirms that the supremum equals precisely $\alpha$.
\end{proof}

\subsection{Proof of Theorem~\ref{thm:robust-power}}
\label{app:pf-thm-robust-power}
\begin{proof}[Proof of Theorem~\ref{thm:robust-power}]
    Fix a true alternative $\theta_C=\theta_1>0$, and consider any admissible historical data configuration $P_H^a\in\mathcal{U}_a(\rho_a)$.
    Following the familiar asymptotic framework (Proposition~\ref{prop:clt} and Lemma~\ref{lem:consistency}):
    \begin{equation*}        
        \frac{\hat{\theta}(\lambda)-b_+(\lambda)}{\hat{s}(\lambda)}
        \longrightarrow_d
        N\Biggl(
            \frac{\theta_1+w_1\Delta_1-w_0\Delta_0-b_+(\lambda)}{s(\lambda)}
            ,
            1
        \Biggr)
        .
    \end{equation*}
    The probability of rejection thus limits to $1-\Phi\bigl(z_{1-\alpha} - u(\Delta_0,\Delta_1)\bigr)$, where $u(\Delta_0,\Delta_1)$ is the noncentrality parameter:
    \begin{equation*}        
        u(\Delta_0,\Delta_1)
        \coloneq
        \frac{\theta_1+w_1\bigl(\Delta_1-\Delta_1^+(\rho_1)\bigr)-w_0\bigl(\Delta_0-\Delta_0^-(\rho_0)\bigr)}{s(\lambda)}
        .
    \end{equation*}
    Because $\Phi$ is strictly monotone, the robust power is defined by the infimum of this parameter over the ambiguity sets.
    Since $w_a \ge 0$, $u(\Delta_0,\Delta_1)$ is minimized by selecting the smallest possible $\Delta_1$ and the largest possible $\Delta_0$ within their permissible bounds.
    Evaluating at $\Delta_1 = \Delta_1^-(\rho_1)$ and $\Delta_0 = \Delta_0^+(\rho_0)$ dictates the minimum robust noncentrality parameter:
    \begin{equation*}        
        \kappa(\lambda)
        =
        \frac{\theta_1
        - w_1\bigl(\Delta_1^+(\rho_1)-\Delta_1^-(\rho_1)\bigr)
        - w_0\bigl(\Delta_0^+(\rho_0)-\Delta_0^-(\rho_0)\bigr)}{s(\lambda)}
        .
    \end{equation*}
    This establishes the limiting robust power.
\end{proof}

\subsection{Proof of Corollary~\ref{cor:lambda-opt}}
\label{app:pf-cor-lambda-opt}
\begin{proof}[Proof of Corollary~\ref{cor:lambda-opt}]
    The objective function $\kappa(\lambda)$ is composed of several mappings.
    The effective weights $w_a(\lambda_a) = \lambda_a n_{H,a}/(n_{C,a}+\lambda_a n_{H,a})$ are continuous on $[0,\Lambda_a]$.
    The variance function $s^2(\lambda)$ is a finite polynomial of continuous functions, thus continuous on $\Lambda$, and strictly positive (Assumption~\ref{ass:asymptotic}), ensuring $s(\lambda)$ is continuous.
    Since the drift bounds $\Delta_a^\pm(\rho_a)$ are independent of $\lambda$, $\kappa(\lambda)$ is continuous over the parameter space $\Lambda$.
    Because $\Lambda$ is a compact interval, the Extreme Value Theorem guarantees that $\kappa(\lambda)$ attains a global maximum on $\Lambda$.
\end{proof}

\subsection{Proof of Proposition~\ref{prop:adaptive-robust-size}}
\label{app:pf-prop-adaptive-robust-size}
\begin{proof}[Proof of Proposition~\ref{prop:adaptive-robust-size}]    
    We fix an arbitrary null configuration where $\theta_C\le 0$ and $P_H^a\in\mathcal{U}_a(\rho_a)$.
    First, we establish the uniform convergence of the empirical objective function.
    We claim that
    \begin{equation}
    \label{eq:uniform-kappa-bar}
        \sup_{\lambda\in\Lambda}\bigl|
            \hat{\bar{\kappa}}(\lambda)-\bar{\kappa}(\lambda)
        \bigr|
        \longrightarrow_p
        0
        .
    \end{equation}
    This follows from three facts:
    (i) $w_a(\lambda_a)$ is continuous and bounded on $[0,\Lambda_a]$,
    (ii) sample variances $\hat{\sigma}_{j,a}^2$ consistently estimate $\sigma_{j,a}^2$, and
    (iii) the scaled variance $a_n s(\lambda)$ (where $a_n \coloneq \sqrt{n_C+n_H}$) is bounded away from $0$ uniformly over $\Lambda$.
    Thus, $\sup_{\lambda\in\Lambda}| \{a_n\hat{s}(\lambda)\}^{-1} - \{a_n s(\lambda)\}^{-1} | \to_p 0$.
    Coupled with the consistency of the plug-in bounds $\hat{D}_a(\rho_a)$ (for binary outcomes via CMT), \eqref{eq:uniform-kappa-bar} is verified.

    Given the well-separated maximizer condition (Assumption~\ref{ass:adaptive-identification}), standard $M$-estimation consistency arguments dictate that if $\|\hat{\lambda}-\lambda^\ast\|_2 \ge \varepsilon$, then $\bar{\kappa}(\lambda^\ast) - \bar{\kappa}(\hat{\lambda}) \ge \delta_\varepsilon$.
    This implies $\delta_\varepsilon \le 2\sup_{\lambda\in\Lambda} | \hat{\bar{\kappa}}(\lambda)-\bar{\kappa}(\lambda) |$.
    The uniform convergence from \eqref{eq:uniform-kappa-bar} guarantees that the probability of this event vanishes, proving $\hat{\lambda}\to_p \lambda^\ast$.
    
    Now, define the centered studentized process $U_n(\lambda) \coloneq (\hat{\theta}(\lambda)-\mathbb{E}[\hat{\theta}(\lambda)])/\hat{s}(\lambda)$.
    Because $\hat{\lambda}\to_p \lambda^\ast$, the continuous mapping ensures $w_a(\hat{\lambda}_a) \to_p w_a(\lambda_a^\ast)$.
    Algebraic decomposition shows that the difference between the centered means at $\hat{\lambda}$ and $\lambda^\ast$ is $o_p(a_n^{-1})$.
    Since $\hat{s}(\lambda^\ast) = \mathcal{O}(a_n^{-1})$, we secure $U_n(\hat{\lambda})-U_n(\lambda^\ast)\to_p 0$, which yields $U_n(\hat{\lambda}) \to_d N(0,1)$.

    The adaptive test statistic expands as:
    \begin{equation*}        
        T_n(\hat{\lambda})
        \coloneq
        \frac{\hat{\theta}(\hat{\lambda})-b_+(\hat{\lambda})}{\hat{s}(\hat{\lambda})}
        =
        U_n(\hat{\lambda})
        +
        \frac{\mathbb{E}[\hat{\theta}(\hat{\lambda})]-b_+(\hat{\lambda})}{\hat{s}(\hat{\lambda})}
        .
    \end{equation*}
    By the supremum definition of $b_+(\lambda)$, $\mathbb{E}[\hat{\theta}(\lambda)]-b_+(\lambda) \le \theta_C \le 0$ holds deterministically for all $\lambda$.
    Thus, the second term is non-positive almost surely, implying $\{T_n(\hat{\lambda})\ge z_{1-\alpha}\} \subseteq \{U_n(\hat{\lambda})\ge z_{1-\alpha}\}$.
    Taking limits bounds the rejection probability cleanly by $\alpha$.
\end{proof}

\subsection{Proof of Lemma~\ref{app:lem:w-lambda}}
\label{app:pf-lem-w-lambda}
\begin{proof}[Proof of Lemma~\ref{app:lem:w-lambda}]
    For $\lambda\ge 0$, the function $w(\lambda)=\lambda n_H/(n_C+\lambda n_H)$ is strictly increasing and continuous, with boundary values $w(0)=0$ and $\lim_{\lambda\to\infty}w(\lambda)=1$.
    Thus, it is a bijection onto $[0,1)$.
    Isolating $\lambda$ from $w = \lambda n_H/(n_C+\lambda n_H)$ yields $w n_C = \lambda n_H(1-w)$, and algebraic rearrangement immediately gives $\lambda = (n_C/n_H) [w/(1-w)]$.
\end{proof}

\subsection{Proof of Lemma~\ref{app:lem:ttp-ebw}}
\label{app:pf-lem-ttp-ebw}
\begin{proof}[Proof of Lemma~\ref{app:lem:ttp-ebw}]
    If $\hat{\eta}=0$, then $\hat{\lambda}=0$, which forces $\hat{\mu}(\hat{\lambda})=\hat{\mu}(0)=\bar{Y}_C=\hat{\mu}_{\mathrm{TTP}}$.
    If $\hat{\eta}=1$, then $\hat{\lambda}=\lambda_{\mathrm{pool}}$, verifying $\hat{\mu}(\hat{\lambda})=\hat{\mu}(\lambda_{\mathrm{pool}})=\hat{\mu}_{\mathrm{TTP}}$.
    The affine formulation strictly follows the definition of the effective weight $w(\cdot)$.
\end{proof}

\subsection{Proof of Lemma~\ref{app:lem:commensurate-ebw}}
\label{app:pf-lem-commensurate-ebw}
\begin{proof}[Proof of Lemma~\ref{app:lem:commensurate-ebw}]
    Fix $\tau>0$ and assume $n_C\ge 1$ and $n_H\ge 1$.
    Since the normal likelihood admits $\bar{Y}_j$ as a sufficient statistic for $\mu_j$, we work with the arm-level summaries
    \begin{equation*}        
        \bar{Y}_H \mid \mu_H
        \sim
        N\Biggl(
            \mu_H,\frac{\sigma^2}{n_H}
        \Biggr)
        ,
        \quad
        \bar{Y}_C \mid \mu_C
        \sim
        N\Biggl(
            \mu_C,\frac{\sigma^2}{n_C}
        \Biggr)
        ,
    \end{equation*}
    and the two summaries are independent across trials.
    
    Under the flat prior $\pi(\mu_H)\propto 1$, conjugacy yields the historical posterior
    \begin{equation*}
        \mu_H \mid D_H
        \sim
        N\Biggl(
            \bar{Y}_H,\frac{\sigma^2}{n_H}
        \Biggr)
        .
    \end{equation*}
    The commensurate link is $\mu_C\mid \mu_H,\tau \sim N(\mu_H,\tau^{-1})$.
    Integrating out $\mu_H$ under its posterior distribution given $D_H$ therefore gives the induced (predictive) prior for $\mu_C$
    conditional on $(D_H,\tau)$:
    \begin{equation*}
        \mu_C \mid D_H,\tau
        \sim
        N\Biggl(
            \bar{Y}_H
            ,\;
            \frac{\sigma^2}{n_H}+\tau^{-1}
        \Biggr)
        ,
    \end{equation*}
    using the fact that a normal location mixture with normal mixing distribution remains normal, with variances adding.
    
    Write the prior variance as $\sigma^2/m_{\mathrm{eff}}(\tau)$, where
    \begin{equation*}
        m_{\mathrm{eff}}(\tau)
        \coloneq
        \frac{\sigma^2}{\sigma^2/n_H+\tau^{-1}}
        =
        \frac{n_H}{1+n_H/(\sigma^2\tau)}
        .
    \end{equation*}
    Then we can equivalently express the induced prior as
    \begin{equation*}
        \mu_C \mid D_H,\tau
        \sim
        N\Biggl(
            \bar{Y}_H
            ,
            \frac{\sigma^2}{m_{\mathrm{eff}}(\tau)}
        \Biggr)
        ,
    \end{equation*}
    so that the corresponding prior precision is $m_{\mathrm{eff}}(\tau)/\sigma^2$.
    
    Updating this prior with the current likelihood (equivalently, with $\bar{Y}_C$) yields the posterior
    \begin{equation*}
        \mu_C \mid D_C,D_H,\tau
        \sim
        N\Biggl(
            \frac{\frac{n_C}{\sigma^2}\bar{Y}_C+\frac{m_{\mathrm{eff}}(\tau)}{\sigma^2}\bar{Y}_H}
            {\frac{n_C}{\sigma^2}+\frac{m_{\mathrm{eff}}(\tau)}{\sigma^2}}
            ,\;
            \frac{1}{\frac{n_C}{\sigma^2}+\frac{m_{\mathrm{eff}}(\tau)}{\sigma^2}}
        \Biggr)
        .
    \end{equation*}
    Therefore, the posterior mean is
    \begin{equation*}
        \mathbb{E}[\mu_C\mid D_C,D_H,\tau]
        =
        \frac{n_C\bar{Y}_C+m_{\mathrm{eff}}(\tau)\bar{Y}_H}{n_C+m_{\mathrm{eff}}(\tau)}
        =
        (1-w(\tau))\bar{Y}_C+w(\tau)\bar{Y}_H,
        \quad
        w(\tau)=\frac{m_{\mathrm{eff}}(\tau)}{n_C+m_{\mathrm{eff}}(\tau)}
        .
    \end{equation*}
    
    Finally, defining $\lambda(\tau)\coloneq m_{\mathrm{eff}}(\tau)/n_H$ gives
    \begin{equation*}
        w(\tau)
        =
        \frac{\lambda(\tau)n_H}{n_C+\lambda(\tau)n_H}
        ,
    \end{equation*}
    which is exactly the EBW form.
    Moreover,
    \begin{equation*}
        \lambda(\tau)
        =
        \frac{m_{\mathrm{eff}}(\tau)}{n_H}
        =
        \frac{\sigma^2\tau}{\sigma^2\tau+n_H}
        \in
        (0,1)
        ,
    \end{equation*}
    as claimed.
\end{proof}

\subsection{Proof of Lemma~\ref{app:lem:map-ebw}}
\label{app:pf-lem-map-ebw}
\begin{proof}[Proof of Lemma~\ref{app:lem:map-ebw}]
    The marginal posterior under a mixture prior is a mixture of the component-wise posteriors.
    Taking the expectation yields $\mathbb{E}[\mu\mid D_C,D_H] = \sum_{k=1}^K \omega_k^{\mathrm{post}} \mathbb{E}_k[\mu\mid D_C]$.
    Substituting the conjugate affine form $\mathbb{E}_k[\mu\mid D_C] = (1-w_k)\bar{Y}_C+w_k m_k$ and aggregating terms over $\bar{Y}_C$ gives:
    \begin{equation*}        
        \mathbb{E}[\mu\mid D_C,D_H]
        =
        \Bigl(
            1-\sum_{k=1}^K \omega_k^{\mathrm{post}}w_k
        \Bigr)\bar{Y}_C
        +
        \sum_{k=1}^K \omega_k^{\mathrm{post}}w_k m_k
        ,
    \end{equation*}
    which structurally matches $(1-w_{\mathrm{eff}})\bar{Y}_C + w_{\mathrm{eff}}m_{\mathrm{eff}}$ when the aggregate parameters are defined as stated.
\end{proof}

\subsection{Proof of Proposition~\ref{app:prop:bplus-multisource}}
\label{app:pf-prop-bplus-multisource}
\begin{proof}[Proof of Proposition~\ref{app:prop:bplus-multisource}]
    This expands the logic of Proposition~\ref{prop:bplus}.
    We define deterministic coefficients $\alpha_{k,a} \coloneq c_aw_{k,a}(\lambda)$ and functional components $g_{a}(Q) \coloneq \mathbb{E}_{Q}[Y]-\mu_C^a$.
    The multi-source worst-case bias \eqref{app:eq:bplus-multisource-def} demands evaluating:
    \begin{equation*}
        b_{+}(\lambda;c)
        =
        \sup_{(Q_{k,a})\in\prod \mathcal{U}_{k,a}} \sum_{k=1}^J\sum_{a\in\mathcal{A}} \alpha_{k,a}g_a(Q_{k,a})
        .
    \end{equation*}
    Because the joint ambiguity space is constructed as a Cartesian product of independent sets, the global supremum of the sum cleanly decomposes into the sum of individual suprema over each specific arm and source constraint.
    For each index $(k,a)$, optimizing $\alpha_{k,a}g_a(Q)$ isolates $\Delta_{k,a}^{+}(\rho_{k,a})$ when $\alpha_{k,a} \ge 0$, and isolates $\Delta_{k,a}^{-}(\rho_{k,a})$ when $\alpha_{k,a} < 0$, precisely determined by the sign of the user-defined contrast coefficient $c_a$.
\end{proof}

\subsection{Proof of Theorem~\ref{app:thm:multisource-size}}
\label{app:pf-thm-multisource-size}
\begin{proof}[Proof of Theorem~\ref{app:thm:multisource-size}]
    We outline the necessary extensions to Theorem~\ref{thm:robust-size} and Proposition~\ref{prop:clt} to address the multi-source topology.
    Let $\mathcal{I}$ index all valid observed cohorts across the current and $J$ historical trials.
    We embed the contrast coefficients and effective weights into aggregate parameters $\beta_{j,a} \coloneq c_aw_{j,a}(\lambda)$, allowing the estimator to be written as $\hat{\theta}(\lambda;c) = \sum_{\mathcal{I}} \beta_{j,a}\bar{Y}_{j,a}$.

    The centered statistic mirrors the triangular array $\xi_{j,a,i} \coloneq \frac{\beta_{j,a}}{n_{j,a}}(Y_{j,a,i}-\mu_j^a)$.
    Because observations remain mutually independent across cohorts (Assumption~\ref{app:ass:multisource-sampling}) and minimum sample sizes diverge (Assumption~\ref{app:ass:multisource-asymptotic}), the Lindeberg condition verification proceeds exactly as in Appendix~\ref{app:pf-prop-clt}.
    The Dominated Convergence Theorem neutralizes the variance tails, securing CLT normality.
    Slutsky's theorem alongside the consistent sample variances (via WLLN) confirms:
    \begin{equation*}
        \frac{\hat{\theta}(\lambda;c)-\mathbb{E}\bigl[
            \hat{\theta}(\lambda;c)
        \bigr]}{\hat{s}(\lambda;c)}
        \longrightarrow_d
        N(0,1)
        .
    \end{equation*}
    Under the null $H_0(c)\colon\theta_C(c)\le 0$, the exact true drift $\sum \sum c_a w_{k,a}(\lambda)\Delta_{k,a}$ is inherently bounded above by the supremum optimization defining $b_+(\lambda;c)$.
    Thus, $\mathbb{E}[\hat{\theta}(\lambda;c)] - b_+(\lambda;c) \le 0$.
    Subtracting this strictly non-positive drift ensures the resulting distribution is stochastically bounded by the standard normal upper tail, restricting asymptotic rejections tightly to $\alpha$.
    The inclusion of plug-in estimates for binary boundaries operates on the exact probabilistic equivalence mechanism defined in Theorem~\ref{thm:robust-size}.
\end{proof}

\subsection{Proof of Proposition~\ref{app:prop:bminus}}
\label{app:pf-prop-bminus}
\begin{proof}[Proof of Proposition~\ref{app:prop:bminus}]
    By defining the worst-case negative bias as an infimum over separable product sets, the optimization functionally decouples:
    \begin{equation*}        
        b_-(\lambda)
        =
        \inf_{Q_1\in\mathcal{U}_1(\rho_1)} \Bigl[
            w_1(\lambda_1)\bigl(
                \mathbb{E}_{Q_1}[Y]-\mu_C^1
            \bigr)
        \Bigr]
        +
        \inf_{Q_0\in\mathcal{U}_0(\rho_0)} \Bigl[
            -w_0(\lambda_0)\bigl(
                \mathbb{E}_{Q_0}[Y]-\mu_C^0
            \bigr)
        \Bigr]
        .
    \end{equation*}
    Because $w_1(\lambda_1) \ge 0$, minimizing the first component naturally retrieves $\Delta_1^-(\rho_1)$.
    For the second term, applying the fundamental equivalence $\inf (-cx) = -c \sup (x)$ translates the target to $-w_0(\lambda_0)\Delta_0^+(\rho_0)$.
    Inserting the distinct explicit bounds identified in Proposition~\ref{prop:closed-form-delta} generates the definitive analytical constraints shown.
\end{proof}

\subsection{Proof of Lemma~\ref{app:lem:consistency-bminus}}
\label{app:pf-lem-consistency-bminus}
\begin{proof}[Proof of Lemma~\ref{app:lem:consistency-bminus}]
    By the WLLN (under Assumptions~\ref{ass:sampling} and \ref{ass:asymptotic}), the empirical rates consistently reflect the true means: $\bar{Y}_{C,a}\to_p \mu_C^a$.
    For any specific scalar boundary $\rho\ge 0$, the mathematical projections $\mu \mapsto \min\{\rho,\mu\}$ and $\mu \mapsto \min\{\rho,1-\mu\}$ are demonstrably continuous globally over $[0,1]$.
    Therefore, the CMT dictates that the empirical bounds forming $\hat{b}_-(\lambda)$ converge uniformly in probability to the deterministic true bound $b_-(\lambda)$.
\end{proof}

\subsection{Proof of Theorem~\ref{app:thm:robust-size-two-sided}}
\label{app:pf-thm-robust-size-two-sided}
\begin{proof}[Proof of Theorem~\ref{app:thm:robust-size-two-sided}]
    Fix $\lambda\in\Lambda$ under an explicit true null scenario where $\theta_C=0$ and bounds $P_H^a\in\mathcal{U}_a(\rho_a)$ hold.
    Identifying true integrated bias as $B \coloneq w_1(\lambda_1)\Delta_1 - w_0(\lambda_0)\Delta_0$, standard convergence (Proposition~\ref{prop:clt} and Lemma~\ref{lem:consistency}) guarantees that the variable $Z_n \coloneq (\hat{\theta}(\lambda)-B)/\hat{s}(\lambda)$ converges precisely to $N(0,1)$.

    For the upper bound analysis, substituting $b_+(\lambda)$ constructs a statistic driven by $Z_n + \delta_n$ where the deterministic deviation is $\delta_n \coloneq (B-b_+(\lambda))/\hat{s}(\lambda)$.
    The mathematical supremum guarantees $B \le b_+(\lambda)$, rendering $\delta_n \le 0$.
    Consequently, limiting upper tail significance remains rigidly confined to $\alpha/2$.

    Symmetrically, evaluating the lower bound boundary constructs the equivalent statistic $Z_n + \eta_n$ utilizing the differential limit $\eta_n \coloneq (B-b_-(\lambda))/\hat{s}(\lambda)$.
    Since $b_-(\lambda)$ functions as the established minimum infimum limit, $B \ge b_-(\lambda)$ enforcing $\eta_n \ge 0$.
    Evaluating against the lower quantile threshold predictably confines statistical impact to the complementary $\alpha/2$.

    Aggregating probabilities over the disjoint limits firmly restrains total statistical variation mapping to $\alpha/2 + \alpha/2 = \alpha$.
    Incorporating binary plug-in estimators retains validity solely by diverging asymptotically by functionally negligible $o_p(1)$ margins.
\end{proof}
\section{Detailed Numerical Experiments}
\label{app:sec:detailed-numerical-experiments}

\subsection{Detailed Experimental Setup}
\label{app:subsec:detailed-experimental-setup}

\subsubsection{Data-Generating Mechanism and Heterogeneity Scenarios}
\label{app:subsubsec:dgp}

We simulated a current randomized trial ($j=C$) and a historical dataset ($j=H$) with baseline covariates $X\in\mathbb{R}^p$ ($p=2$).
Treatment assignment $A\in\{0,1\}$ followed $\mathbb{P}(A=1)=\pi$ in trials that included both arms.
Covariates were generated as $X_{C,i}\sim N(0,I_p)$ and $X_{H,i}\sim N(\gamma m, I_p)$, where the vector $m\in\{0,1\}^p$ dictates the presence of covariate shift and the scalar $\gamma\ge 0$ controls the magnitude of heterogeneity.
The linear predictor was formulated as
\begin{equation*}    
    \eta_{j,i}
    =
    \beta_0 + X_{j,i}^\top \beta
    + \theta A_{j,i}
    + A_{j,i} X_{j,i}^\top \eta
    + u_{j,0} + (u_{j,1}-u_{j,0})A_{j,i}
    ,
\end{equation*}
where $\theta$ is the main treatment effect, $\eta$ encodes treatment effect modification, and $(u_{j,0},u_{j,1})$ dictate arm-specific drifts.
Outcomes were generated conditionally: $Y_{j,i} = \eta_{j,i} + \varepsilon_{j,i}$ with $\varepsilon_{j,i}\sim N(0,\sigma^2)$ for continuous endpoints, and $Y_{j,i} \sim \mathrm{Bernoulli}(\mathrm{expit}(\eta_{j,i}))$ for binary endpoints.
We fixed $(\beta_0,\beta,\sigma,\pi)=(0,(0.5,0.5)^\top,1,0.5)$ for continuous outcomes and $(\beta_0,\beta,\pi)=(-1,(0.5,0.5)^\top,0.5)$ for binary outcomes.
Effect modification, when present, was set to $\eta=(0.3,0.3)^\top$; otherwise $\eta=0$.

We evaluated three noncommensurability cases across a grid of $\gamma$ ranging from $0$ to $2$ in increments of $0.1$:
\begin{itemize}
    \item Commensurate: $m=0$, $\eta=0$, $(u_{H,0},u_{H,1})=(0,0)$; historical data include both arms.
    
    \item Covariate shift + effect modification: $m=\mathbf{1}_p$, $\eta=(0.3,0.3)^\top$, $(u_{H,0},u_{H,1})=(0,0)$; historical data include both arms.
    
    \item Control drift (historical control-only): $m=0$, $\eta=0$, $(u_{H,0},u_{H,1})=(\gamma,0)$; historical data include only control subjects ($A_{H,i}\equiv 0$).
\end{itemize}

\subsubsection{Trial Parameters and Calibration}
\label{app:subsubsec:parameter-calibration}

The current trial size was fixed at $n_C=200$ with $1:1$ allocation.
The historical sample size was $n_H=500$, allocated either equally across arms when historical treatment data were available, or entirely to the control arm in the historical control-only case.
We evaluated the one-sided test $H_0\colon\theta_C\le 0$ at the nominal level $\alpha=0.025$.
The target alternative was set to $\theta_1=0.3$ on the mean-difference scale.
For continuous outcomes, we set $\theta=0$ under the null and $\theta=\theta_1$ under the alternative.
For binary outcomes, we calibrated $\theta$ via Monte Carlo root-finding to exactly match the marginal risk difference $\theta \in \{0, \theta_1\}$.
Empirical type I error and power were estimated using $20{,}000$ and $10{,}000$ independent Monte Carlo replications, respectively.

\subsubsection{Implementation Details and Baselines}
\label{app:subsubsec:implementation}

We evaluated BOND under two Wasserstein radius configurations:
(i) an oracle radius $\rho_a=|\mu_H^a-\mu_C^a|$ derived from the true data-generating distributions under the null, and
(ii) a data-driven proxy $\hat{\rho}_a=c\widehat{W}_1(\widehat{P}_{C}^a,\widehat{P}_{H}^a)$ with an inflation multiplier $c=1.5$, utilizing the empirical 1-Wasserstein distance between the observed arm-specific outcome distributions.
BOND optimized $(\lambda_0,\lambda_1) \in [0,1]^2$ over a uniform grid to maximize the robust noncentrality parameter (with $\lambda_1$ fixed to $0$ when historical treatment data were unavailable).

We compared BOND against baseline methods implemented, configured with standard weakly informative base priors (e.g., diffuse normals or $\mathrm{Beta}(1,1)$).
Key hyperparameters included:
fixed EBW $\lambda\in\{0.25,0.5,0.75\}$;
power prior $\lambda=0.5$;
commensurate precision $\tau=1$;
robust MAP with vague weight $\varepsilon=0.2$ and informative components $\lambda\in\{0.25,1\}$;
elastic prior scale $=1$;
UIP $M=100$;
LEAP with $50\%$ prior exchangeability and a nonexchangeable variance inflation factor of $9$;
MEM with $50\%$ inclusion probability;
and TTP with a screening threshold $\alpha_{\mathrm{pool}}=0.1$.

\subsection{Detailed Numerical Results}
\label{app:subsec:detailed-numerical-results}

This subsection provides the exhaustive simulation outputs, supplementing the abbreviated summaries presented in Section~\ref{subsec:num-results}.
Figures~\ref{fig:app-type1-power-continuous-S0}--\ref{fig:app-type1-power-binary-S3} display the full trajectories of empirical type I error and power across the heterogeneity index $\gamma$ for all evaluated methods, covering both continuous and binary outcomes under both oracle and data-driven radius specifications.
Tables~\ref{tab:app-worst-cont-oracle}--\ref{tab:app-worst-bin-data} consolidate these curves by reporting the maximum type I error and minimum power over the evaluated $\gamma$ grid.
Note that the standard errors for the Monte Carlo estimates are approximately $0.0011$ for type I error and $\le 0.005$ for power.

\subsubsection{Operating Characteristics Under Heterogeneity}
\label{app:subsubsec:characteristics-hetero}
The comprehensive results starkly illustrate the vulnerability of standard borrowing techniques.
Methods enforcing a minimum degree of borrowing (e.g., fixed-discount rules, standard power priors, and naive pooling) invariably exhibit severe size distortions as the heterogeneity index $\gamma$ increases, with type I errors frequently approaching $1.0$ under the Covariate shift + effect modification scenario.
Dynamic borrowing methods designed to handle conflict (e.g., robust MAP, commensurate priors) manage to rein in type I error to varying degrees.
However, some of these methods can become excessively conservative in the Control drift scenario;
for instance, robust MAP suffers a severe collapse in power for continuous outcomes (with worst-case power dropping to near zero), whereas commensurate priors manage to retain moderate power.
In contrast, BOND reliably controls the maximum type I error near the nominal $0.025$ level (staying bounded below $\approx 0.029$) across all configurations while systematically averting power collapse.
It achieves this by adaptively reverting to the baseline Current-only performance in the most severe conflict scenarios.

\subsubsection{Calibrated Borrowing Profiles}
\label{app:subsubsec:calibrated-borrow-profile}
Figures~\ref{fig:app-lambda-continuous-S0}--\ref{fig:app-lambda-binary-S3} map the borrowing parameters $(\lambda_a^\ast)$ and the resulting effective weights $(w_a(\lambda_a^\ast))$ selected by BOND as a function of $\gamma$.
Under the oracle radii, BOND aggressively borrows (assigning near full weight) when true noncommensurability is minimal ($\rho_a \approx 0$). As the true discrepancy increases, these weights precipitously drop to zero.
The data-driven Wasserstein radii ($\hat{\rho}_a$) naturally introduce finite-sample variability, generally resulting in more conservative (larger) estimated radii even under true commensurability.
Consequently, the data-driven BOND specification borrows less aggressively at $\gamma=0$ compared to the oracle version.
As detailed in the continuous outcome tables, this trades a moderate reduction in maximum power (e.g., from $0.896$ under the oracle to $0.773$ under the data-driven approach) for strictly data-adaptive robustness and valid type I error control without relying on unobservable true parameters.

\begin{table}[tb]
\centering
\caption{
Worst-case operating characteristics over the heterogeneity grid $\gamma\in\{0,0.1,\dots,2\}$ for continuous outcomes under the oracle radius specification.
For each method and case, we report $\max_{\gamma}\widehat{\mathrm{TypeI}}(\gamma)$ and $\min_{\gamma}\widehat{\mathrm{Power}}(\gamma)$.
}
\label{tab:app-worst-cont-oracle}
\small
\renewcommand{\arraystretch}{1.1}
\setlength{\tabcolsep}{3pt}
\begin{tabular}{p{3.6cm}rrrrrr}
\toprule
& \multicolumn{2}{c}{Commensurate} & \multicolumn{2}{c}{Covariate shift + effect modification} & \multicolumn{2}{c}{Control drift (historical control-only)} \\
\cmidrule(lr){2-3}\cmidrule(lr){4-5}\cmidrule(lr){6-7}
Method & $\max_\gamma\widehat{\mathrm{TypeI}}$ & $\min_\gamma\widehat{\mathrm{Power}}$ & $\max_\gamma\widehat{\mathrm{TypeI}}$ & $\min_\gamma\widehat{\mathrm{Power}}$ & $\max_\gamma\widehat{\mathrm{TypeI}}$ & $\min_\gamma\widehat{\mathrm{Power}}$ \\
\midrule
Current-only
& $0.027$ & $0.401$
& $0.028$ & $0.331$
& $0.028$ & $0.400$ \\

Naive pooling
& $0.027$ & $0.896$
& $1.000$ & $0.825$
& $0.027$ & $0.000$ \\

Fixed $\lambda=0.25$
& $0.028$ & $0.736$
& $0.945$ & $0.646$
& $0.025$ & $0.000$ \\

Fixed $\lambda=0.50$
& $0.028$ & $0.857$
& $1.000$ & $0.779$
& $0.026$ & $0.000$ \\

Fixed $\lambda=0.75$
& $0.028$ & $0.888$
& $1.000$ & $0.818$
& $0.027$ & $0.000$ \\

Power prior ($\lambda=0.50$)
& $0.013$ & $0.767$
& $0.982$ & $0.664$
& $0.021$ & $0.000$ \\

Commensurate prior ($\tau=1.00$)
& $0.027$ & $0.406$
& $0.040$ & $0.345$
& $0.026$ & $0.357$ \\

Robust MAP ($\epsilon=0.20$)
& $0.012$ & $0.737$
& $0.082$ & $0.144$
& $0.028$ & $0.028$ \\

Elastic prior(scale=1.00)
& $0.019$ & $0.784$
& $0.094$ & $0.560$
& $0.024$ & $0.003$ \\

UIP ($M=100$)
& $0.012$ & $0.712$
& $0.953$ & $0.608$
& $0.017$ & $0.000$ \\

LEAP
& $0.011$ & $0.643$
& $0.052$ & $0.274$
& $0.018$ & $0.066$ \\

MEM
& $0.019$ & $0.583$
& $0.035$ & $0.285$
& $0.028$ & $0.283$ \\

BHMOI
& $0.029$ & $0.432$
& $0.031$ & $0.334$
& $0.028$ & $0.398$ \\

Nonparametric Bayes
& $0.027$ & $0.895$
& $0.189$ & $0.197$
& $0.028$ & $0.032$ \\

TTP
& $0.050$ & $0.820$
& $0.145$ & $0.331$
& $0.035$ & $0.383$ \\

BOND
& $0.028$ & $0.896$
& $0.029$ & $0.331$
& $0.028$ & $0.400$ \\

\bottomrule
\end{tabular}
\end{table}

\begin{table}[tb]
\centering
\caption{Worst-case operating characteristics over the heterogeneity grid $\gamma\in\{0,0.1,\dots,2\}$ for continuous outcomes under the data-driven (Wasserstein-based) radius specification with inflation multiplier $c=1.5$.}
\label{tab:app-worst-cont-data}
\small
\renewcommand{\arraystretch}{1.1}
\setlength{\tabcolsep}{3pt}
\begin{tabular}{p{3.6cm}rrrrrr}
\toprule
& \multicolumn{2}{c}{Commensurate} & \multicolumn{2}{c}{Covariate shift + effect modification} & \multicolumn{2}{c}{Control drift (historical control-only)} \\
\cmidrule(lr){2-3}\cmidrule(lr){4-5}\cmidrule(lr){6-7}
Method & $\max_\gamma\widehat{\mathrm{TypeI}}$ & $\min_\gamma\widehat{\mathrm{Power}}$ & $\max_\gamma\widehat{\mathrm{TypeI}}$ & $\min_\gamma\widehat{\mathrm{Power}}$ & $\max_\gamma\widehat{\mathrm{TypeI}}$ & $\min_\gamma\widehat{\mathrm{Power}}$ \\
\midrule
Current-only
& $0.029$ & $0.402$
& $0.027$ & $0.326$
& $0.028$ & $0.400$ \\

Naive pooling
& $0.028$ & $0.892$
& $1.000$ & $0.826$
& $0.026$ & $0.000$ \\

Fixed $\lambda=0.25$
& $0.029$ & $0.734$
& $0.949$ & $0.647$
& $0.026$ & $0.000$ \\

Fixed $\lambda=0.50$
& $0.028$ & $0.854$
& $1.000$ & $0.777$
& $0.026$ & $0.000$ \\

Fixed $\lambda=0.75$
& $0.029$ & $0.886$
& $1.000$ & $0.820$
& $0.026$ & $0.000$ \\

Power prior ($\lambda=0.50$)
& $0.013$ & $0.765$
& $0.984$ & $0.669$
& $0.022$ & $0.000$ \\

Commensurate prior ($\tau=1.00$)
& $0.028$ & $0.406$
& $0.042$ & $0.352$
& $0.027$ & $0.348$ \\

Robust MAP ($\epsilon=0.20$)
& $0.012$ & $0.736$
& $0.081$ & $0.145$
& $0.028$ & $0.026$ \\

Elastic prior (scale=1.00)
& $0.021$ & $0.784$ 
& $0.093$ & $0.551$
& $0.023$ & $0.003$ \\

UIP ($M=100$)
& $0.013$ & $0.710$
& $0.954$ & $0.611$
& $0.019$ & $0.000$ \\

LEAP
& $0.011$ & $0.639$
& $0.051$ & $0.269$
& $0.019$ & $0.065$ \\

MEM
& $0.012$ & $0.583$
& $0.036$ & $0.294$
& $0.027$ & $0.283$ \\

BHMOI
& $0.030$ & $0.431$
& $0.029$ & $0.328$
& $0.028$ & $0.401$ \\

Nonparametric Bayes
& $0.028$ & $0.892$
& $0.190$ & $0.196$
& $0.027$ & $0.028$ \\

TTP
& $0.051$ & $0.816$
& $0.143$ & $0.326$
& $0.035$ & $0.378$ \\

BOND
& $0.025$ & $0.773$
& $0.027$ & $0.326$
& $0.028$ & $0.400$ \\
\bottomrule
\end{tabular}
\end{table}

\begin{table}[tb]
\centering
\caption{
Worst-case operating characteristics over the heterogeneity grid $\gamma\in\{0,0.1,\dots,2\}$ for binary outcomes under the oracle radius specification.
}
\label{tab:app-worst-bin-oracle}
\small
\renewcommand{\arraystretch}{1.1}
\setlength{\tabcolsep}{3pt}
\begin{tabular}{p{3.6cm}rrrrrr}
\toprule
& \multicolumn{2}{c}{Commensurate} & \multicolumn{2}{c}{Covariate shift + effect modification} & \multicolumn{2}{c}{Control drift (historical control-only)} \\
\cmidrule(lr){2-3}\cmidrule(lr){4-5}\cmidrule(lr){6-7}
Method & $\max_\gamma\widehat{\mathrm{TypeI}}$ & $\min_\gamma\widehat{\mathrm{Power}}$ & $\max_\gamma\widehat{\mathrm{TypeI}}$ & $\min_\gamma\widehat{\mathrm{Power}}$ & $\max_\gamma\widehat{\mathrm{TypeI}}$ & $\min_\gamma\widehat{\mathrm{Power}}$ \\
\midrule
Current-only
& $0.028$ & $0.990$
& $0.028$ & $0.990$
& $0.029$ & $0.991$ \\

Naive pooling
& $0.028$ & $1.000$
& $0.842$ & $1.000$
& $0.022$ & $0.002$ \\

Fixed $\lambda=0.25$
& $0.029$ & $1.000$
& $0.236$ & $1.000$
& $0.022$ & $0.233$ \\

Fixed $\lambda=0.50$
& $0.029$ & $1.000$
& $0.561$ & $1.000$
& $0.022$ & $0.028$ \\

Fixed $\lambda=0.75$
& $0.029$ & $1.000$
& $0.754$ & $1.000$
& $0.022$ & $0.006$ \\

Power prior ($\lambda=0.50$)
& $0.012$ & $1.000$ 
& $0.306$ & $1.000$
& $0.017$ & $0.022$ \\

Commensurate prior ($\tau=1.00$)
& $0.024$ & $0.994$
& $0.025$ & $0.992$
& $0.022$ & $0.986$ \\

Robust MAP ($\epsilon=0.20$)
& $0.010$ & $1.000$
& $0.030$ & $0.958$
& $0.022$ & $0.947$ \\

Elastic prior (scale=1.00)
& $0.027$ & $1.000$
& $0.055$ & $0.987$
& $0.021$ & $0.963$ \\

UIP ($M=100$)
& $0.012$ & $1.000$
& $0.205$ & $1.000$
& $0.016$ & $0.304$ \\

LEAP
& $0.028$ & $0.991$
& $0.036$ & $0.991$
& $0.024$ & $0.991$ \\

MEM
& $0.012$ & $0.998$
& $0.037$ & $0.970$
& $0.022$ & $0.963$ \\

BHMOI
& $0.030$ & $0.992$
& $0.033$ & $0.990$
& $0.026$ & $0.991$ \\

Nonparametric Bayes
& $0.011$ & $0.998$
& $0.036$ & $0.974$
& $0.022$ & $0.971$ \\

TTP
& $0.052$ & $0.998$
& $0.118$ & $0.985$
& $0.032$ & $0.987$ \\

BOND
& $0.025$ & $1.000$
& $0.028$ & $0.990$
& $0.028$ & $0.991$ \\
\bottomrule
\end{tabular}
\end{table}

\begin{table}[tb]
\centering
\caption{
Worst-case operating characteristics over the heterogeneity grid $\gamma\in\{0,0.1,\dots,2\}$ for binary outcomes under the data-driven (Wasserstein-based) radius specification with inflation multiplier $c=1.5$.
}
\label{tab:app-worst-bin-data}
\small
\renewcommand{\arraystretch}{1.1}
\setlength{\tabcolsep}{3pt}
\begin{tabular}{p{3.6cm}rrrrrr}
\toprule
& \multicolumn{2}{c}{Commensurate} & \multicolumn{2}{c}{Covariate shift + effect modification} & \multicolumn{2}{c}{Control drift (historical control-only)} \\
\cmidrule(lr){2-3}\cmidrule(lr){4-5}\cmidrule(lr){6-7}
Method & $\max_\gamma\widehat{\mathrm{TypeI}}$ & $\min_\gamma\widehat{\mathrm{Power}}$ & $\max_\gamma\widehat{\mathrm{TypeI}}$ & $\min_\gamma\widehat{\mathrm{Power}}$ & $\max_\gamma\widehat{\mathrm{TypeI}}$ & $\min_\gamma\widehat{\mathrm{Power}}$ \\
\midrule
Current-only
& $0.029$ & $0.991$
& $0.027$ & $0.991$
& $0.028$ & $0.990$ \\

Naive pooling
& $0.028$ & $1.000$
& $0.838$ & $1.000$
& $0.021$ & $0.002$ \\

Fixed $\lambda=0.25$
& $0.028$ & $1.000$
& $0.232$ & $1.000$
& $0.022$ & $0.231$ \\

Fixed $\lambda=0.50$
& $0.029$ & $1.000$
& $0.555$ & $1.000$
& $0.021$ & $0.028$ \\

Fixed $\lambda=0.75$
& $0.028$ & $1.000$
& $0.746$ & $1.000$
& $0.021$ & $0.007$ \\

Power prior ($\lambda=0.50$)
& $0.012$ & $1.000$
& $0.302$ & $1.000$
& $0.017$ & $0.021$ \\

Commensurate prior ($\tau=1.00$)
& $0.024$ & $0.994$
& $0.024$ & $0.993$
& $0.022$ & $0.986$ \\

Robust MAP ($\epsilon=0.20$)
& $0.009$ & $1.000$
& $0.029$ & $0.958$
& $0.023$ & $0.944$ \\

Elastic prior (scale=1.00)
& $0.027$ & $1.000$
& $0.056$ & $0.988$
& $0.021$ & $0.965$ \\

UIP ($M=100$)
& $0.011$ & $1.000$
& $0.201$ & $1.000$
& $0.016$ & $0.307$ \\

LEAP
& $0.028$ & $0.991$
& $0.036$ & $0.991$
& $0.023$ & $0.990$ \\

MEM
& $0.012$ & $0.998$
& $0.036$ & $0.971$
& $0.023$ & $0.964$ \\

BHMOI 
& $0.031$ & $0.992$
& $0.032$ & $0.992$
& $0.025$ & $0.990$ \\

Nonparametric Bayes
& $0.011$ & $0.998$
& $0.034$ & $0.974$
& $0.023$ & $0.972$ \\

TTP
& $0.052$ & $0.998$
& $0.120$ & $0.988$
& $0.032$ & $0.987$ \\

BOND 
& $0.025$ & $1.000$
& $0.027$ & $0.992$
& $0.028$ & $0.990$ \\
\bottomrule
\end{tabular}
\end{table}

\begin{figure}[tb]
\centering
    \includegraphics[width=\columnwidth]{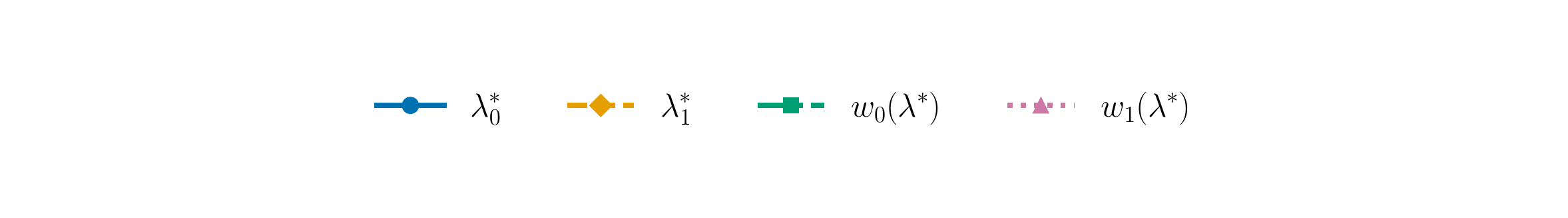}
    \begin{subfigure}[t]{0.49\linewidth}
    \centering
    \includegraphics[width=\linewidth]{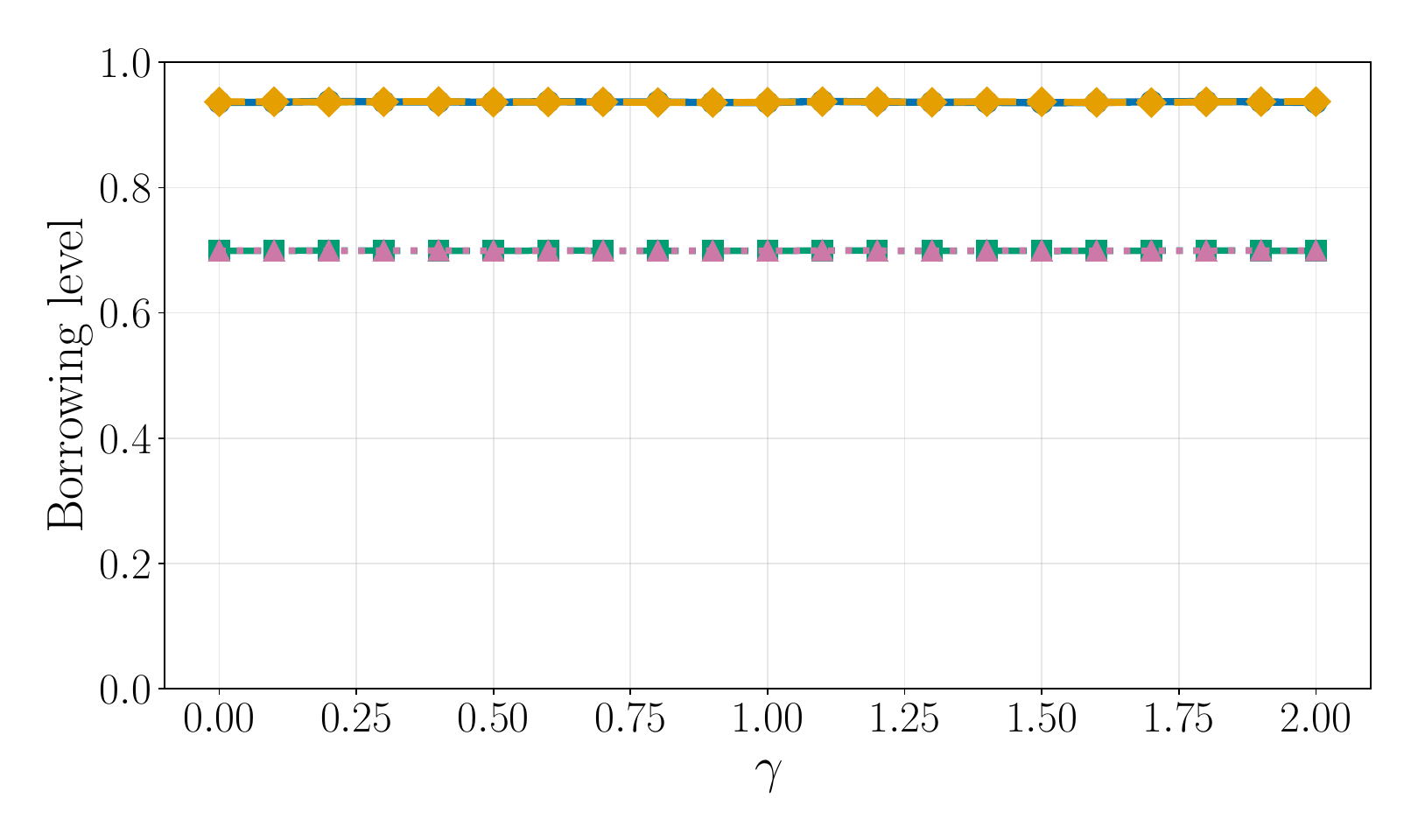}
    \caption{Oracle radii.}
    \end{subfigure}\hfill
    \begin{subfigure}[t]{0.49\linewidth}
    \centering
    \includegraphics[width=\linewidth]{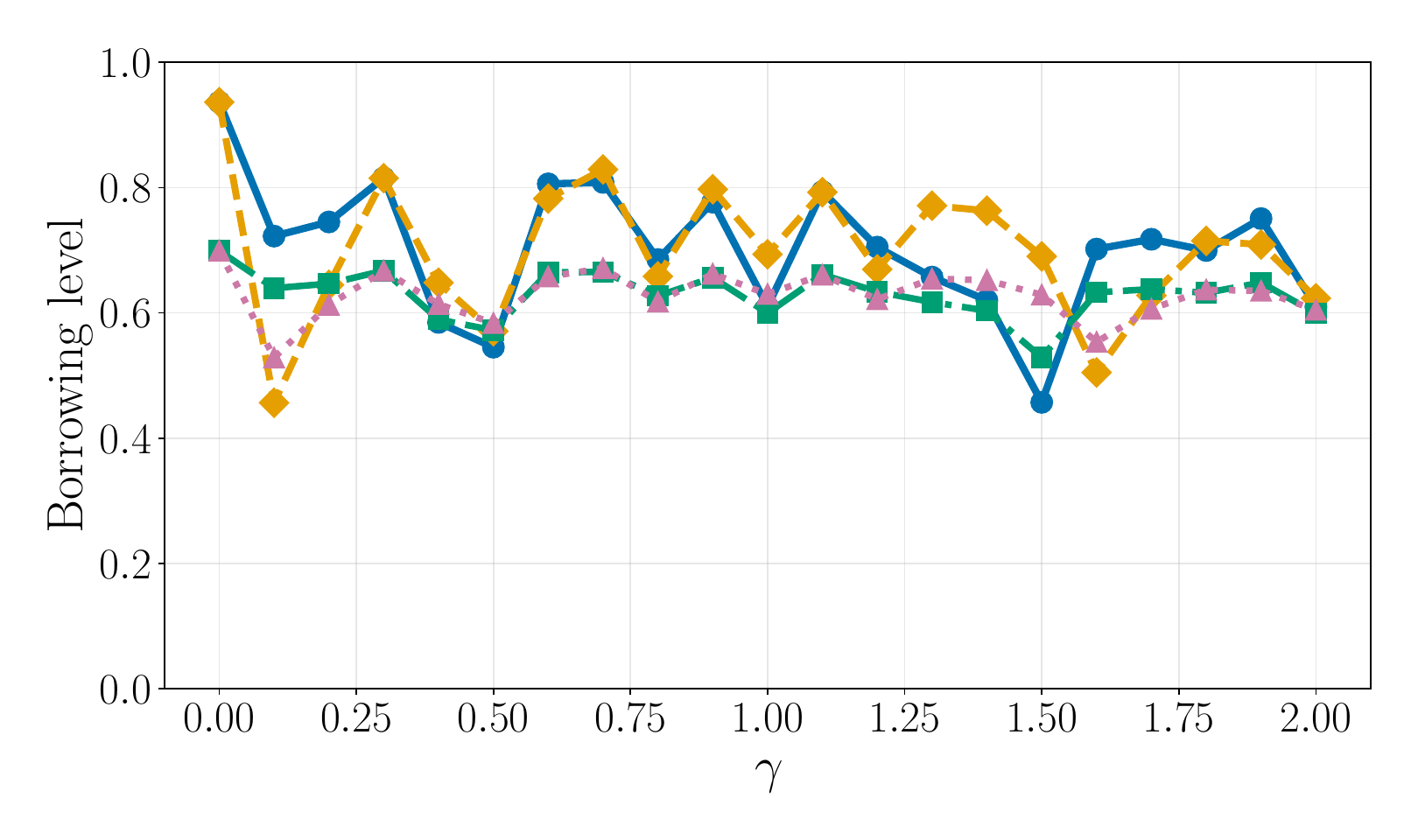}
    \caption{Wasserstein-based radii ($c=1.5$).}
    \end{subfigure}
\caption{
BOND calibrated borrowing levels versus $\gamma$ for continuous outcomes under Commensurate with $n_C=200$ and $n_H=500$.
Each panel reports the optimizer $\lambda_a^\ast$ and the induced effective weight $w_a(\lambda_a^\ast)$ for arm $a\in\{0,1\}$.
}
\label{fig:app-lambda-continuous-S0}
\end{figure}

\begin{figure}[tb]
\centering
    \includegraphics[width=\columnwidth]{figs/numerical/lambda_continuous_S2_data_c1p5_nH500_legend.pdf}
    \begin{subfigure}[t]{0.49\linewidth}
    \centering
    \includegraphics[width=\linewidth]{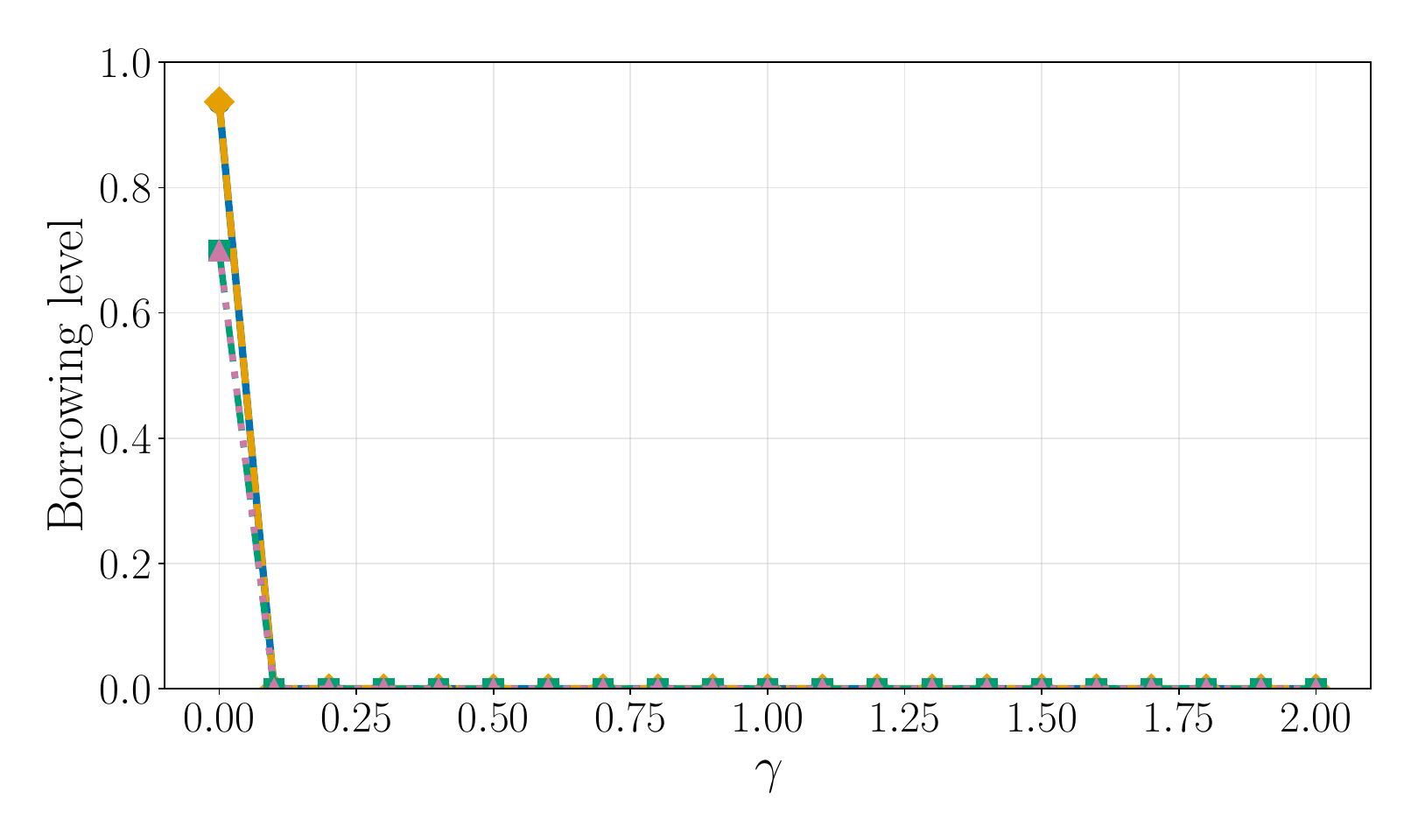}
    \caption{Oracle radii.}
    \end{subfigure}\hfill
    \begin{subfigure}[t]{0.49\linewidth}
    \centering
    \includegraphics[width=\linewidth]{figs/numerical/lambda_continuous_S2_data_c1p5_nH500.pdf}
    \caption{Wasserstein-based radii ($c=1.5$).}
    \end{subfigure}
\caption{
BOND calibrated borrowing levels versus $\gamma$ for continuous outcomes under Covariate shift + effect modification with $n_C=200$ and $n_H=500$.
Each panel reports the optimizer $\lambda_a^\ast$ and the induced effective weight $w_a(\lambda_a^\ast)$ for arm $a\in\{0,1\}$.
}
\label{fig:app-lambda-continuous-S2}
\end{figure}

\begin{figure}[tb]
\centering
    \includegraphics[width=\columnwidth]{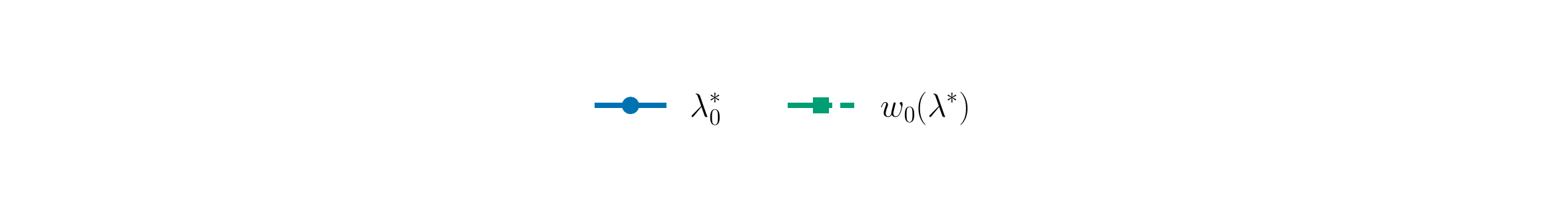}
    \begin{subfigure}[t]{0.49\linewidth}
    \centering
    \includegraphics[width=\linewidth]{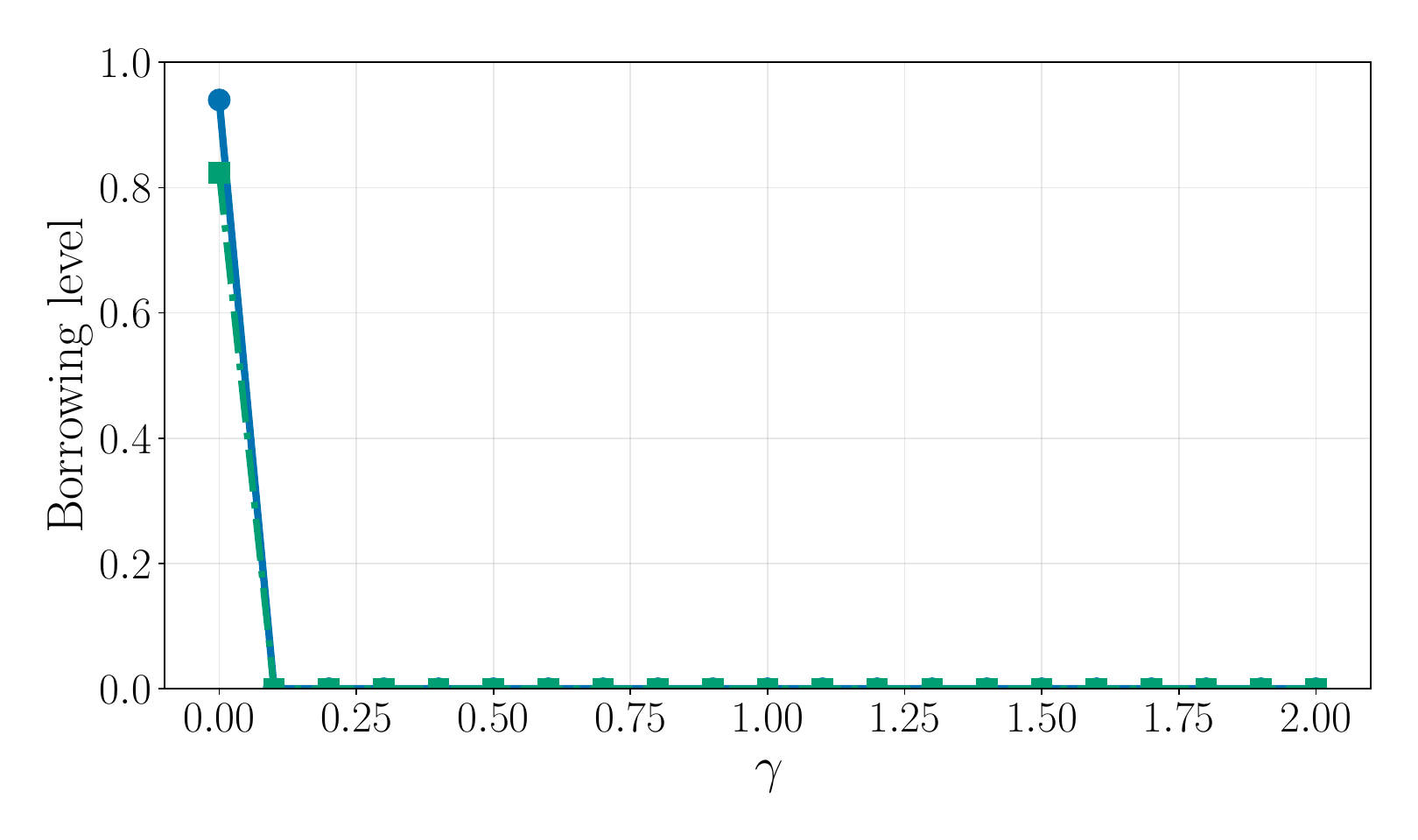}
    \caption{Oracle radii.}
    \end{subfigure}\hfill
    \begin{subfigure}[t]{0.49\linewidth}
    \centering
    \includegraphics[width=\linewidth]{figs/numerical/lambda_continuous_S3_data_c1p5_nH500.pdf}
    \caption{Wasserstein-based radii ($c=1.5$).}
    \end{subfigure}
\caption{
BOND calibrated borrowing levels versus $\gamma$ for continuous outcomes under Control drift (historical control-only) with $n_C=200$ and $n_H=500$.
Each panel reports the optimizer $\lambda_a^\ast$ and the induced effective weight $w_a(\lambda_a^\ast)$ for arm $a\in\{0,1\}$.
In the historical control-only case the historical treatment arm is unavailable, so $\lambda_1^\ast\equiv 0$ and only $\lambda_0^\ast$ is optimized.
}
\label{fig:app-lambda-continuous-S3}
\end{figure}

\begin{figure}[tb]
\centering
    \includegraphics[width=\columnwidth]{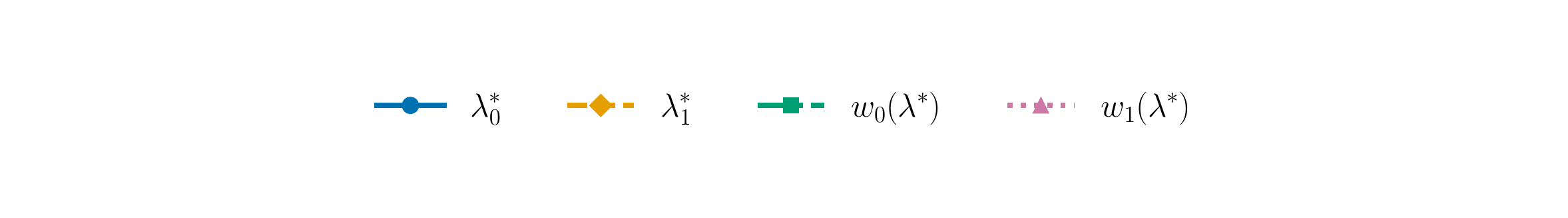}
    \begin{subfigure}[t]{0.49\linewidth}
    \centering
    \includegraphics[width=\linewidth]{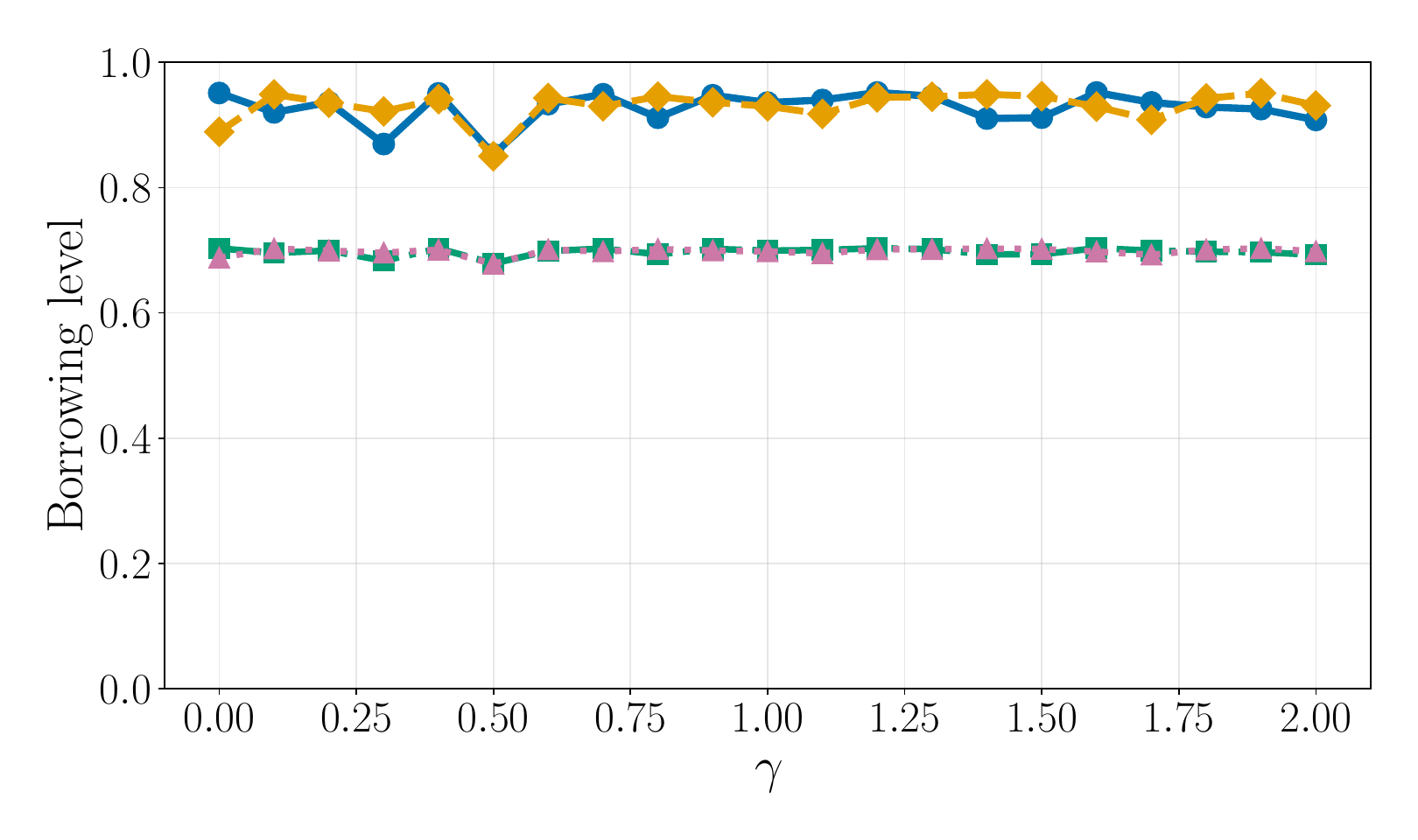}
    \caption{Oracle radii.}
    \end{subfigure}\hfill
    \begin{subfigure}[t]{0.49\linewidth}
    \centering
    \includegraphics[width=\linewidth]{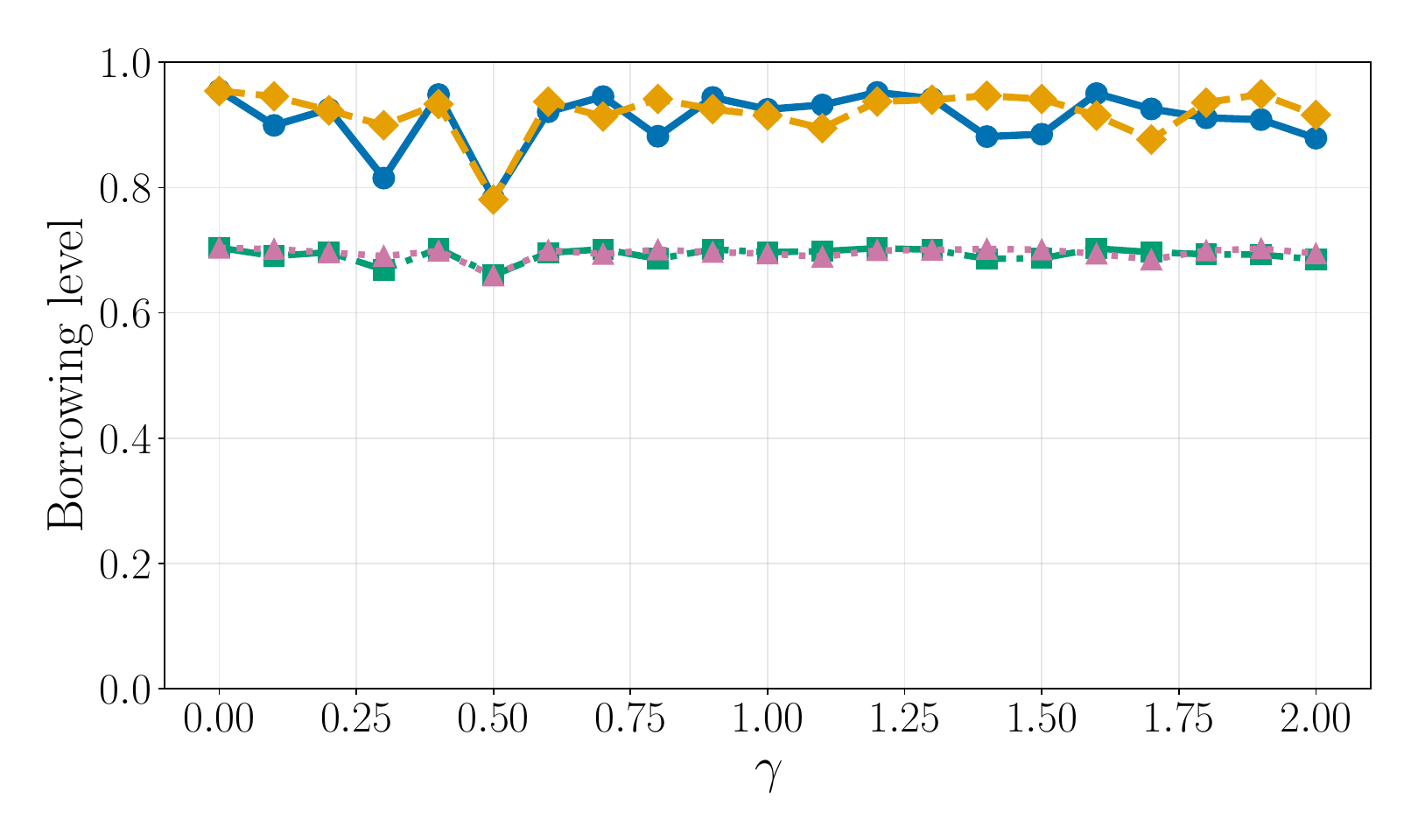}
    \caption{Wasserstein-based radii ($c=1.5$).}
    \end{subfigure}
\caption{
BOND calibrated borrowing levels versus $\gamma$ for binary outcomes under Commensurate with $n_C=200$ and $n_H=500$.
Each panel reports the optimizer $\lambda_a^\ast$ and the induced effective weight $w_a(\lambda_a^\ast)$ for arm $a\in\{0,1\}$.
}
\label{fig:app-lambda-binary-S0}
\end{figure}

\begin{figure}[tb]
\centering
    \includegraphics[width=\columnwidth]{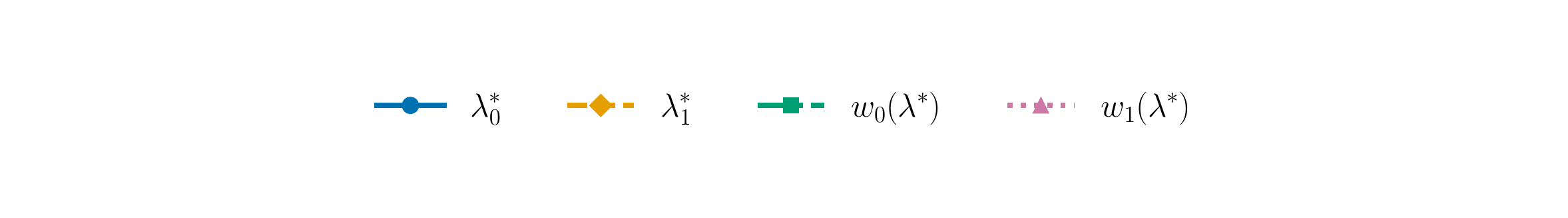}
    \begin{subfigure}[t]{0.49\linewidth}
    \centering
    \includegraphics[width=\linewidth]{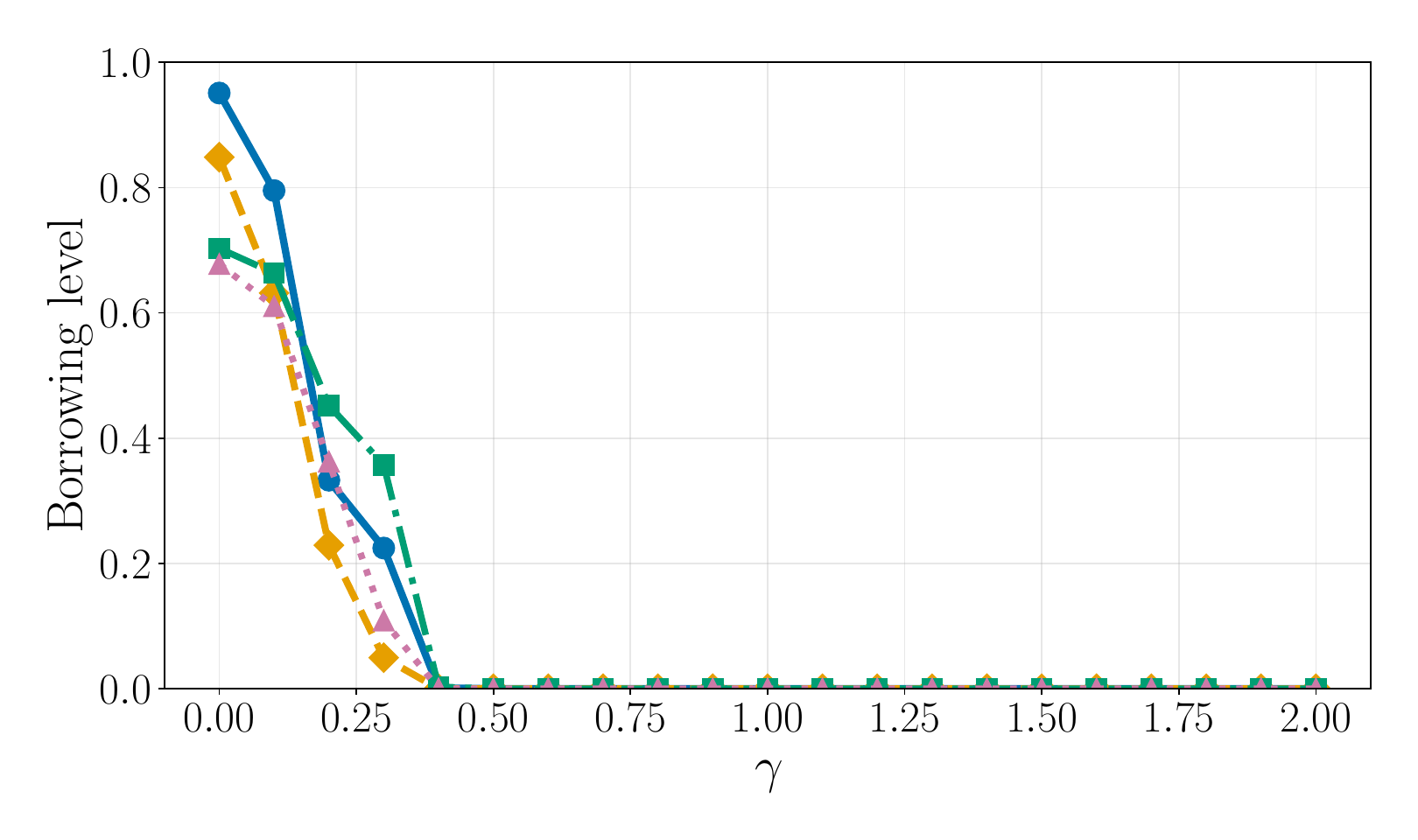}
    \caption{Oracle radii.}
    \end{subfigure}\hfill
    \begin{subfigure}[t]{0.49\linewidth}
    \centering
    \includegraphics[width=\linewidth]{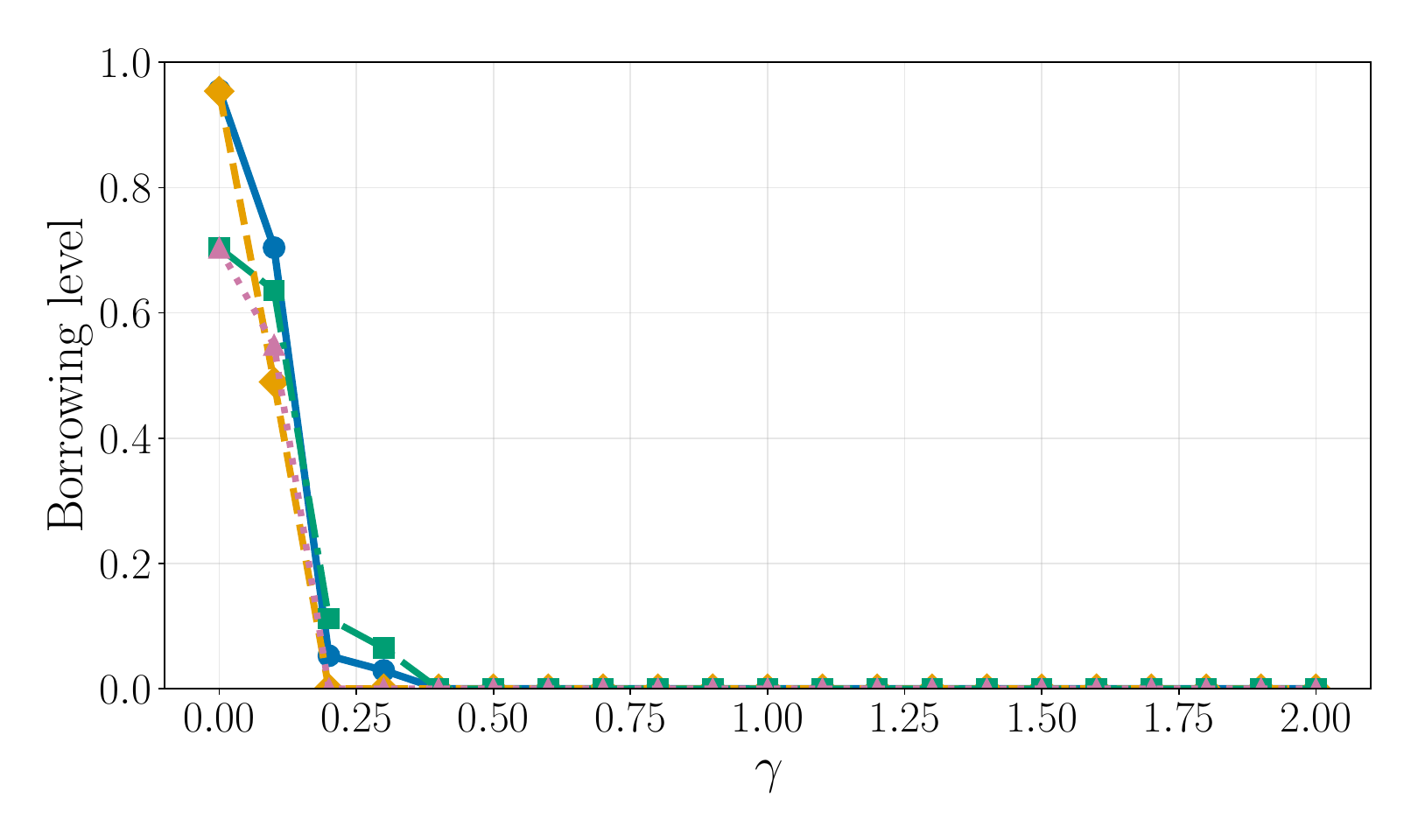}
    \caption{Wasserstein-based radii ($c=1.5$).}
    \end{subfigure}
\caption{
BOND calibrated borrowing levels versus $\gamma$ for binary outcomes under Covariate shift + effect modification with $n_C=200$ and $n_H=500$.
Each panel reports the optimizer $\lambda_a^\ast$ and the induced effective weight $w_a(\lambda_a^\ast)$ for arm $a\in\{0,1\}$.
}
\label{fig:app-lambda-binary-S2}
\end{figure}

\begin{figure}[tb]
\centering
    \includegraphics[width=\columnwidth]{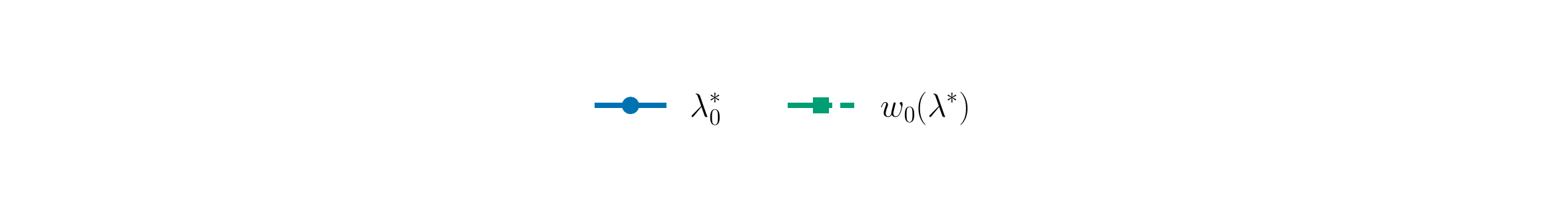}
    \begin{subfigure}[t]{0.49\linewidth}
    \centering
    \includegraphics[width=\linewidth]{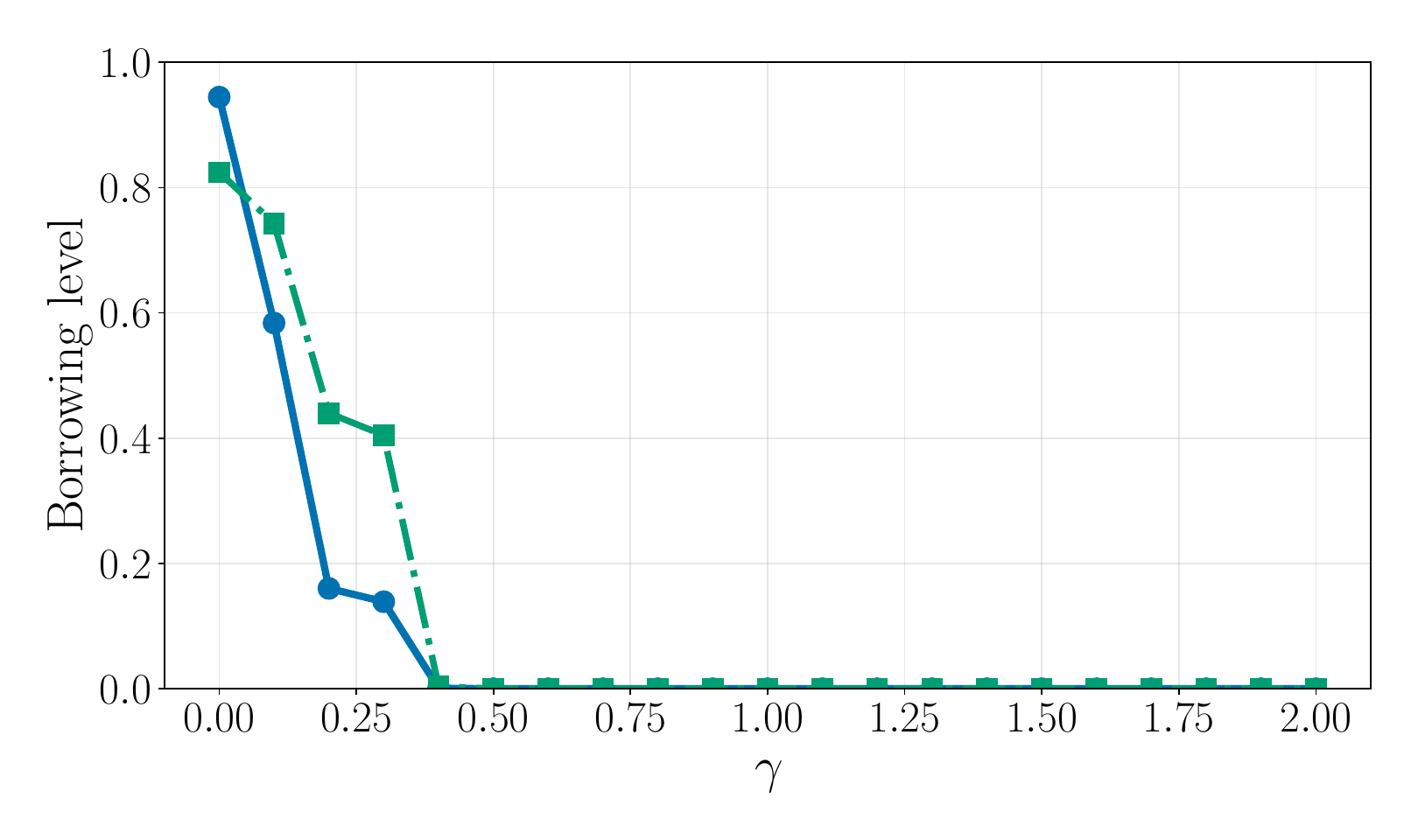}
    \caption{Oracle radii.}
    \end{subfigure}\hfill
    \begin{subfigure}[t]{0.49\linewidth}
    \centering
    \includegraphics[width=\linewidth]{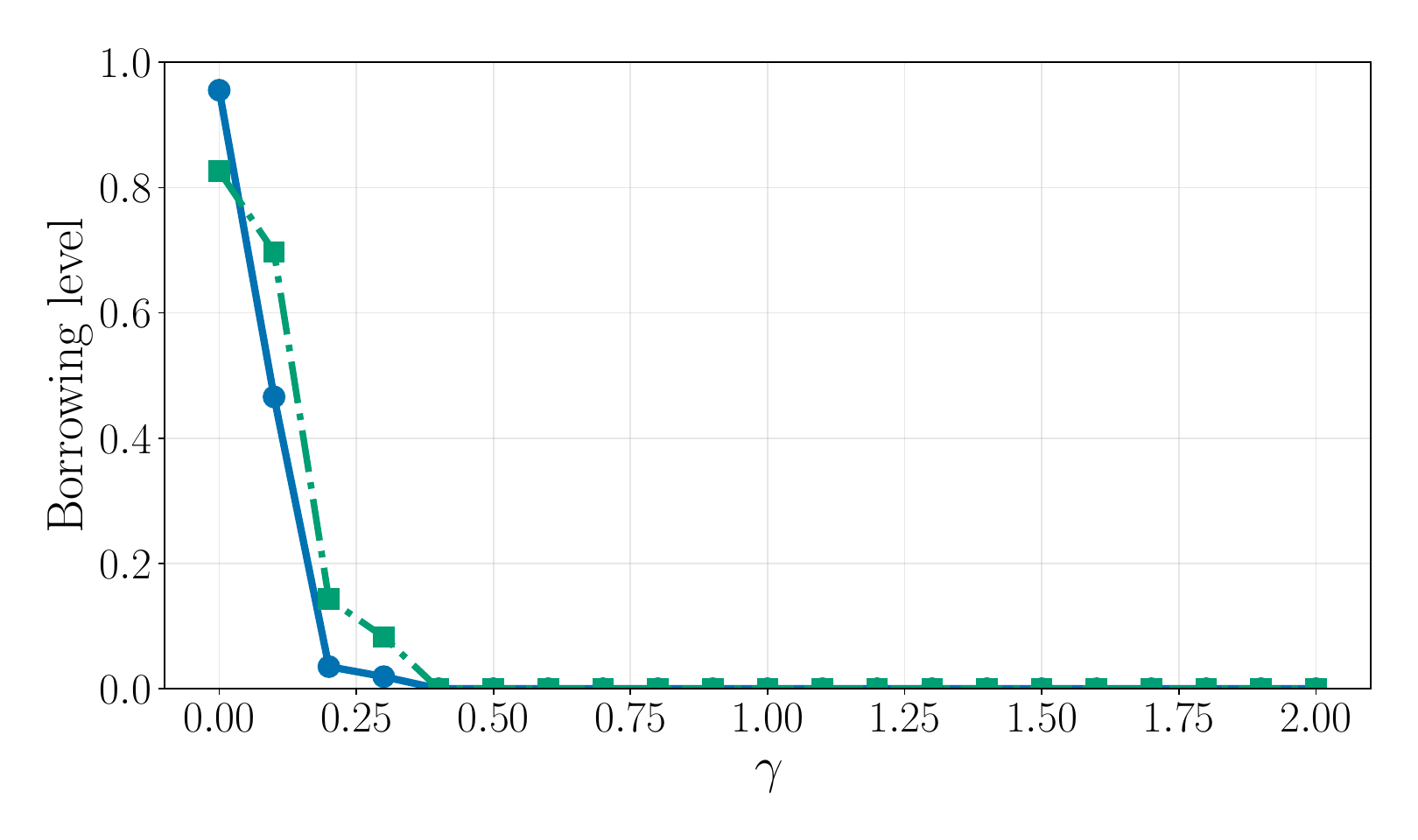}
    \caption{Wasserstein-based radii ($c=1.5$).}
    \end{subfigure}
\caption{
BOND calibrated borrowing levels versus $\gamma$ for binary outcomes under Control drift (historical control-only) with $n_C=200$ and $n_H=500$.
Each panel reports the optimizer $\lambda_a^\ast$ and the induced effective weight $w_a(\lambda_a^\ast)$ for arm $a\in\{0,1\}$.
In the historical control-only case the historical treatment arm is unavailable, so $\lambda_1^\ast\equiv 0$ and only $\lambda_0^\ast$ is optimized.
}
\label{fig:app-lambda-binary-S3}
\end{figure}

\begin{figure}[tb]
\centering
    \includegraphics[width=\columnwidth]{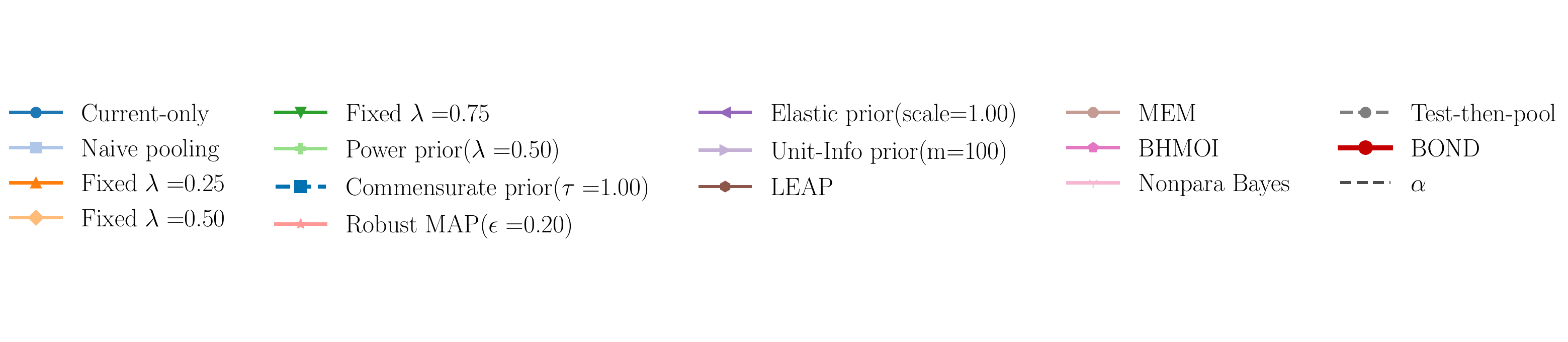}
    \begin{subfigure}[t]{0.49\linewidth}
    \centering
    \includegraphics[width=\linewidth]{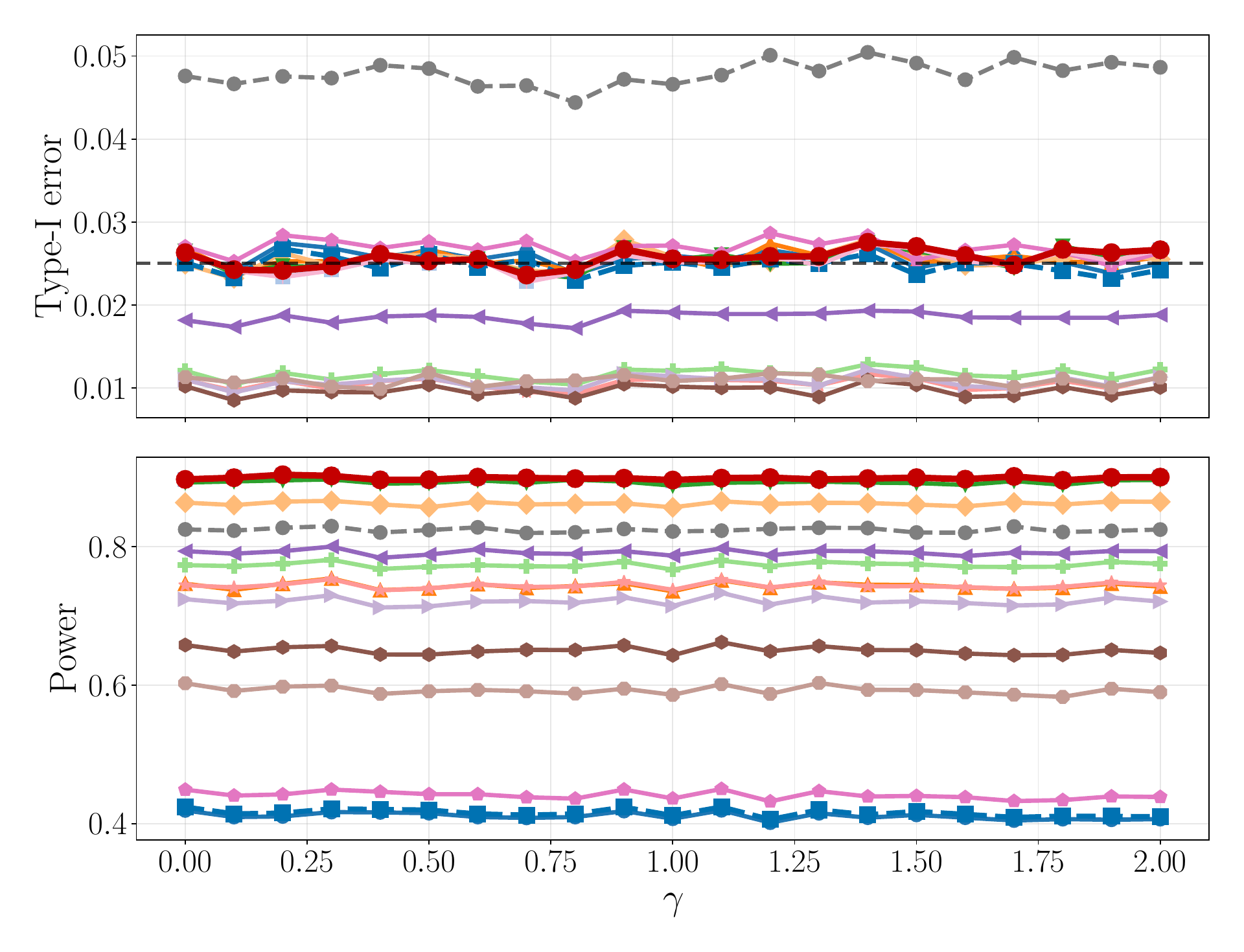}
    \caption{Oracle radii.}
    \end{subfigure}\hfill
    \begin{subfigure}[t]{0.49\linewidth}
    \centering
    \includegraphics[width=\linewidth]{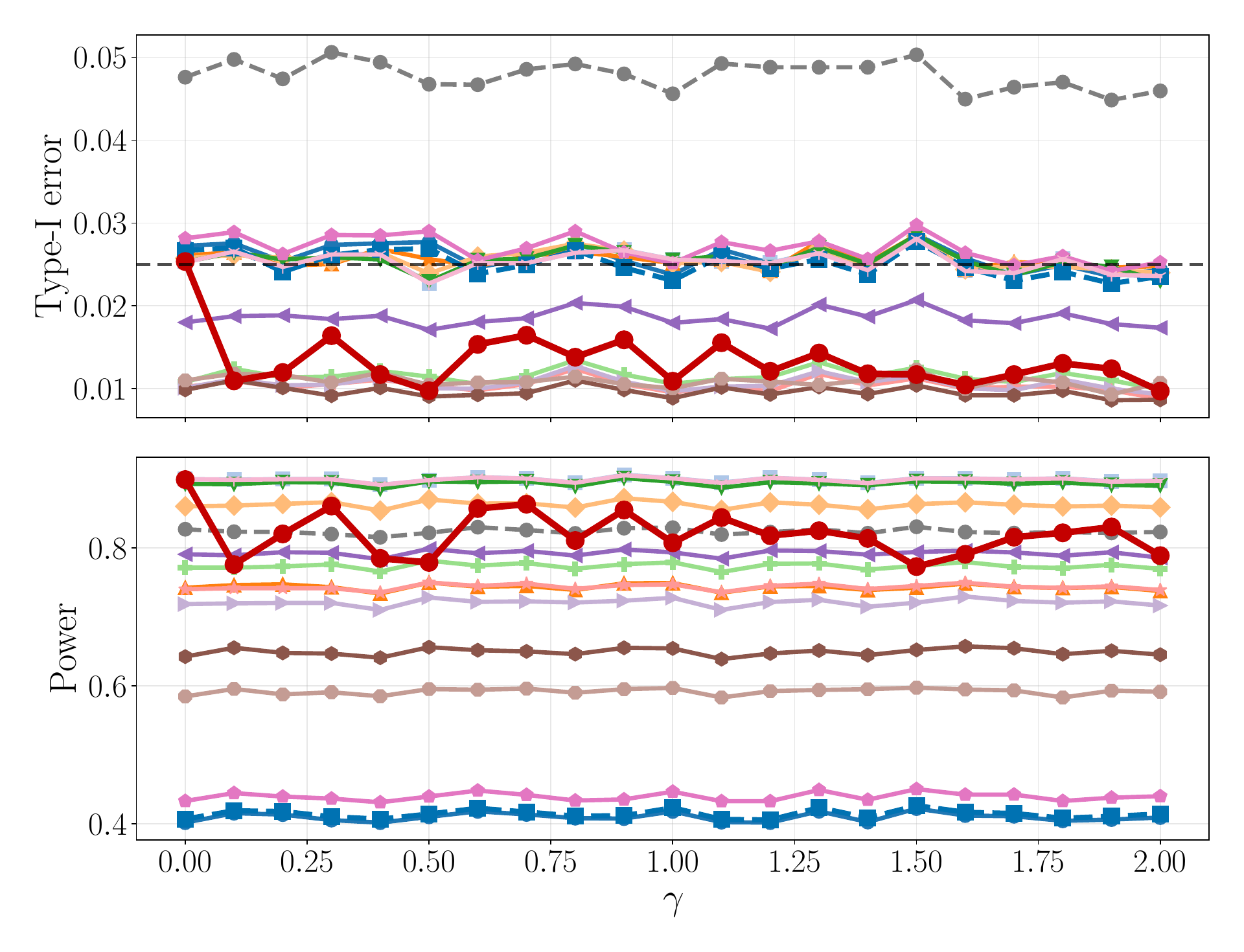}
    \caption{Wasserstein-based radii ($c=1.5$).}
    \end{subfigure}
\caption{
Type I error (top) and power (bottom) versus $\gamma$ for continuous outcomes under Commensurate with $n_C=200$ and $n_H=500$.
The horizontal reference line is at $\alpha=0.025$.
}
\label{fig:app-type1-power-continuous-S0}
\end{figure}

\begin{figure}[tb]
\centering
    \includegraphics[width=\columnwidth]{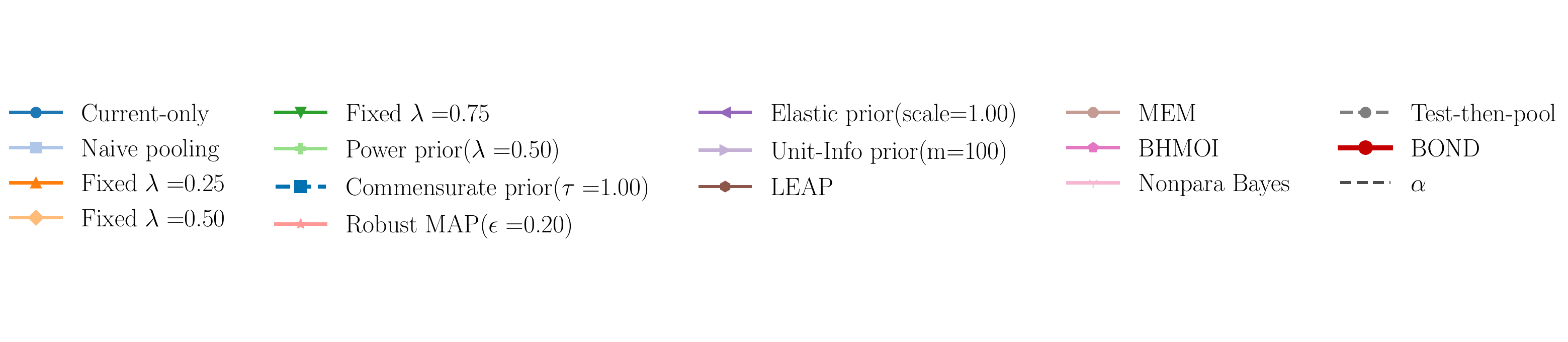}
    \begin{subfigure}[t]{0.49\linewidth}
    \centering
    \includegraphics[width=\linewidth]{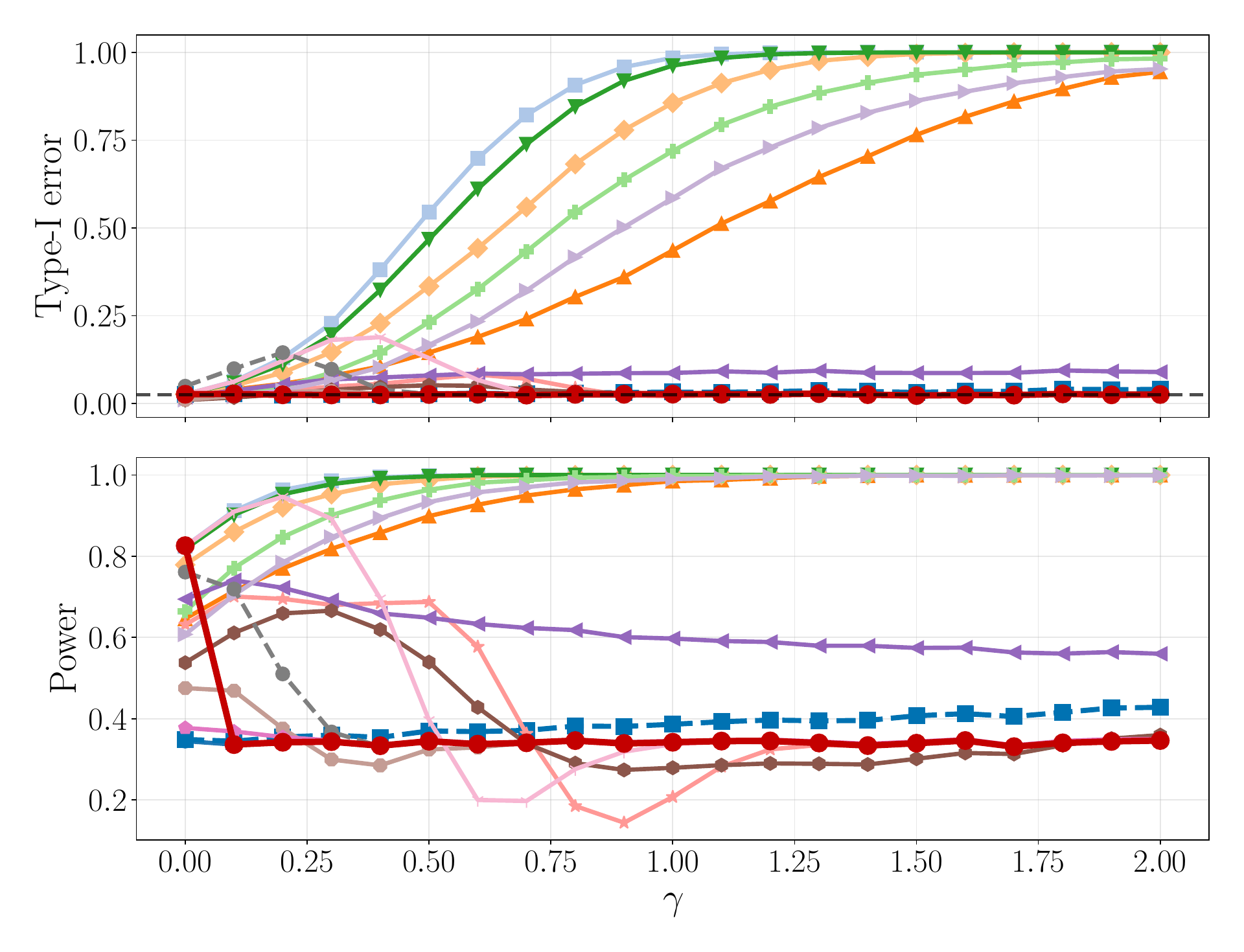}
    \caption{Oracle radii.}
    \end{subfigure}\hfill
    \begin{subfigure}[t]{0.49\linewidth}
    \centering
    \includegraphics[width=\linewidth]{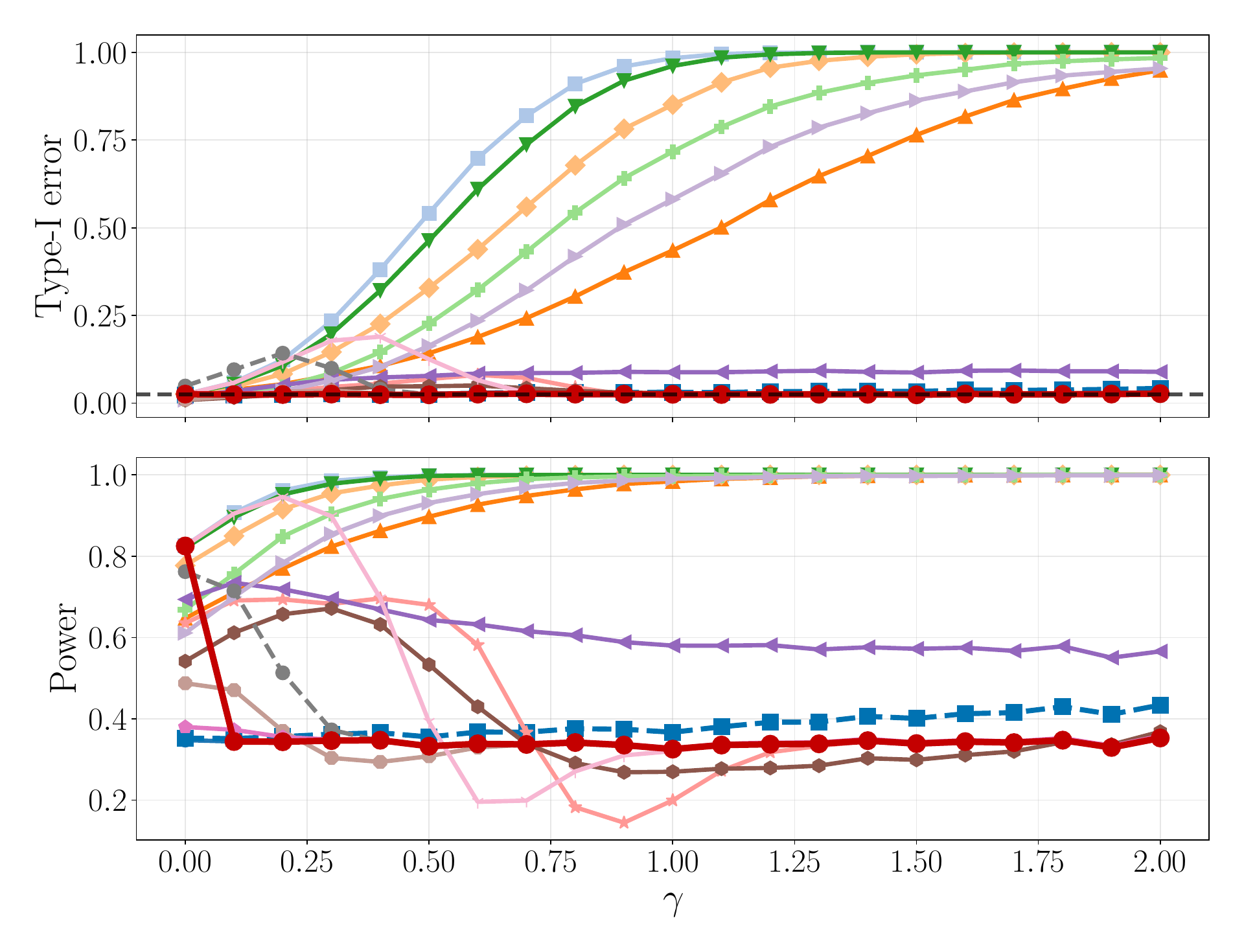}
    \caption{Wasserstein-based radii ($c=1.5$).}
    \end{subfigure}
\caption{
Type I error (top) and power (bottom) versus $\gamma$ for continuous outcomes under Covariate shift + effect modification with $n_C=200$ and $n_H=500$.
The horizontal reference line is at $\alpha=0.025$.
}
\label{fig:app-type1-power-continuous-S2}
\end{figure}

\begin{figure}[tb]
\centering
    \includegraphics[width=\columnwidth]{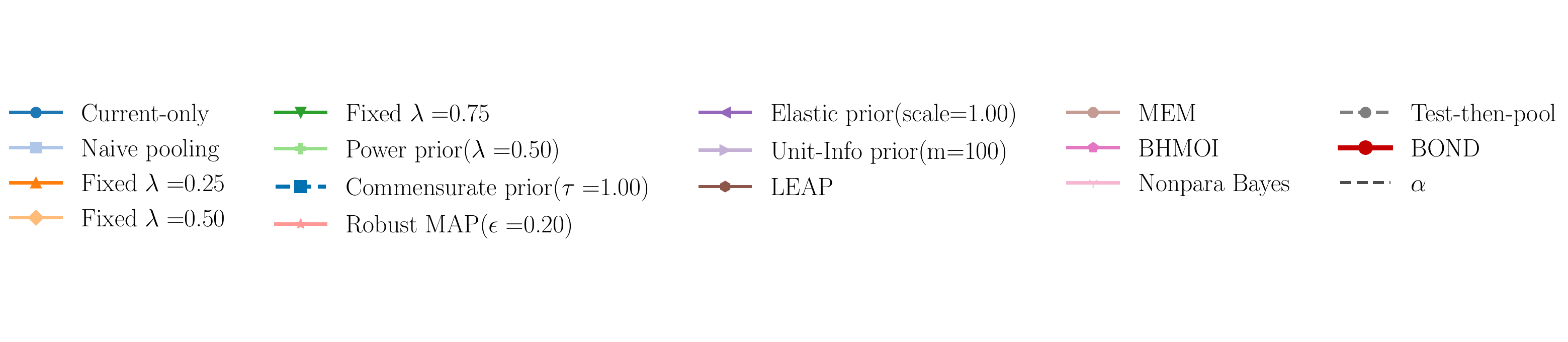}
    \begin{subfigure}[t]{0.49\linewidth}
    \centering
    \includegraphics[width=\linewidth]{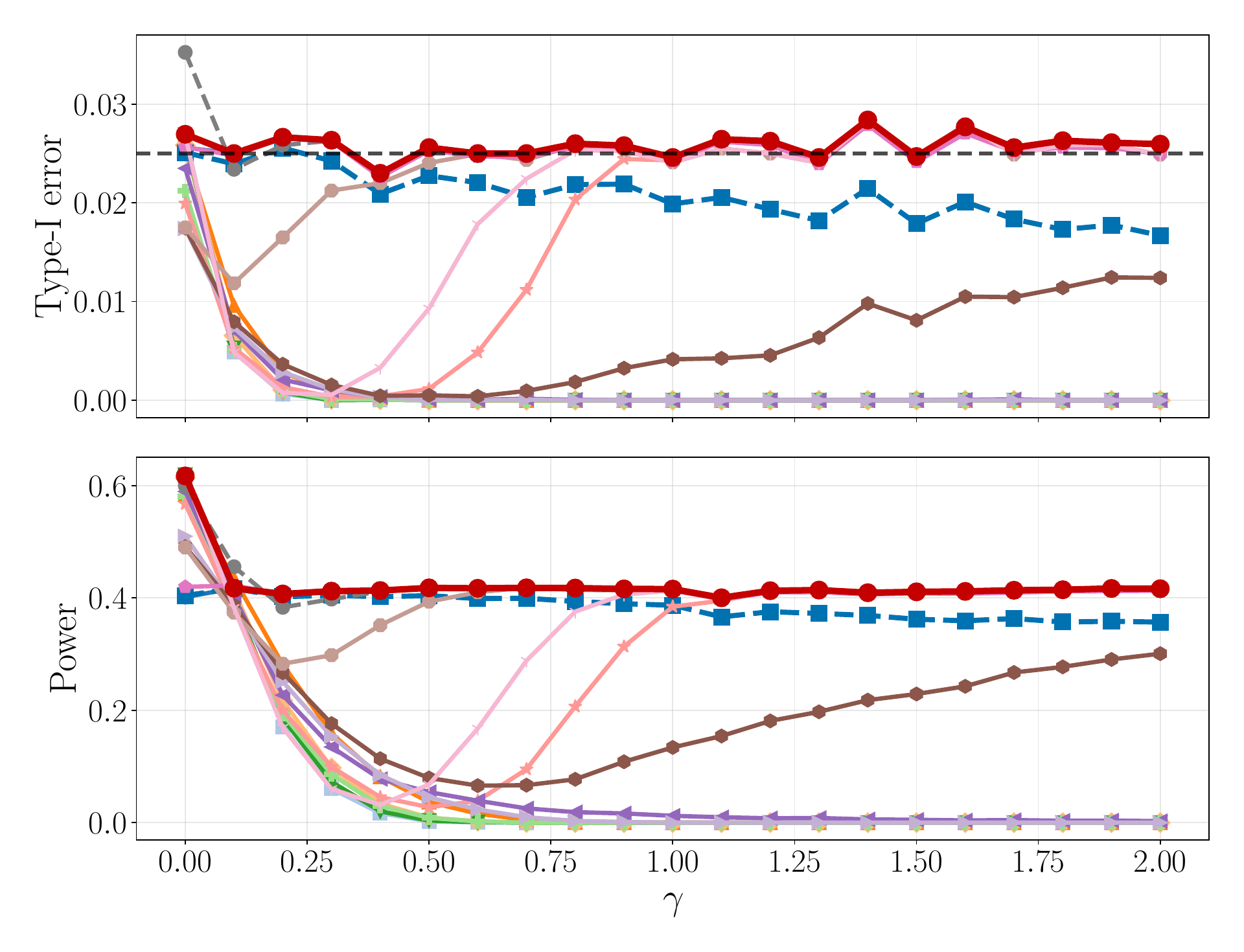}
    \caption{Oracle radii.}
    \end{subfigure}\hfill
    \begin{subfigure}[t]{0.49\linewidth}
    \centering
    \includegraphics[width=\linewidth]{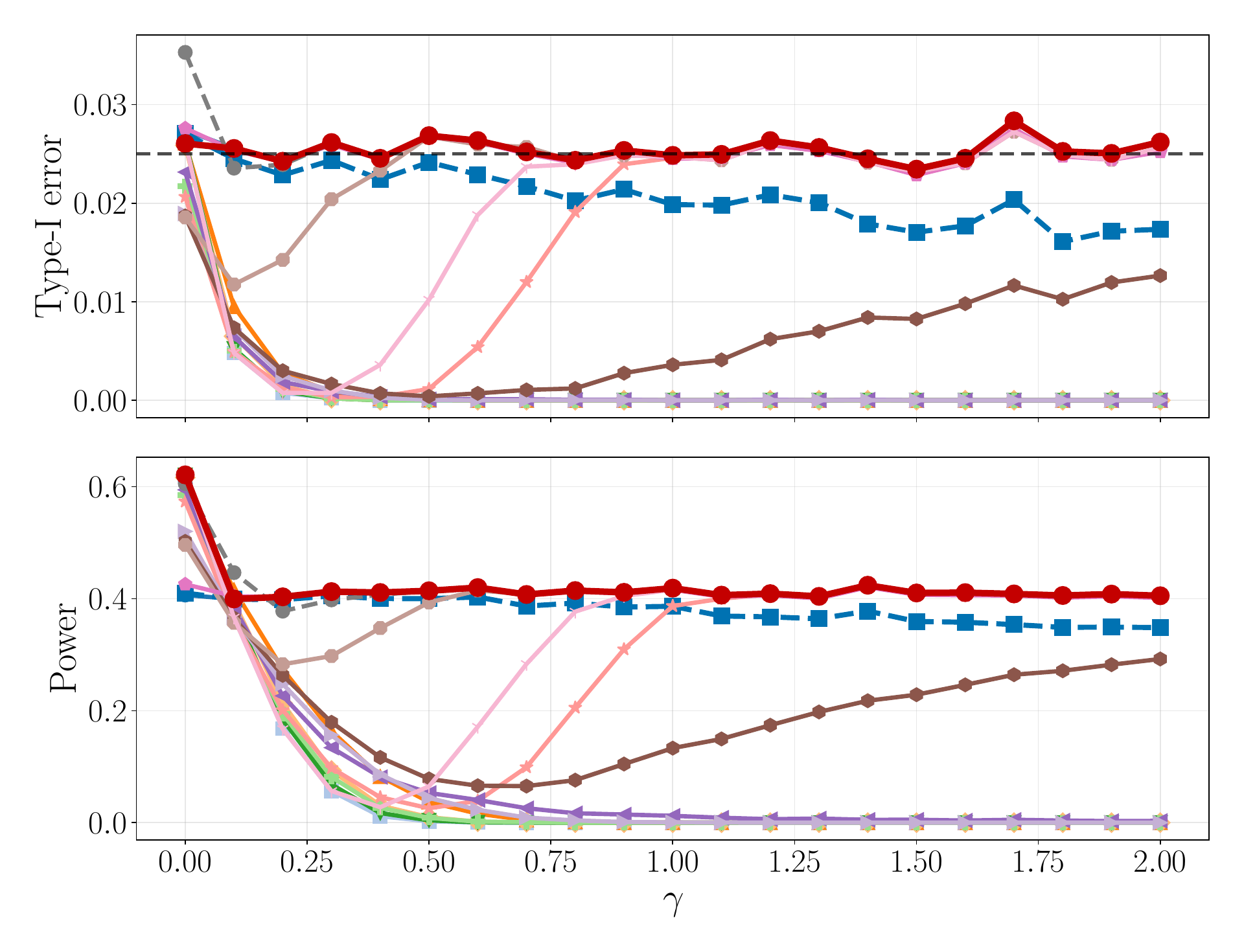}
    \caption{Wasserstein-based radii ($c=1.5$).}
    \end{subfigure}
\caption{
Type I error (top) and power (bottom) versus $\gamma$ for continuous outcomes under Control drift (historical control-only) with $n_C=200$ and $n_H=500$.
The horizontal reference line is at $\alpha=0.025$.
}
\label{fig:app-type1-power-continuous-S3}
\end{figure}

\begin{figure}[tb]
\centering
    \includegraphics[width=\columnwidth]{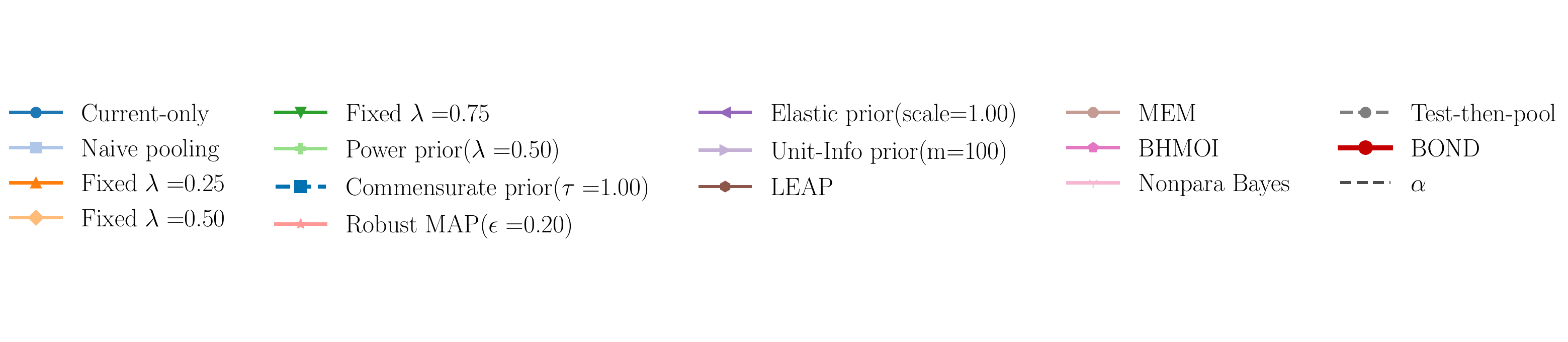}
    \begin{subfigure}[t]{0.49\linewidth}
    \centering
    \includegraphics[width=\linewidth]{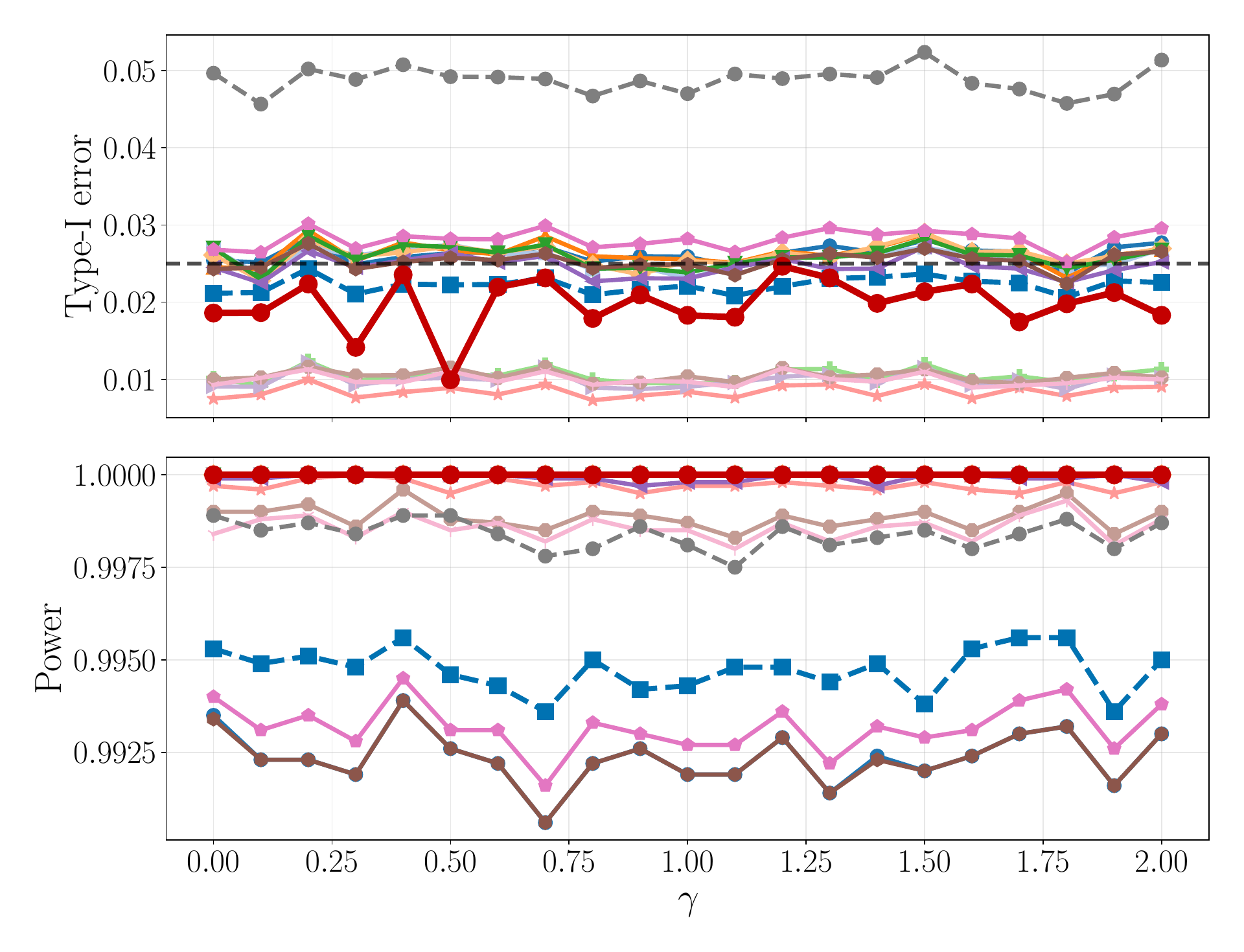}
    \caption{Oracle radii.}
    \end{subfigure}\hfill
    \begin{subfigure}[t]{0.49\linewidth}
    \centering
    \includegraphics[width=\linewidth]{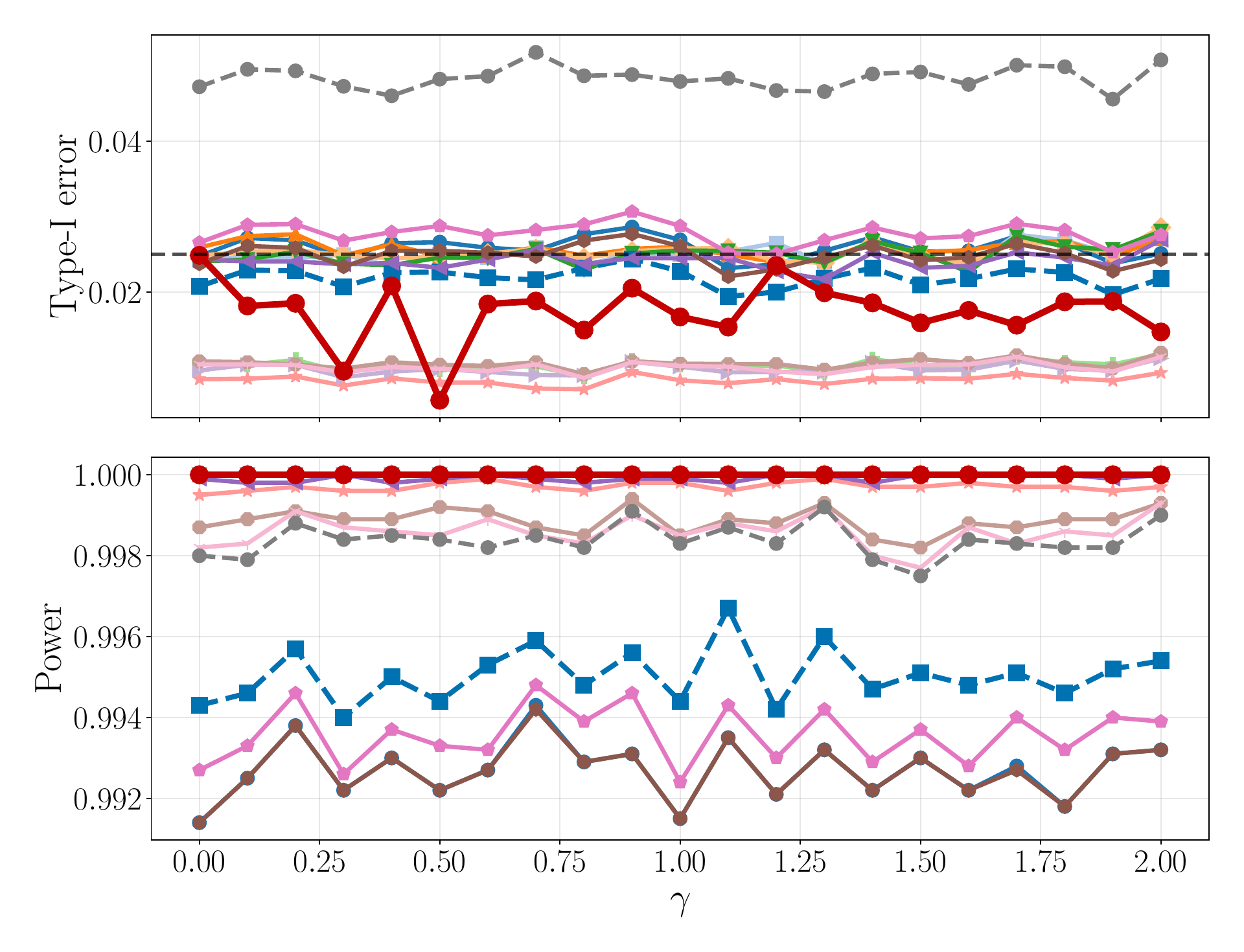}
    \caption{Wasserstein-based radii ($c=1.5$).}
    \end{subfigure}
\caption{
Type I error (top) and power (bottom) versus $\gamma$ for binary outcomes under Commensurate with $n_C=200$ and $n_H=500$.
The horizontal reference line is at $\alpha=0.025$.
}
\label{fig:app-type1-power-binary-S0}
\end{figure}

\begin{figure}[tb]
\centering
    \includegraphics[width=\columnwidth]{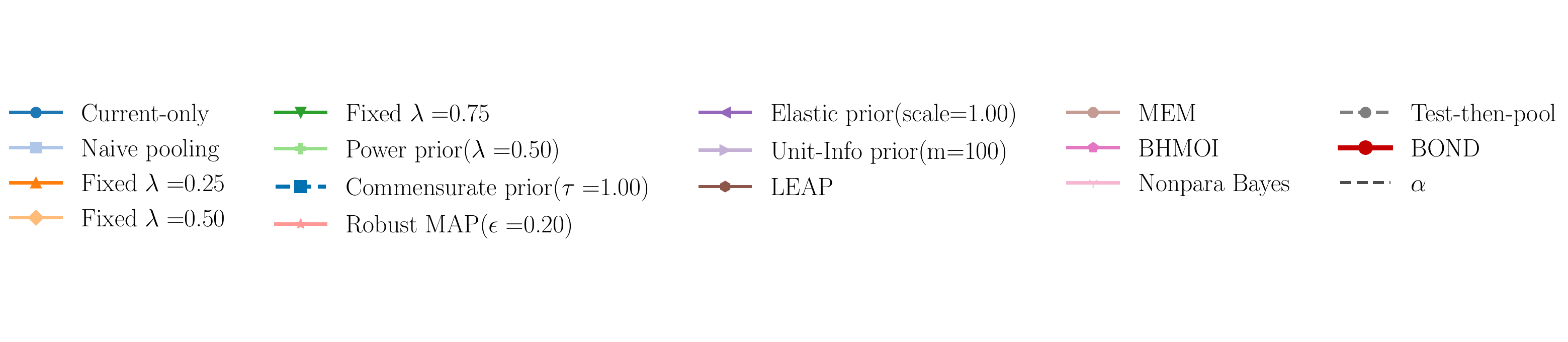}
    \begin{subfigure}[t]{0.49\linewidth}
    \centering
    \includegraphics[width=\linewidth]{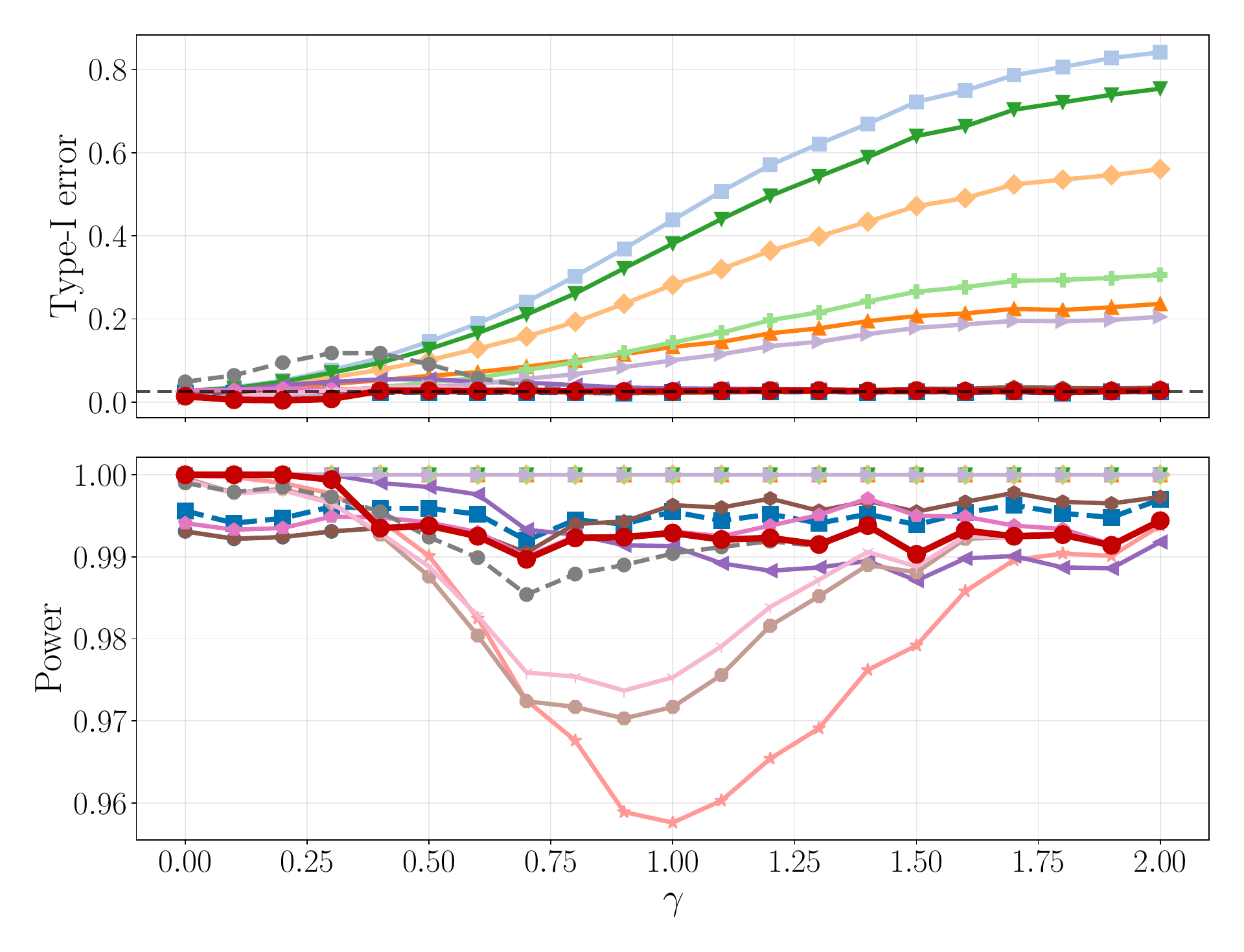}
    \caption{Oracle radii.}
    \end{subfigure}\hfill
    \begin{subfigure}[t]{0.49\linewidth}
    \centering
    \includegraphics[width=\linewidth]{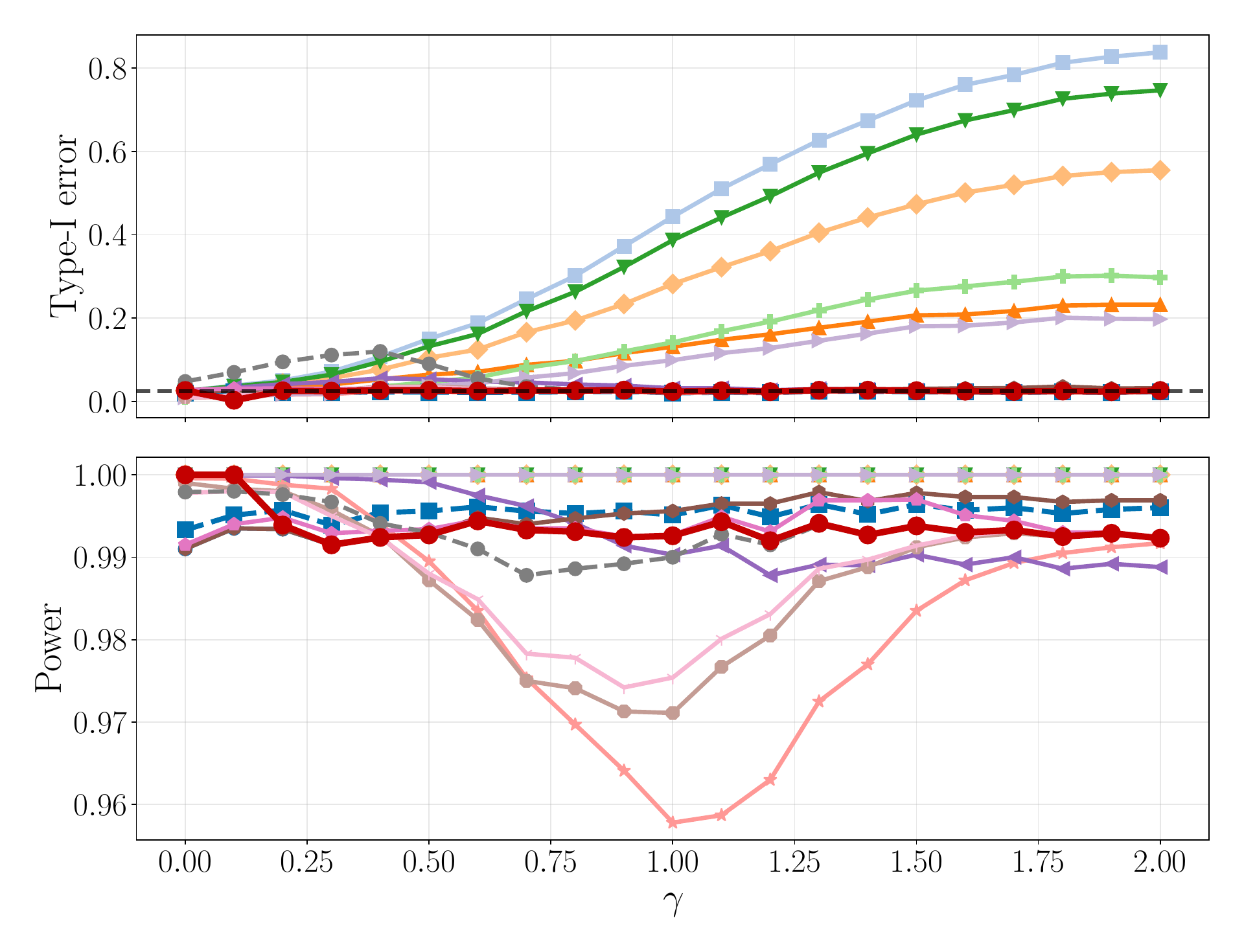}
    \caption{Wasserstein-based radii ($c=1.5$).}
    \end{subfigure}
\caption{
Type I error (top) and power (bottom) versus $\gamma$ for binary outcomes under Covariate shift + effect modification with $n_C=200$ and $n_H=500$.
The horizontal reference line is at $\alpha=0.025$.
}
\label{fig:app-type1-power-binary-S2}
\end{figure}

\begin{figure}[tb]
\centering
    \includegraphics[width=\columnwidth]{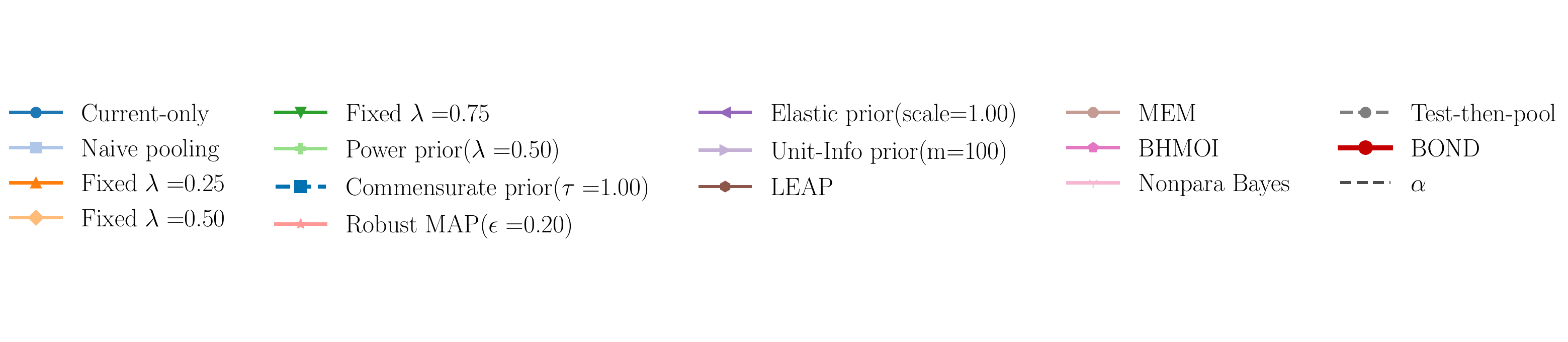}
    \begin{subfigure}[t]{0.49\linewidth}
    \centering
    \includegraphics[width=\linewidth]{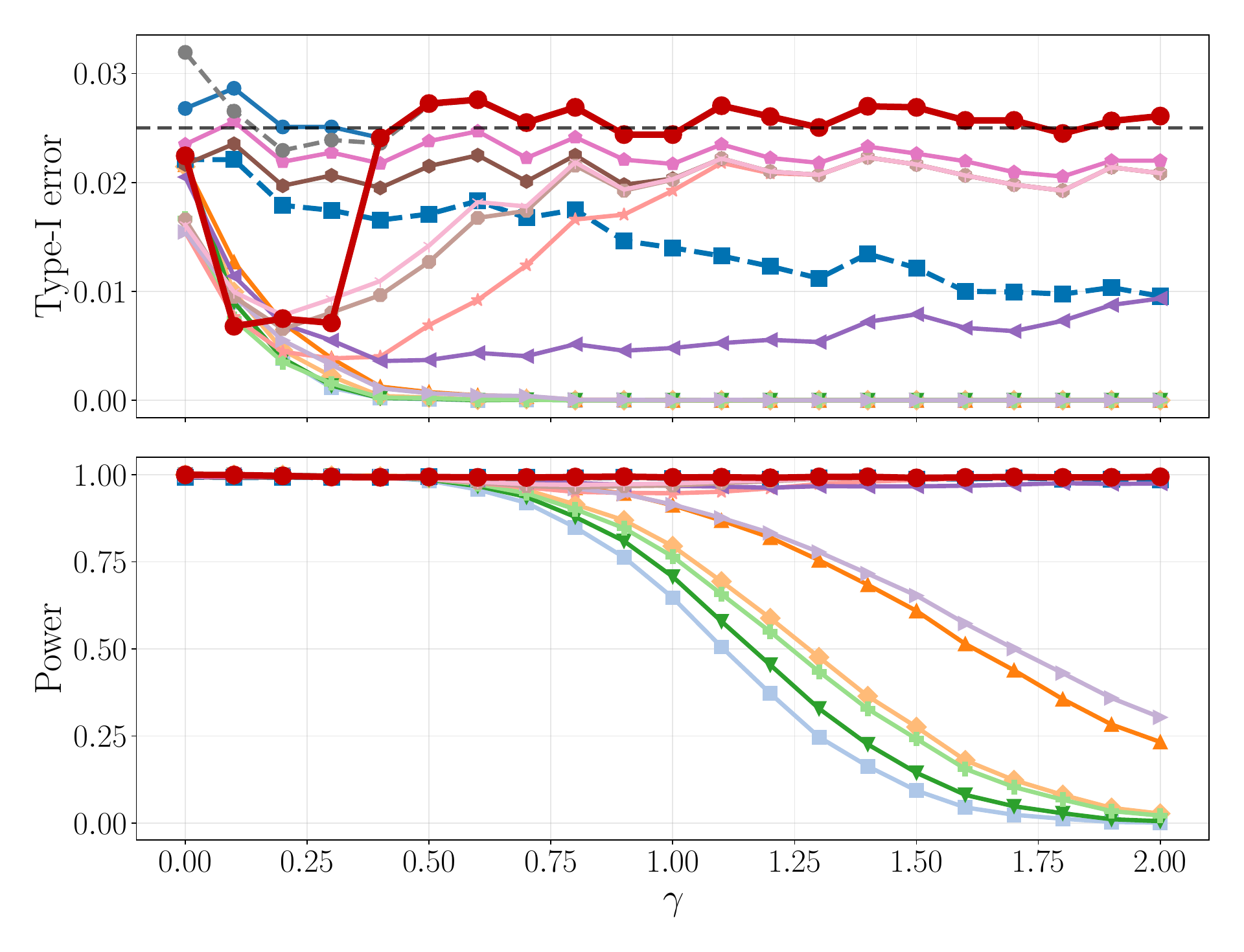}
    \caption{Oracle radii.}
    \end{subfigure}\hfill
    \begin{subfigure}[t]{0.49\linewidth}
    \centering
    \includegraphics[width=\linewidth]{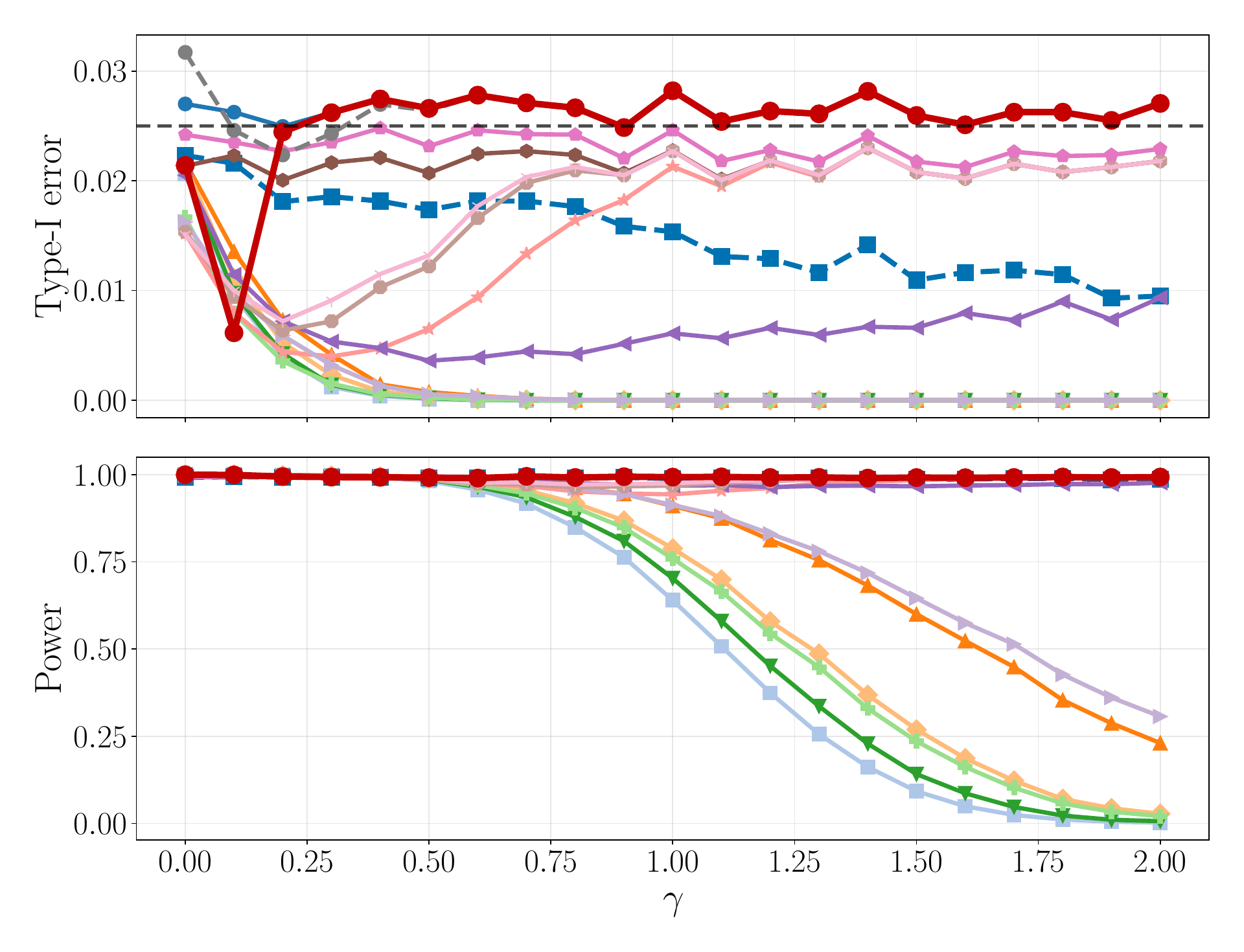}
    \caption{Wasserstein-based radii ($c=1.5$).}
    \end{subfigure}
\caption{
Type I error (top) and power (bottom) versus $\gamma$ for binary outcomes under Control drift (historical control-only) with $n_C=200$ and $n_H=500$.
The horizontal reference line is at $\alpha=0.025$.
}
\label{fig:app-type1-power-binary-S3}
\end{figure}
\section{Detailed Real-World Data Experiments}
\label{app:sec:additional-real-world}

This section provides supplementary implementation details and extended numerical results for the real-world application in mCRC presented in Section~\ref{sec:real-world-experiments}.

\subsection{Implementation and Sensitivity Considerations}
\label{app:subsec:rw-implementation-considerations}

\subsubsection{Sensitivity to the Wasserstein Radius and Data-Driven Proxies}
\label{app:subsubsec:rw-wasserstein-radius}

In practice, the Wasserstein radius $\rho_0$ explicitly quantifies the investigator's tolerance for unmeasured noncommensurability.
We treat $\rho_0$ as a primary sensitivity parameter, evaluating the robust estimates and calibrated weights across a clinically meaningful grid of maximal tolerable drift values.
When individual-level data are unavailable, a simple empirical proxy for binary outcomes can be constructed, for example, as $\hat{\rho}_0(c) = c|\bar{Y}_{H,0} - \bar{Y}_{C,0}|$.
The multiplier $c \ge 1$ allows for a conservative inflation to account for residual heterogeneity that might not be fully captured by the marginal ORR discrepancy.

\subsubsection{Sensitivity to the Target Alternative}
\label{app:subsubsec:rw-target-alternative}

The DRO calibration objective in \eqref{eq:kappa-hat} relies on a prespecified target alternative $\theta_1>0$.
In clinical settings, $\theta_1$ aligns with a minimum clinically important difference on the absolute ORR scale.
While our primary analysis fixes $\theta_1$, varying this parameter reveals a general trade-off:
larger values of $\theta_1$ tend to favor more aggressive borrowing (prioritizing variance reduction to capture the large effect), whereas smaller values yield more conservative borrowing behavior to strictly protect against bias.

\subsubsection{Two-Sided Inference}
\label{app:subsubsec:rw-two-sided}

Although evaluating ORR improvement is inherently a one-sided hypothesis, two-sided robust confidence intervals are required for comprehensive reporting.
Applying the framework from Appendix~\ref{app:sec:two-sided} to the control-only borrowing setting, the two-sided bias corrections simplify cleanly to:
\begin{equation*}
    \tilde{b}_+(\lambda_0)
    =
    w_0(\lambda_0)\min\{\rho_0,\bar{Y}_{C,0}\},
    \quad
    \tilde{b}_-(\lambda_0)
    =
    -w_0(\lambda_0)\min\{\rho_0,1-\bar{Y}_{C,0}\}
    .
\end{equation*}
These boundaries map directly into the robust interval defined in \eqref{app:eq:robust-ci-two-sided}.
Crucially, this interval dynamically widens as either the EBW $w_0(\lambda_0)$ or the tolerance radius $\rho_0$ increases, reflecting the epistemic uncertainty.

\subsection{Extended Real-World Results}
\label{app:subsec:appendix-realworld-results}

Table~\ref{tab:appendix-realworld-all-methods} provides the complete set of results for all evaluated baseline methods on the mCRC dataset.
The behavior perfectly aligns with the simulation findings:
dynamic borrowing priors that incorporate explicit conflict adaptation (e.g., TTP, robust MAP, MEM, BHMOI, Nonparametric Bayes) detect the substantial empirical discrepancy ($\bar{Y}_{H,0}=0.367$ vs. $\bar{Y}_{C,0}=0.128$) and correctly isolate the current trial, yielding estimates virtually identical to the current-only analysis ($\hat{\theta} \approx 0.155$).
Conversely, methods imposing fixed borrowing structures (e.g., Naive pooling, fixed $\lambda=0.75$) incur substantial attenuation of the estimated treatment effect, often failing to reject the null hypothesis at the $0.05$ one-sided level.

Table~\ref{tab:appendix-realworld-bond-sensitivity} details the specific sensitivity of BOND to the robustness radius $\rho$ (here representing $\rho_0$).
As the tolerance for bias ($\rho$) increases from $0$ to $0.05$, the optimization seamlessly reduces the effective historical weight $\lambda_0^{\mathrm{eff}}$ from $0.482$ to $0.020$, dropping the effective borrowed sample size $n_{\mathrm{hist}}^{\mathrm{eff}}$ from $294$ to just $12$.
The estimated effect correspondingly recovers from an attenuated $\hat{\theta}=0.065$ back to $\hat{\theta}=0.150$.
For $\rho \ge 0.10$, BOND recognizes the potential bias is too severe, assigns zero weight to the historical controls, and coincides perfectly with the current-only analysis.

Figures~\ref{fig:appendix-realworld-bond-weights} and~\ref{fig:appendix-realworld-bond-theta} illustrate this continuous transition.
They demonstrate how the calibrated borrowing levels and the robust effect estimate $\hat{\theta}$ smoothly recover toward the unconfounded current-only analysis as the procedure curtails borrowing to satisfy strict type I error constraints.
Because no historical treatment arm is available in this dataset, parameters $\lambda_1^*$ and $w_1$ remain identically $0$.

\begin{figure*}[tb]
\centering
\includegraphics[width=1.00\textwidth]{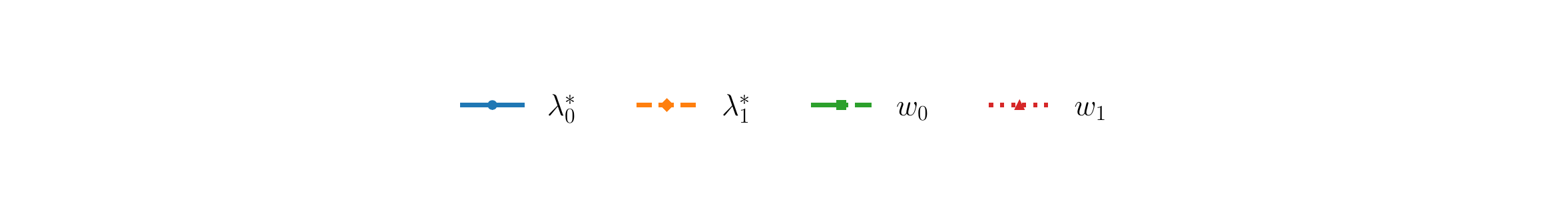}
\includegraphics[width=0.75\textwidth]{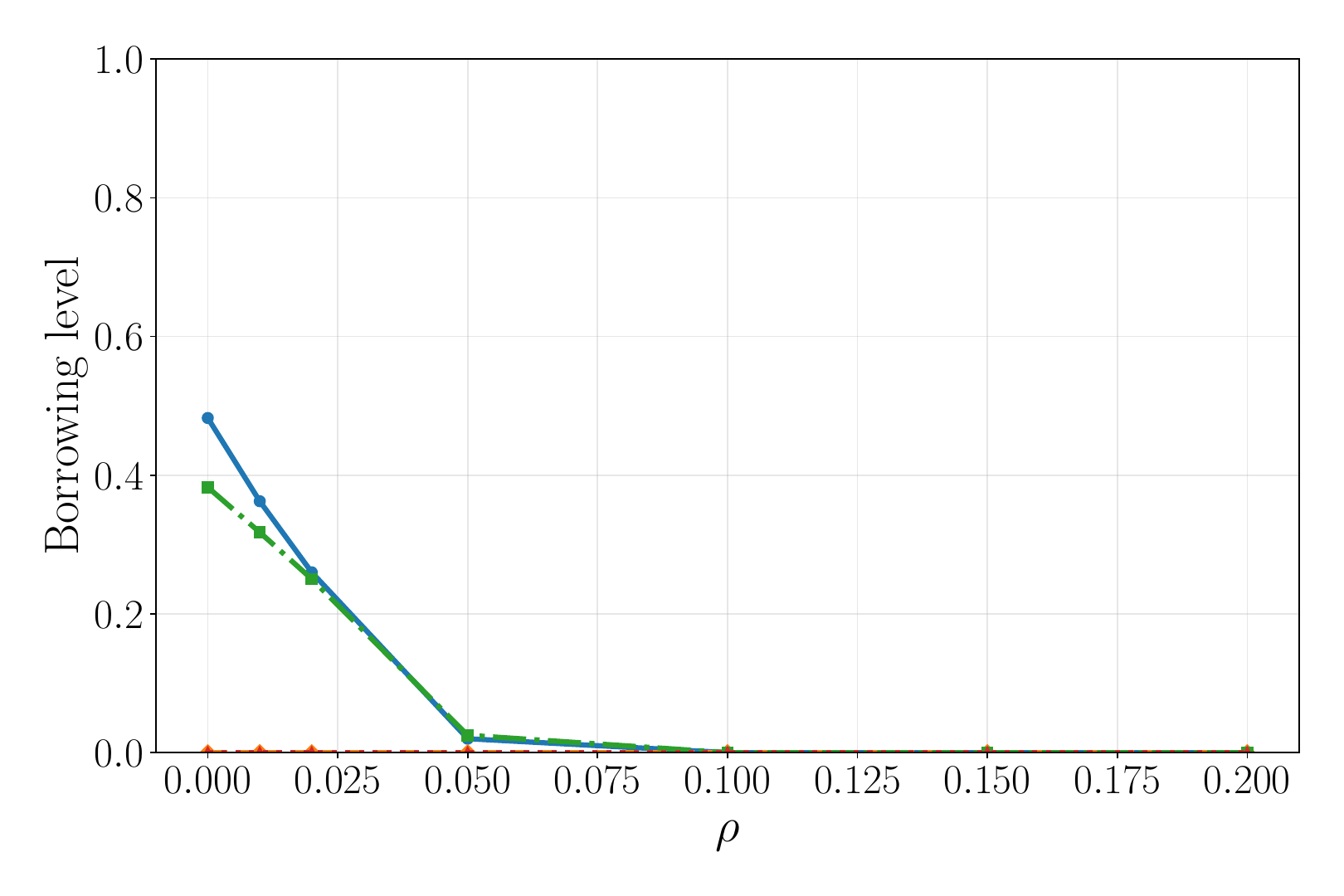}
\caption{
BOND borrowing levels versus robustness radius $\rho$.
We plot the optimal robustness-calibrated weights $\lambda_0^*$ (and $\lambda_1^*$) together with the induced effective borrowing levels $w_0$ (and $w_1$).
In this dataset, historical treatment is unavailable, hence $\lambda_1^*$ and $w_1$ remain at zero and the sensitivity is driven solely by control-side borrowing.}
\label{fig:appendix-realworld-bond-weights}
\end{figure*}

\begin{figure*}[tb]
\centering
\includegraphics[width=\textwidth]{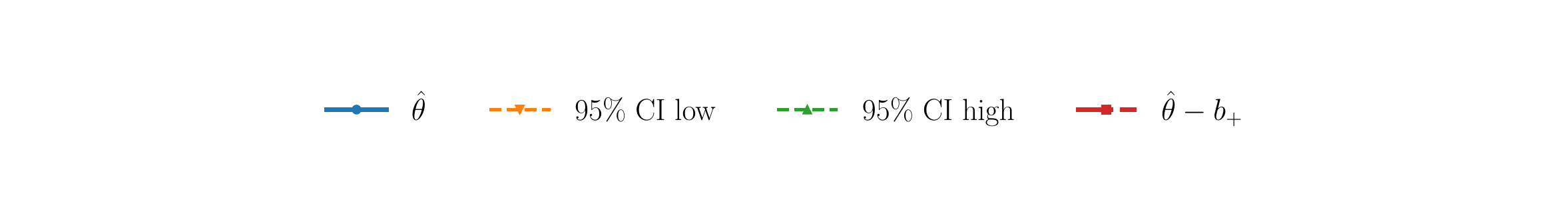}
\includegraphics[width=0.75\textwidth]{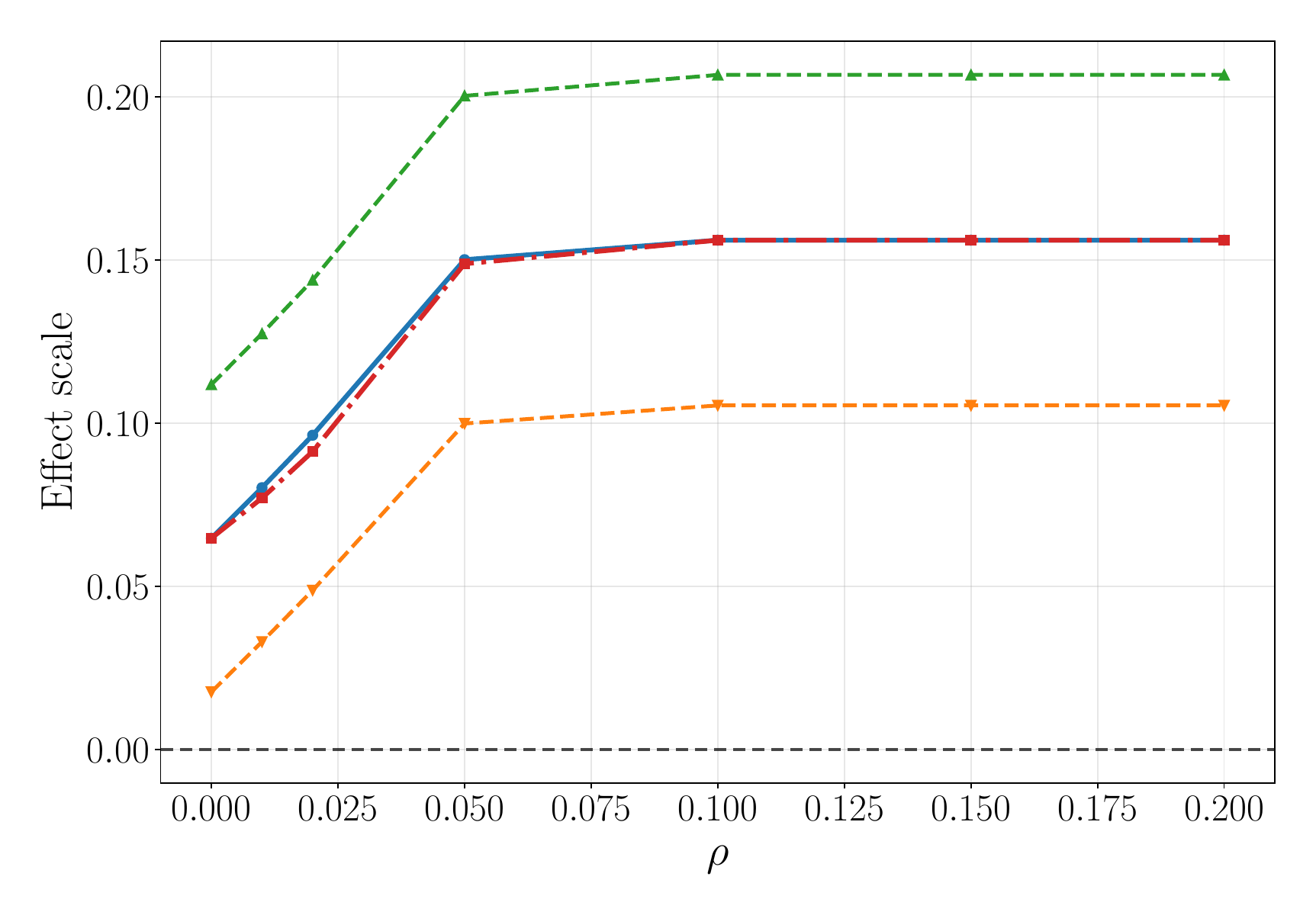}
\caption{
BOND effect estimate sensitivity versus $\rho$.
We show the standard estimate $\hat{\theta}$ with its 95\% CI, together with the robust bias-adjusted curve $\hat{\theta}-b_+$.
Increasing $\rho$ reduces borrowing and pulls the estimate back toward the current-only analysis.}
\label{fig:appendix-realworld-bond-theta}
\end{figure*}

\begin{table*}[tb]
\centering
\caption{Real-world ORR analysis:
full results for all methods.}
\label{tab:appendix-realworld-all-methods}
\setlength{\tabcolsep}{6pt}
\begin{tabular}{lrrrrr}
\toprule
Method & $\hat{\mu}_0$ & $\hat{\theta}$ & Width ratio & $n_{\mathrm{hist}}^{\mathrm{eff}}$ & $p$ \\
\midrule
Current-only
& $0.128$ & $0.156$ & $1.000$ & $0$ & $7.7\times 10^{-10}$ \\

Naive pooling
& $0.263$ & $0.022$ & $0.946$ & $610$ & $0.186$ \\

TTP
& $0.128$ & $0.156$ & $1.000$ & $0$ & $7.7\times 10^{-10}$ \\

Fixed $\lambda=0.25$
& $0.186$ & $0.098$ & $0.940$ & $152$ & $2.7\times 10^{-5}$ \\

Fixed $\lambda=0.5$
& $0.222$ & $0.063$ & $0.930$ & $305$ & $0.005$ \\

Fixed $\lambda=0.75$
& $0.246$ & $0.039$ & $0.936$ & $458$ & $0.054$ \\

Power prior ($\lambda=0.5$)
& $0.223$ & $0.062$ & $0.989$ & $305$ & $0.008$ \\

UIP ($M=100$)
& $0.170$ & $0.115$ & $1.007$ & $99$ & $5.0\times 10^{-6}$ \\

Elastic prior (scale=1)
& $0.132$ & $0.152$ & $1.001$ & $8$ & $2.0\times 10^{-9}$ \\

Robust MAP ($\epsilon=0.2$)
& $0.130$ & $0.155$ & $1.001$ & $0$ & $1.1\times 10^{-9}$ \\

Commensurate prior ($\tau=1$)
& $0.132$ & $0.152$ & $1.004$ & $4$ & $2.1\times 10^{-9}$ \\

MEM
& $0.130$ & $0.155$ & $1.001$ & $0$ & $1.1\times 10^{-9}$ \\

BHMOI
& $0.129$ & $0.155$ & $1.001$ & $1$ & $1.0\times 10^{-9}$ \\

Nonparametric Bayes
& $0.130$ & $0.155$ & $1.001$ & $0$ & $1.1\times 10^{-9}$ \\

LEAP
& $0.130$ & $0.154$ & $1.001$ & $313$ & $1.3\times 10^{-9}$ \\

BOND
& $0.220$ & $0.065$ & $0.930$ & $294$ & $0.004$ \\
\bottomrule
\end{tabular}
\end{table*}

\begin{table}[tb]
\centering
\caption{
BOND sensitivity to the robustness radius $\rho$.
We report the effective historical weight $\lambda_0^{\mathrm{eff}}$, effective borrowed historical sample size
$n_{\mathrm{hist}}^{\mathrm{eff}}$, the estimated control response $\hat{\mu}_0$, treatment effect $\hat{\theta}$,
relative interval width, and the robust one-sided $p$-value.}
\label{tab:appendix-realworld-bond-sensitivity}
\setlength{\tabcolsep}{5pt}
\begin{tabular}{rrrrrrr}
\toprule
$\rho$ & $\lambda_0^{\mathrm{eff}}$ & $n_{\mathrm{hist}}^{\mathrm{eff}}$ & $\hat{\mu}_0$ & $\hat{\theta}$ & Width ratio & $p_{\mathrm{rob}}$ \\
\midrule
$0$    & $0.482$ & $294$ & $0.220$ & $0.065$ & $0.930$ & $0.004$ \\
$0.01$ & $0.362$ & $221$ & $0.204$ & $0.080$ & $0.932$ & $0.001$ \\
$0.02$ & $0.260$ & $159$ & $0.188$ & $0.096$ & $0.939$ & $8.4\times 10^{-5}$ \\
$0.05$ & $0.020$ & $12$  & $0.134$ & $0.150$ & $0.991$ & $3.1\times 10^{-9}$ \\
$0.1$  & $0.000$ & $0$   & $0.128$ & $0.156$ & $1.000$ & $7.7\times 10^{-10}$ \\
$0.15$ & $0.000$ & $0$   & $0.128$ & $0.156$ & $1.000$ & $7.7\times 10^{-10}$ \\
$0.2$  & $0.000$ & $0$   & $0.128$ & $0.156$ & $1.000$ & $7.7\times 10^{-10}$ \\
\bottomrule
\end{tabular}
\end{table}

\end{document}